\documentclass[a4paper,10pt,oneside,openany]{article}

\usepackage[colorlinks]{hyperref}
\usepackage{xspace}

\usepackage{exscale}%
\usepackage{amsmath}%
\usepackage{amsfonts}%
\usepackage{times}
\usepackage{amsthm}
\usepackage{enumerate}
\usepackage{enumitem}
\newlist{Enumerate}{enumerate}{2} 
\setlist[Enumerate,1]{
  label=\upshape(\roman*\upshape),
  ref = \upshape(\roman*\upshape),
  leftmargin=*,
  widest=99,
  align=left,
  itemsep = 2pt,
  labelindent= 6pt,
  resume
}
\setlist[Enumerate,2]{
  label=\upshape(\alph*\upshape),
  ref = \upshape(\alph*\upshape),
  leftmargin=*,
  widest=99,
  align=left,
  itemsep = 2pt
}
\newcommand{\vphi}{{\varphi}}           
 
\newcommand{\la}{\langle} \newcommand{\ra}{\rangle}
\newcommand{\cB}{\mathcal{B}}
\newcommand{\cC}{\mathcal{C}}
\newcommand{\cD}{\mathcal{D}}
\newcommand{\cE}{\mathcal{E}}
\newcommand{\cF}{\mathcal{F}}
\newcommand{\cG}{\mathcal{G}}
\newcommand{\cH}{\mathcal{H}}

\newcommand{\cK}{\mathcal{K}}
\newcommand{\cM}{\mathcal{M}}         
\newcommand{\cP}{\mathcal{P}}         
\newcommand{\cQ}{\mathcal{Q}}
\newcommand{\cR}{\mathcal{R}}

\newcommand{\cU}{\mathcal{U}}
\newcommand{\cV}{\mathcal{V}}

\newcommand{\RR}{\mathbb{R}}            
\newcommand{\NN}{\mathbb{N}}            
\newcommand{\CC}{\mathbb{C}}            

\newcommand{\ZZ}{\mathbb{Z}}

\newcommand{\bu}{{\bar{u}}}
\newcommand{\bv}{{\bar{v}}}

\newcommand{\bz}{{\bar{z}}}

\newcommand{\bP}{{\overline P}}        

\newcommand{\bT}{{\overline T}}

\newcommand{\bJ}{{\overline J}}

\newcommand{\udelta}{{\underline{\delta}}}

\newcommand{\hL}{\widehat{L}}

\newcommand{\hH}{\widehat{H}}            

\newcommand{\hg}{\hat{g}}

\newcommand{\tA}{\widetilde{A}}
\newcommand{\tC}{\widetilde{C}}

\newcommand{\tE}{\widetilde{E}}
\newcommand{\tF}{\widetilde{F}}
\newcommand{\tG}{\widetilde{G}}
\newcommand{\tK}{\widetilde{K}}
\newcommand{\tN}{{\widetilde{N}}}
\newcommand{\tH}{\widetilde{H}}
\newcommand{\tL}{\widetilde{L}}

\newcommand{\tW}{\widetilde{W}}
\newcommand{\ta}{\tilde{a}}             
\newcommand{\tf}{\tilde{f}}
\newcommand{\tg}{\tilde{g}}
\newcommand{\tih}{\tilde{h}}            

\newcommand{\tcH}{\widetilde{\cH}}
\newcommand{\tcF}{{\widetilde{\cF}}}
\newcommand{\tcC}{\widetilde{\cC}}

\newcommand{\tphi}{\tilde{\phi}}        

\newcommand{\tpsi}{\tilde{\psi}}

\newcommand{\tvarphi}{\tilde{\varphi}}

\newcommand{\trho}{\widetilde{\rho}}

\newcommand{\tgothM}{\widetilde{\mathfrak{M}}}
\newcommand{\tgothW}{\widetilde{\mathfrak{W}}}

\newcommand{\gothh}{\mathfrak{h}}
\newcommand{\gothM}{\mathfrak{M}}
\newcommand{\gothW}{\mathfrak{W}}

\newcommand{\ri}{\mathrm{i}}

\newcommand{\Ran}{\mathop{\mathrm{Ran}}}     
\newcommand{\huelle}{\mathop{\mathrm{span}}} 

\newcommand{\Tr}{\mathop{\mathrm{Tr}}}
\newcommand{\diag}{\mathop{\mathrm{diag}}}

\newcommand{\supp}{\mathop{\mathrm{supp}}}

\newcommand{\dprime}{{\prime\prime}}

\newcommand{\loc}{\mathrm{loc}}

\newcommand{\ad}{\mathrm{ad}}

\renewcommand{\thesection}
{\arabic{section}}                     
\renewcommand{\theequation}
{\thesection.\arabic{equation}}        
\numberwithin{equation}{section}

\newtheorem{theorem}{Theorem}[section]
\newtheorem{lemma}[theorem]{Lemma}
\newtheorem{corollary}[theorem]{Corollary}
\newtheorem{proposition}[theorem]{Proposition}

\theoremstyle{definition}
\newtheorem{definition}[theorem]{Definition}

\theoremstyle{remark}
\newtheorem{remark}[theorem]{Remark}
\newtheorem{remarks}[theorem]{Remarks}

\newtheorem{problem}{Problem}[section]

\newcommand{\slim}{\mathrm{s}-\lim}
\newcommand{\wlim}{\mathrm{w}-\lim}

\newcommand{\pp}{\mathrm{pp}}
\newcommand{\sico}{\mathrm{sc}}

\newcommand{\im}{\mathrm{Im}}
\newcommand{\re}{\mathrm{Re}}

\def\bbbone{{\mathchoice {\mathrm{1\mskip-4mu l}} {\mathrm{1\mskip-4mu l}}
{\mathrm{1\mskip-4.5mu l}} {\mathrm{1\mskip-5mu l}}}}
\newcommand{\one}{\bbbone}

\newcommand{\Mat}{\mathrm{M}}

\newcommand{\e}{\mathrm{e}}
\newcommand{\D}{\mathrm{d}}

\newcommand{\set}[2]{\{#1\, | \, #2\} }
\newcommand{\bigset}[2]{\bigl\{#1\, \big| \, #2\bigr\} }

\newcommand{\ph}{\mathrm{ph}}

\newcommand{\cGamma}{\check{\Gamma}}
\newcommand{\romL}{\mathrm{L}}
\newcommand{\Mo}{\mathrm{Mo}}

\newcommand{\tgothh}{{\tilde{\gothh}}}
\newcommand{\roml}{\mathrm{l}}
\newcommand{\romr}{\mathrm{r}}
\newcommand{\roms}{\mathrm{s}}
\newcommand{\romp}{\mathrm{p}}

\newcommand{\rommax}{\mathrm{max}}
\newcommand{\rommin}{\mathrm{min}}
\newcommand{\conj}{\mathbf{c}}
\newcommand{\Conj}{\mathbf{C}}
\newcommand{\coup}{G}
\newcommand{\tcoup}{\widetilde{\coup}}
\newcommand{\hcoup}{\widehat{\coup}}
\newcommand{\vacuum}{\vert 0\ra}
\newcommand{\tvacuum}{\la 0\vert}
\newcommand{\dvacuum}{\vert 0\otimes 0\ra}
\newcommand{\tdvacuum}{\la 0\otimes 0\vert}
\newcommand{\JakPil}{Jak\v{s}i{\'c}-Pillet\xspace}
\newcommand{\HGCond}[1]{\textup{\textbf{(H\coup#1)}}} 
\newcommand{\LGCond}[1]{\textup{\textbf{(L\coup#1)}}} 
\newcommand{\LGCondp}[1]{\textup{\textbf{(L\coup#1')}}} 
\newcommand{\MCond}[1]{\textup{\textbf{(M#1)}}}
\newcommand{\Div}{\mathop{\mathrm{div}}}
\newcommand{\fin}{\mathrm{fin}}
\newcommand{\sym}{{\mathrm{sym}}}

\newcommand{\hcH}{\widehat{\cH}}

\newcommand{\ext}{\mathrm{x}}

\newcommand{\remarkQED}{$\hfill\lozenge$}
\newcommand{\AW}{\mathrm{AW}}
\newcommand{\PF}{\mathrm{PF}}
\newcommand{\shuf}{\mathbf{s}}

\newcommand{\myquote}[1]{``#1"}

\newcommand{\coupbl}{\coup_{\beta,\roml}}
\newcommand{\coupbr}{\coup_{\beta,\romr}}

\usepackage{enumitem}
\setlist{itemsep=0pt}

\begin{document}

\bibliographystyle{amsplain}

\setcounter{page}{0}

\title{Fully Coupled Pauli-Fierz Systems at Zero and Positive Temperature\\  --------------- \\ 
Summer school on Non-equilibrium Statistical Mechanics \\ ---------------\\  Montreal 2011 \\ Seventh Version\footnote{{Apart from typos, recent changes are as follows. Seventh version: Appendix~\ref{App-GeomLoc} added, the use of geometric localization fleshed out and inconsistent use of 
$\delta_0$
corrected (both in Subsect.~\ref{Subsec-EstZero}). Sixth version: Appendix~\ref{App-CommCalc} made more self-contained. Fifth version: Some nonsensical resolvent bounds in Subsect.~\ref{Subsec-AprioriRes} corrected. Fourth version: Section~\ref{Sec-LAP} with the Limiting Absorption Principle added.}}
}

\author{Jacob Schach M\o ller\\
Department of Mathematics\\
Aarhus University\\
Denmark}

\date{\today}

\maketitle

\thispagestyle{empty}

\newpage

\tableofcontents

\thispagestyle{empty}

\newpage

\section{Introduction}

The purpose of these notes is to give a fairly narrow but thorough
introduction to the spectral analysis of Hamiltonians and standard 
Liouvilleans describing finite dimensional small systems linearly coupled
to a scalar massless field or reservoir. The Hamiltonians describe the
system at zero temperature, and the standard Liouvillean implements 
unitarily the dynamics of the system at positive temperature.

We focus our attention on results valid at arbitrary 
coupling strength and whose proofs are purely operator theoretic,
i.e. for the standard Liouvillean does not make use of the
underlying modular structure. For the standard Liouvillean this means that 
important structure results that does not seem to have a purely
operator theoretic proof will only be reviewed. 

 In the rest of the introduction we will assume some familiarity 
with quantum dynamical systems, standard representations of von Neumann algebras, and
in particular the bosonic Weyl-algebra and its positive temperature
standard form, given in terms of Araki-Woods fields. 
We refer the reader to the three review papers
\cite{Attal2006,Merkli2006,Pillet2006} written by Attal, Merkli, and Pillet
in connection with an earlier summer school held in Grenoble in 2003. They
combine to give an excellent introduction to the subject.
As soon as the problems we aim for have been cast in spectral terms at
the end of the introduction, these topics will not play a central role anymore.

\subsection{The Small Quantum System}\label{Subsec-smallSystem}

Our small quantum system lives in the finite dimensional Hilbert space
$\cK = \CC^\nu$. The associated observable algebra is taken to be
$\gothM_\romp = \cB(\cK)=\Mat_\nu(\CC)$, the $\nu^2$-dimensional space of $\nu\times \nu$
matrices. The subscript $\romp$ signals that the algebra belongs to
the particle system.

As a Hamiltonian we take a diagonal matrix
\[
K = \diag\{E_1,E_2,\dotsc,E_\nu\},
\]
where the eigenvalues 
$E_1\leq E_1\leq \cdots \leq E_\nu$ are real numbers. 
It suffices for $K$ to be merely self-adjoint, but we can obviously
always choose a basis in which $K$ has the form considered here.

States $\omega$ on the algebra $\gothM_\romp$ are identified with density 
matrices $\rho$, i.e. self-adjoint matrices with $0\leq \rho\leq 1$
and $\Tr(\rho)=1$. The associated state is
\[
\omega(T) = \Tr(\rho T).
\]
A state is called a vector state if $\rho$ is a rank one projection,
i.e. $\rho = \vert u\ra\la u\vert$ and hence $\omega(T) = \la u,T u\ra$.
Here $u$ is a normalized vector.

The Hamiltonian $K$ generates a dynamics on the algebra as well as on
the states $\omega$ as follows: $T\to T_t = \e^{\ri t K} T\e^{-\ri t K}$ and
$\omega\to \omega_t$ defined by $\omega_t(T) = \omega(T_t)$.

The set of states invariant under the dynamics
can be identified with density matrices of the form
\[
\sum \alpha_j P_j,
\]
where $0\leq \alpha_j\leq 1$, $P_j$ are orthogonal projections onto
(subspaces of) eigenspaces of $K$, and $\sum \alpha_j\Tr(P_j) = 1$.

At inverse temperature $\beta$ the so-called Gibbs, hence $\beta$-KMS,
invariant state on $\gothM_\romp$ is
\[
\omega^\romp_\beta(T) = \Tr(\rho_\beta T), \qquad \rho_\beta =
\frac{\e^{-\beta K}}{\Tr(\e^{-\beta K})}.
\]
At zero temperature, or $\beta = +\infty$, this becomes
$\rho_\infty = \nu_0^{-1} P_0$, where $P_0$ is the orthogonal projection
onto the span of ground states, and $\nu_0 = \Tr(P_0)$ is the multiplicity of the
ground state. If the ground state is non-degenerate, $\omega^\romp_\infty$ is a
vector state.

As a last item we wish to discuss the GNS representation
of the algebra $\gothM_\romp$ with respect to the Gibbs state $\omega^\romp_\beta$,
which in fact coincides with its standard representation. 

Typically it is introduced as left multiplication on the algebra
itself, but through the identification $|u\ra\la v|\to u\otimes\bv \in
\cK\otimes\cK$ we pass straight to a representation in terms of
operators on $\cK\otimes\cK$. 

We have the left representation $\pi^\romp_\roml$, which serves as the GNS representation, and
the (conjugate linear) right representation  $\pi^\romp_\romr$ 
defined by
\[
\pi^\romp_\roml(T) = T\otimes\one_\cK \quad \textup{and} \quad
\pi^\romp_\romr(T) = \one_\cK \otimes \bT.
\]
They are intertwined by the conjugate linear modular conjugation 
$J_\romp\colon \cK\otimes \cK\to \cK\otimes\cK$ defined by
\begin{equation}\label{Jp}
J_\romp(u\otimes v) = \bv\otimes \bu,
\end{equation}
and extended by linearity to $\cK\otimes\cK$.
That is, $\pi_\romr^\romp(T) = J_\romp \pi^\romp_\roml(T) J_\romp$. 
The state $\omega_\beta^\romp$ goes into a vector state with respect to
 the GNS vector
\[
\Omega_\beta^\romp = \sum_{j=1}^\nu\frac{\e^{-\beta E_j/2}}{\sqrt{\Tr(\e^{-\beta K})}}\, e_j\otimes e_j.
\]
That is,
\[
\omega_\beta^\romp(T) = \la \Omega_\beta^\romp, \pi^\romp_\roml(T) \Omega_\beta^\romp\ra.
\]

Associated with
the representation is the standard self-dual cone
\[
\cP^\romp_\beta = \bigset{\pi^\romp_\roml(T) \pi^\romp_\romr(T) \Omega_\beta^\romp}{T\in
  \gothM_\romp} = \bigset{T\otimes\bT \,\Omega_\beta^\romp}{T\in\gothM_\romp},
\]
which is invariant under $J_\romp$.

There are a priori many ways to lift the dynamics on the algebra
$\gothM_\romp$ to a unitarily implemented dynamics on 
its image $\pi^\romp_\roml(\gothM_\romp)$ inside
$\cB(\cK\otimes\cK)$. The simplest choice is to use $K\otimes\one_\cK$ as generator.
However, there are two other natural choices available to us, one is the
$\Omega_\beta^\romp$-Liouvillean, fixed by requiring 
\begin{equation}\label{Fixing-Lp1}
L\Omega_\beta^\romp = 0.
\end{equation}
The other being the standard Liouvillean,
selected uniquely by the requirement
\begin{equation}\label{Fixing-Lp2}
 \e^{\ri t L}\cP^\romp_\beta\subseteq\cP^\romp_\beta.
\end{equation}
For faithful states -- as is the case here -- 
the two Liouvilleans coincide and we get:
\begin{equation}\label{Lp}
L_\romp = K\otimes\one_\cK- \one_\cK\otimes K.
\end{equation}
This is the unique choice satisfying either \eqref{Fixing-Lp1} or \eqref{Fixing-Lp2}, and
unitarily implementing the dynamics
\[
\pi^\romp_\roml(\e^{\ri t K}T\e^{-\ri t K}) = \e^{\ri t L_\romp} \pi^\romp_\roml(T) \e^{-\ri t L_\romp}.
\]
The GNS vector $\Omega^\romp_\beta$ is a $\beta$-KMS vector for the dynamics 
$\tau_\romp^t(T) = \e^{\ri t L_\romp} T \e^{-\ri t L_\romp}$ on $\pi^\romp_\roml(\gothM)$.
Note that
\[
\sigma(L_\romp) =\sigma_\pp(L_\romp) = \sigma(K)-\sigma(K) = \bigset{E_i-E_j}{1\leq i,j\leq n},
\]
with the obvious degeneracies. In particular $0\in\sigma(L_\romp)$ is at least $\nu$-fold degenerate.

\subsection{The Reservoir}

We begin by introducing the Fock representation of
Weyl operators on 
the bosonic Fock space $\cF =
\Gamma(\gothh)$, build over the one-particle space $\gothh =
L^2(\RR^3)$. We refer the reader to Appendix~\ref{SecondQuant}
for the notation and the basic constructions pertaining to second quantization that we use here.
Note that
$\cF = \oplus_{n=0}^\infty \cF^{(n)}$, where $\cF^{(n)} =\gothh^{\otimes_\mathrm{s} n}=
 L^2_{\sym}(\RR^{3n})$, square integrable
functions symmetric under interchange of the $n$ variables.
Recall the convention that $\gothh^{\otimes_{\mathrm{s}}0} = \CC$ and that
$\otimes_{\mathrm{s}}n$ denotes the $n$-fold symmetric tensor product.

We recall furthermore from Appendix~\ref{SegalFields} that the Segal fields $\phi(f)$, for $f\in\gothh$, are the self-adjoint operators given
by $\phi(f) = 2^{-1/2}(a^*(f) + a(f))$,
where $a(f)$ and $a^*(f)$ are the bosonic annihilation and creation operators.

The Weyl operators in the Fock representation are given by the expression
\[
W(f) = \e^{\ri \phi(f)},
\]
where $f$ runs over the one-particle space $\gothh$. We define
$\gothW$ to be the complex $*$-algebra generated by the Weyl operators,
pertaining to functions $f$ from.
\begin{equation}\label{RestrictedStates}
 \gothh_0 = \bigset{f\in\gothh}{|k|^{-\frac12}f\in\gothh}.
\end{equation}
Observe that this amounts to taking the linear span of the Weyl operators:
\begin{equation}\label{Weyl-Algebra}
\gothW = \huelle\bigset{W(f)}{f\in\gothh_0}.
\end{equation}
That $\gothW$ coincides with the $*$-algebra generated by the Weyl operators follows from the Weyl
relations
\begin{equation}\label{Weyl-Relations}
W(f)^* = W(-f) \quad \textup{and} \quad W(f)W(g) = \e^{\ri\im \la f,g\ra}W(f+g).
\end{equation}
We remark that one should really pass on to the norm closure to get a $C^*$-algebra,
but since we will eventually pass on to double commutants, i.e. von Neumann envelopes,
we bypass the $C^*$-setting. This simplifies the construction of states and representations below.
We will however still discuss the GNS representation of $\gothW$ as if it was a $C^*$-algebra.
See \cite{Merkli2006,Pillet2006} for how to extend states and representations to
the intermediate $C^*$-algebra setting.

As a dynamics on the Weyl algebra, we take in these notes the second
quantized massless dispersion relation $\RR^3\ni k\to |k|$. That is,
we take as generator the self-adjoint non-negative operator
\[
H_\ph = \D\Gamma(|k|),
\] 
with a single (up to a constant multiple) non-degenerate  bound state $\vacuum = (1,0,0,\dotsc)$,
the vacuum vector. This is in fact the ground state of the reservoir
Hamiltonian. That the Weyl algebra is invariant under the Heisenberg
dynamics follows from the computation
\begin{equation}\label{Evol-Of-WeylOps}
\e^{\ri t H_\ph}W(f) \e^{-\ri t H_\ph} = W(\e^{\ri t|k|} f),
\end{equation}
together with linearity, cf.~\eqref{Weyl-Algebra}.

The associated $\beta$-KMS states are defined on Weyl operators, and
extended by linearity to $\gothW$, by the relation
\begin{equation}\label{WeylInKMS}
\omega_\beta^\cR(W(f)) = \e^{-\|\sqrt{1+2\rho_\beta}f\|^2/4},
\end{equation}
where
\begin{equation}\label{Intro-Planck}
\rho_\beta(k) = \frac1{\e^{\beta |k|}-1}
\end{equation}
is Planck's thermal density for black body radiation. Note that $\rho_\beta(k)\sim (\beta |k|)^{-1}$ at $k=0$, 
which is the reason for considering only Weyl operators for $f\in\gothh_0$, cf.~\eqref{RestrictedStates}.
It follows from the computation \eqref{Evol-Of-WeylOps} that $\omega_\beta^\cR$ are invariant states.
The superscript $\cR$ indicates that the state acts on the reservoir.
The zero temperature state is the vector state
\begin{equation}\label{WeylInVacuum}
\omega_\infty^\cR(W(f)) =  \tvacuum W(f) \vacuum = \e^{-\|f\|^2/4}.
\end{equation}

To construct the GNS representation of the Weyl algebra, with
respect to the $\beta$-KMS state $\omega_\beta^\cR$,  we introduce so-called left (and right) 
Araki-Woods fields associated with the thermal density $\rho_\beta$ from \eqref{Intro-Planck}. 
The left fields are used to construct the GNS representation, 
whereas the right fields are kept for later use when we identify 
the standard form of the GNS representation. 

The left and right Araki-Woods annihilation and
creation operators are acting in $\cF\otimes\cF$, and defined for $f\in \gothh_0$ by
\begin{equation}\label{Intro-lr-AW-Fields}
\begin{aligned}
& a_{\beta,\roml}^\AW(f) =  a_\roml(\sqrt{1+\rho_\beta}\,f) + a^*_\romr(\sqrt{\rho_\beta}\,\bar{f})\\
& a_{\beta,\romr}^\AW(f) =  a_\romr(\sqrt{1+\rho_\beta}\,\bar{f}) + a^*_\roml(\sqrt{\rho_\beta}\,f).
\end{aligned}
\end{equation}
Here, for $g\in\gothh$,
\begin{equation}\label{Intro-lr-AW-Fields0}
a_\roml(g) = a(g)\otimes\one_\cF, \quad a_\romr(g) = \one_\cF\otimes \,a(g),
\end{equation}
and likewise for their adjoints $a^*_\roml(g)$ and
$a^*_\romr(g)$. One can think of $g\to a_{\roml}^\#(g)$ and $g\to a_\romr^\#(\bar{g})$
as left and right zero-temperature Araki-Woods
annihilation and creation operators. Here $a^\#$ denotes either $a$
or $a^*$. 
The Araki-Woods operators form two (equivalent up to complex conjugation)  
non-Fock representations of the canonical commutation relations.
They give rise to smeared Araki-Woods fields $\phi^\AW_{\beta,\roml}(f)$ and
$\phi^\AW_{\beta,\romr}(f)$ and hence Weyl operators $W^\AW_{\beta,\roml}(f)$ and
$W^\AW_{\beta,\romr}(f)$. We denote by 
\[
\tgothW_{\beta,\roml} = \huelle\bigset{W^\AW_{\beta,\roml}(f)}{f\in\gothh_0}
\quad\textup{and}\quad
\tgothW_{\beta,\romr} = \huelle\bigset{W^\AW_{\beta,\romr}(f)}{f\in\gothh_0},
\]
the left and right Araki-Woods Weyl algebras, as
complex $*$-subalgebras of $\cB(\cF\otimes\cF)$.

The left and right Araki-Woods algebras form representations of the Weyl
algebra by the prescriptions
\[
\pi_{\beta,\roml}^\AW(W(f))  = W^\AW_{\beta,\roml}(f) \quad
\textup{and} 
\quad \pi_{\beta,\romr}^\AW(W(f)) = W^\AW_{\beta,\romr}(f),
\]
and extended by linearity to $\gothW$. The $\beta$-KMS state $\omega_\beta^\cR$, 
cf.~\eqref{WeylInKMS}, goes into
the $\beta$-independent vector state 
\[
\forall W\in\gothW:\quad \omega_\beta^\cR(W) = \tdvacuum \pi^\AW_{\beta,\roml}(W) \dvacuum
\]
associated with the GNS vector $\dvacuum:= \vacuum\otimes\vacuum$.
This follows from the observation that 
\[
\begin{aligned}
W^\AW_{\beta,\roml}(f) & = W(\sqrt{1+\rho_\beta}f)\otimes W(\sqrt{\rho_\beta}\bar{f})\\
 W^\AW_{\beta,\romr}(f) & = W(\sqrt{\rho_\beta}\bar{f})\otimes W(\sqrt{1+\rho_\beta}f)
\end{aligned}
\]
 together with \eqref{WeylInKMS} and \eqref{WeylInVacuum}. Note that
 it follows from the Weyl relation \eqref{Weyl-Relations} that the
 left and right Araki-Woods Weyl operators (at the same inverse
 temperature) commute, i.e. 
 \begin{equation}\label{ForFaith}
 \tgothW_{\beta,\roml} \subseteq
 (\tgothW_{\beta,\romr})' \quad\textup{and}\quad \tgothW_{\beta,\romr} \subseteq
 (\tgothW_{\beta,\roml})'.
 \end{equation}

 The dynamics on the Weyl-algebra in the GNS representation can
again be implemented in several ways by a unitary group generated by a self-adjoint
operator $L_\cR$. The requirement
\[
L_\cR\dvacuum = 0 
\]
selects the unique generator 
\[
L_\cR = H_\ph \otimes\one_\cF - \one_\cF \otimes H_\ph,
\]
which we call the  reservoir $\dvacuum$-Liouvillean. It unitarily implements the
dynamics in the GNS representation
\[
\pi^\AW_{\beta,\roml}(\e^{\ri t H_\ph} A \e^{-\ri t H_\ph}) = \e^{\ri t
  L_\cR}\pi^\AW_{\beta,\roml}(A)\e^{-\ri t L_\cR}.
\]
The GNS vector $\dvacuum$ is a $\beta$-KMS vector for the dynamics on $\tgothW_{\beta,\roml}$.
Note that $\sigma(L_\cR) = \RR$ is purely absolutely continuous, except for a
non-degenerate eigenvalue at $0$, with the $\beta$-KMS vector $\dvacuum$ as eigenvector.

When we later start to perturb the free Pauli-Fierz dynamics, we will
be forced to leave the complex algebra -- as well as the $C^*$-algebra -- setting and pass on to the
enveloping Araki-Woods von Neumann algebras
\[
\gothW_{\beta,\roml} = \bigl(\tgothW_{\beta,\roml}\bigr)''
\quad\textup{and}\quad 
\gothW_{\beta,\romr} = \bigl(\tgothW_{\beta,\romr}\bigr)''.
\]
The GNS vector $\dvacuum$ becomes a cyclic faithful vector
for the left algebra $\gothW_{\beta,\roml}$, and hence the algebra
is in its standard representation. That the extension is faithful follows from 
$\dvacuum$ being cyclic for the commutant, cf.~\eqref{ForFaith}.

The modular conjugation on the reservoir $J_\cR$ is simply given by
\[
J_\cR(\varphi\otimes\psi) = \Gamma(\conj)\psi\otimes\Gamma(\conj)\varphi,
\]
and extended to $\cF\otimes\cF$ by linearity and continuity. 
Here $\Gamma(\conj)$ is the lifting of complex  conjugation $\conj(f) = \bar{f}$
on $\gothh$ to $\cF$ using Segal's second quantization functor
$\Gamma$. The modular conjugation intertwines the left and right representations
\[
\forall A\in\gothW:\qquad \pi^\AW_{\beta,\romr}(A) = J_\cR \pi^\AW_{\beta,\roml}(A) J_\cR
\]
and we can identify the right algebra with the commutant of 
the left algebra and vice versa: 
\[
\bigl(\gothW_{\beta,\roml}\bigr)' = J_\cR\gothW_{\beta,\roml} J_\cR = \gothW_{\beta,\romr}.
\]

The dynamics $W\to \tau_\cR^t(W) = \e^{\ri t L_\cR} W \e^{-\ri t L_\cR}$  extends by continuity
from $\tgothW_{\beta,\roml}$ to the von Neumann envelope $\gothW_{\beta,\roml}$,
and -- being unitarily implemented -- is a $\sigma$-weakly continuous group of automorphisms.
The GNS vector $\dvacuum$ remains a $\beta$-KMS vector also for the extended dynamics on $\gothW_{\beta,\roml}$.

The extended dynamics on $\gothW_{\beta,\roml}$ is now (also) unitarily implemented
by the standard Liouvillean, which due to faithfulness of the GNS state $\dvacuum$ 
is identical to the $\dvacuum$-Liouvillean $L_\cR$. Recall that the standard Liouvillean is fixed
by the requirement that it keeps the
self-dual ($J_\cR$-invariant) standard cone $\cP^\cR_\beta$ invariant, where
\begin{align*}
\cP^\cR_\beta & = \overline{\bigset{B J_\cR B \dvacuum
  }{B\in\gothW_{\beta,\roml}}}.
\end{align*}

We remark that the reader should think of the invariant state $\dvacuum$ as 
a background thermal photon cloud, with momentum distribution given by 
Planck's law \eqref{Intro-Planck}. The left Fock space accounts for photons, 
and the right Fock space are holes. 
The left creation operators $a^{\AW*}_{\beta,\roml}$ add photons to the
thermal background, while  the left annihilation operators
$a_{\beta,\roml}^\AW$ add holes. The vacuum state -- the thermal background -- is not annihilated
by any Araki-Woods annihilation or creation operator.
In this picture, the standard Liouvillean $L_\cR$ can be interpreted as a Hamiltonian.
Adding electrons cost energy, while adding holes release energy.

\subsection{Non-interacting Pauli-Fierz Systems}

 Tensoring the small quantum system with the reservoir yields
the non-interacting Pauli-Fierz system. The Hilbert space for the full
system is
\[
\cH = \cK\otimes\cF
\]
and the free Hamiltonian is
\[
H_0 = K\otimes \one_\cF + \one_\cK\otimes H_\ph.
\]
As the algebra of observables we take the (algebraic) tensor product
\[
\gothM = \gothM_\romp\otimes\gothW \subseteq \cB(\cH).
\]
The free Heisenberg dynamics generated by $H_0$ clearly preserves $\gothM$.

Note that $\sigma(H_0) = [E_1,\infty)$ and $\sigma_\pp(H_0) = \sigma(K)$, hence;
$H_0$'s eigenvalues are all embedded in a half-axis of continuous spectrum. 
The ground state energy is $E_1$ with eigenvector $e_1\otimes\vacuum$.

We construct $\beta$-KMS states on $\gothM$ by tensoring the relevant
states on the constituent systems:
\[
\omega_\beta(T\otimes W) = \omega_\beta^\romp(T)\omega_\beta^\cR(W), 
\]
where $T\in\gothM_\romp=\cB(\cK)$ and $W\in\gothW$.

We proceed to construct left and right representations at positive temperature
by tensoring the two individual representations. Again, the left representation is the GNS representation, whereas the right is for extending the GNS representation to standard form at a later stage.
We want representations on the Hilbert space
\[
\cH^\romL = \cH\otimes\cH = \cK\otimes\cF\otimes\cK\otimes\cF,
\] 
but obviously it is often more natural to construct operators on
\[
\cH^\romL_\shuf = \cK\otimes\cK\otimes\cF\otimes\cF.
\]
For this purpose we introduce a unitary transformation shuffling the tensor components:
\begin{equation}\label{Shuffle}
\cH^\romL \ni u\otimes\varphi\otimes v\otimes \psi \to \shuf(u\otimes\varphi\otimes v\otimes \psi) = u\otimes v\otimes \varphi\otimes\psi \in\cH^\romL_\shuf,
\end{equation}
and extended by linearity and continuity.

We define
\begin{equation}\label{Pi-PF}
\begin{aligned}
\pi_{\beta,\roml}^\PF(T\otimes W(f)) & = 
\shuf^*\bigl(T\otimes\one_\cK\otimes
W^\AW_{\beta,\roml}(f)\bigr)\shuf,\\
\pi_{\beta,\romr}^\PF(T\otimes W(f)) & = 
\shuf^*\bigl(\one_\cK\otimes\bT \otimes
W^\AW_{\beta,\romr}(f)\bigr)\shuf,
\end{aligned}
\end{equation}
and extend by linearity to $\gothM$. Write
\[
\tgothM_{\beta,\roml} = \pi_{\beta,\roml}^\PF(\gothM) \quad
\textup{and} \quad
\tgothM_{\beta,\romr} = \pi_{\beta,\romr}^\PF(\gothM)
\]
for the image algebras which are complex $*$-subalgebras of $\cB(\cH^\romL)$.

The $\beta$-KMS state $\omega_\beta$ is represented by the  GNS vector state
\[
\omega_\beta(A) =   \la \Omega_\beta^\PF, \pi^\PF_{\beta,\roml}(A)\Omega_\beta^\PF\ra, 
\]
where the GNS vector is
\[
\Omega_\beta^\PF = \shuf^*\bigl(\Omega_\beta^\romp\otimes \dvacuum\bigr).
\]

Finally we can single out the free $\Omega_\beta^\PF$-Liouvillean as the sum of the
particle and reservoir Liouvilleans:
\[
L_0 = H_0\otimes \one_\cH - \one_\cH\otimes H_0 = 
\shuf^*\bigl(L_\romp\otimes\one_{\cF\otimes\cF} + \one_{\cK\otimes\cK}\otimes L_\cR\bigr)\shuf.
\]
Again the property $L_0 \Omega_\beta^\PF = 0$  determines the
choice of generator uniquely.

As for the spectrum of $L_0$, it is the sum of the spectra of $L_\romp$ and $L_\cR$,
that is
\[
\sigma(L_0) = \RR \quad\textup{and}\quad \sigma_\pp(L_0) = \sigma(L_\romp) = \sigma(K)-\sigma(K).
\]
In particular, $0$ is an at least $\nu$-fold degenerate embedded eigenvalue. 

We end this subsection by passing on to the von Neumann envelopes
\[
\gothM_{\beta,\roml} = \bigl(\tgothM_{\beta,\roml}\bigr)''
 \quad\textup{and}\quad 
 \gothM_{\beta,\romr} = \bigl(\tgothM_{\beta,\romr}\bigr)''.
\]
As for the reservoir, $\Omega^\PF_\beta$ is also
a cyclic and faithful state on $\gothM_{\beta,\roml}$
and hence; $\gothM_{\beta,\roml}$ is in its standard representation.
We identify the modular conjugation to be
\begin{equation}\label{Intro-J}
J = \shuf^*\bigl( J_\romp\otimes J_\cR\bigr)\shuf,
\end{equation}
intertwining the left and right algebras
\[
\bigl(\gothM_{\beta,\roml}\bigr)' = J\gothM_{\beta,\roml} J = \gothM_{\beta,\romr}.
\]

As for the reservoir, the dynamics  $A\to \tau^t_0(A) = \e^{\ri t L_0} A \e^{-\ri t L_0}$ 
extends by continuity from $\tgothM_{\beta,\roml}$ to the von Neumann envelope $\gothM_{\beta,\roml}$,
and -- being unitarily implemented -- is a $\sigma$-weakly continuous group of automorphisms.
The GNS vector $\Omega^\PF_\beta$ remains a $\beta$-KMS vector also for the extended dynamics
on the von Neumann algebra $\gothM_{\beta,\roml}$.

The extended dynamics is (also) unitarily implemented by the standard Liouvillean, which is fixed uniquely
by the demand that $\e^{\ri t L}$ preserves the standard cone
\begin{align}\label{PFCone}
\cP_\beta^\PF = 
\overline{\bigset{AJA\Omega^\PF_\beta}{A\in\gothM_{\beta,\roml}}}.
\end{align}
Again, the standard Liouvillean coincides with the $\Omega^\PF_\beta$-Liouvillean $L_0$.

Note the intertwining property $L_0J = - JL_0$, which is consistent with the 
symmetric structure of the spectrum of $L_0$. 

\subsection{Interacting Pauli-Fierz Systems}\label{Subsect-IntPF}

 In order to obtain an interacting system we add a perturbation to the
 free Hamiltonian, which couples the particle system and the
 reservoir. In this contribution we consider couplings linear in the
 field operators. 

 The object carrying the coupling is a function $\coup\in
 L^2(\RR^3;\cB(\cK))$, an $\Mat_\nu(\CC)$-valued square integrable
 function. The perturbation is then of the form
\[
\phi^\PF(\coup) := \frac1{\sqrt{2}}\int_{\RR^3} \bigl\{\coup(k)\otimes a^*(k) +
\coup(k)^*\otimes a(k)\bigr\}\, \D k,
\]
acting in $\cH = \cK\otimes\cF$. This can be thought of as a Segal field and indeed; if $\nu=1$, it is a Segal field.

The interacting Pauli-Fierz Hamiltonian is given by
\[
 H = H_0 + \phi^\PF(\coup) = K\otimes\one_\cF + \one_\cK\otimes H_\ph + \phi^\PF(\coup).
\]
Discussions of self-adjointness is postponed to the next chapter. We call $H$
the \emph{Pauli-Fierz Hamiltonian}.

At this point we meet a fundamental issue arising when studying
interacting dynamics on Weyl-algebras. 
The Heisenberg evolution is no longer going to preserve the algebra of
observables $\gothM$. 
By the Trotter product formula \cite[Thm.~VIII.31]{ReedSimonI1980} one sees that the double commutant $\gothM''$ 
(the closure of $\gothM$ in the weak operator topology) is invariant, 
but since the Weyl-algebra acts irreducibly on Fock space we find that $\gothM'' = \cB(\cH)$. 
We remark that recently Buchholz and Grundling \cite{BuchholzGrundling2008} introduced an intermediary
resolvent algebra, as an alternative to the Weyl-algebra, which admits non-trivial interacting quantum dynamical systems.

On the positive temperature side, we can perturb the free dynamics $\tau^t_0$ 
on the algebra $\gothM_{\beta,\roml}$. Being a $\sigma$-weakly continuous
group of automorphisms on $\gothM_{\beta,\roml}$ it has a generator $\delta$ which is
a closed operator on $\gothM_{\beta,\roml}$ as a Banach space. Using Araki-Dyson expansions
one can add to $\delta$ perturbations of the form $[A,\cdot]$, for self-adjoint $A\in\gothM_{\beta,\roml}$,
and construct a perturbed dynamics with $\delta + [A,\cdot]$ as generator.
Araki derived the form of the standard Liouvillean unitarily implementing the 
perturbed dynamics to be $L_0 + A - JAJ$.

Derezi{\'n}ski, Jak\v{s}i{\'c} and Pillet \cite{DerezinskiJaksicPillet2003}
extended the analysis of Araki to self-adjoint operators affiliated
with $\gothM_{\beta,\roml}$. The precise construction of the
perturbation we consider can be found in Subsect.~\ref{Subsec-Liovillean}.  Here we just
heuristically explain its structure. Formally, we take the zero
temperature perturbation $\phi^\PF(\coup)$ in its positive temperature
form ``$\pi^\PF_{\beta,\roml}(\phi^\PF(\coup))$''. For bosons 
this is not meaningful, but for fermions where field operators are
bounded this can be taken literally. 
One can try and circumvent this by constructing a strongly continuous
one-parameter group of unitaries ``$s\to\pi^\PF_{\beta,\roml}\bigl(\e^{\ri
  \phi^\PF(s\coup)}\bigr)$''. In general, however, the operators
$\e^{\ri t \phi^\PF(\coup))}$ may not be in the (norm closure) of
$\gothM$. An exception to this is the spin-boson model, or more
generally $\coup$'s of the form $\coup(k) = \coup_0 g(k)$, with
$\coup_0$ a self-adjoint matrix. 

For now we simply write $\phi_{\beta,\roml}^\PF(\coup)$ for the
formal positive temperature Pauli-Fierz field obtained by formally applying the
left representation $\pi^\PF_{\beta,\roml}$ to $\phi^\PF(\coup)$. In
Subsect.~\ref{Subsec-ExisNonexis}, we argue that $\phi_{\beta,\roml}^\PF(\coup)$ constructed
this way is indeed affiliated with $\gothM_{\beta,\roml}$.

One can now argue that
\[
L_\beta = L_0 + \phi_{\beta,\roml}^\PF(\coup) - J\phi_{\beta,\roml}^\PF(\coup)J
\]
is essentially self-adjoint on the intersection of the domains of the three unbounded summands.
This perturbed Liouvillean generates a $\sigma$-weakly continuous
group of automorphisms on $\gothM_{\beta,\roml}$, and by design it is already 
its own standard Liouvillean. It is a fact immediate from Trotter's
product formula that $\e^{\ri t L_\beta}$
preserves the standard  cone \eqref{PFCone}. We call $L_\beta$ the \emph{standard Pauli-Fierz Liouvillean}.

We remark at this stage that it is a result of Derezi{\'n}ski, Jak\v{s}i{\'c} and Pillet \cite{DerezinskiJaksicPillet2003}, that the interacting Pauli-Fierz dynamics
$A\to \tau_\coup^t(A) = \e^{\ri t L_\beta}A\e^{-\ri t L_\beta}$ -- under suitable assumptions -- also admits a normalized faithful $\beta$-KMS vector  $\Omega^\PF_{\beta,\coup}$. A result which goes back to Araki for bounded perturbations from $\gothM_{\beta,\roml}$. 

\subsection{(Open) Problems I}

In this subsection we would like to present some of the questions
one is interested in regarding  Pauli-Fierz systems at zero and positive temperature.
Our focus is on fully coupled Pauli-Fierz systems, that is we largely ignore questions
relevant primarily for the weak coupling regime, which is much better understood than the fully coupled regime.
Some problems are open, some will be either resolved or discussed in the following sections.

\vskip3mm

\noindent\textbf{Zero Temperature:}

\begin{itemize}
\item Under what conditions on $\coup$ does the Hamiltonian admit a ground state, that is;
when is $\Sigma = \inf\sigma(H)$ an eigenvalue. When is it non-degenerate?
\item What can one say about the general structure of point spectrum?
\item What can one say about general regularity properties of eigenstates?
\item Under what conditions do the excited states  vanish due to energy being dispersed by the radiation field?
\item Under what conditions is the underlying continuous spectrum absolutely continuous?
\item Can one prove asymptotic completeness, i.e. show that the canonical wave operators $W_\pm$ are unitary?
\end{itemize}

\noindent\textbf{Positive Temperature:}

\begin{itemize}
\item Under what conditions does the interacting Pauli-Fierz dynamics
$\tau^t_\coup$ on $\gothM_{\beta,\roml}$ admit a $\beta$-KMS state? 
\item Under what conditions is the $\beta$-KMS state the only invariant normal state?
\item Under what conditions is the interacting Pauli-Fierz dynamics \emph{ergodic}:
\[
\forall \Psi\in\cP^\PF_\beta, A\in\gothM_{\beta,\roml}:\quad
\lim_{T\to\infty}\frac1{T}\int_0^T \la\Psi,\tau^t_\coup(A)\Psi\ra \,\D t = 
\la\Omega^\PF_{\beta,\coup},A\Omega^\PF_{\beta,\coup}\ra.
\]
Recall that the normal states on $\gothM_{\beta,\roml}$ are in one-one correspondence with vector states,
pertaining to unit vectors from the standard cone $\cP^\PF_\beta$.
\item Under what conditions is the interacting Pauli-Fierz dynamics \emph{mixing}:
\[
\forall \Psi\in\cP^\PF_\beta, A\in\gothM_{\beta,\roml}:\quad
\lim_{t\to\infty} \la\Psi,\tau^t_\coup(A)\Psi\ra  = 
\la\Omega^\PF_{\beta,\coup},A\Omega^\PF_{\beta,\coup}\ra.
\]
The mixing property is sometimes referred to as \myquote{Return to Equilibrium}.
\item Can one prove an asymptotic completeness property, i.e.
if the $\beta$-KMS state is the only invariant state, is there a unitary \myquote{wave operator}
intertwining the interacting Pauli-Fierz dynamical system
and the Araki-Woods dynamical system for the reservoir.
\end{itemize}

 We make some clarifying remarks pertaining to the questions above.

In order to ensure uniqueness of an existing interacting ground state for $H$, one is typically forced
to assume a priori that $K$'s ground state is non-degenerate, that is $E_1< E_2$. Otherwise one must rely on
the coupling to induce a splitting of the ground state energy, which is difficult to control away from weak coupling. The typical tool here for the Hamiltonian is Perron-Frobenius methods.
We will discuss a possible mechanism to invoke similar arguments for the Liouvillean.

The key tool employed here to study the structure of point spectrum  
and to prove absence of singular continuous spectrum for $H$ and $L_\beta$, is the positive commutator method originally due to Eric Mourre  \cite{Mourre1981}. 
 The deepest results established here will be derived using this technique.
Note that eigenvalues of $H$ and $L_\beta$ are embedded in the continuous spectrum. Pauli-Fierz Hamiltonians and standard Liouvilleans are, however,
too singular to permit an application of standard implementations of the Mourre method \cite{AmreinMonvelGeorgescu1996}. Instead we rely on the singular Mourre theory developed in \cite{GeorgescuGerardMoeller2004a,MoellerSkibsted2004,Skibsted1998}. 

 Suppose one can establish that $H$ has a non-degenerate ground state, no excited states, and purely
absolutely continuous spectrum above the ground state. Then it is a direct consequence of the spectral theorem and the Riemann-Lebesgue lemma that 
\[
\forall \psi,\varphi\in\cH:\quad  \lim_{t\to\infty} \la\varphi,\e^{-\ri t (H-\Sigma)}\psi\ra =
\la\varphi,\Omega_0\ra \la\Omega_0,\psi\ra,
\]
where $\Omega_0$ is a normalized ground state and $\Sigma$ the corresponding ground state energy. 
This is sometimes called \myquote{Approach to the Ground State}
and is the zero-temperature analogue of mixing. 

 By regularity properties of bound states we refer here primarily to
 number bounds, 
which have different interpretations at zero and positive temperature. At zero temperature, bounds such as
 $\la\psi,N^k\psi\ra<\infty$ give control over the infrared catastrophe, in that it controls the number
 of soft photons a bound state carries. Here $N=\one_\cK\otimes\,\D\Gamma(\one_\gothh)$ is the number operator and $\psi$ is a bound state.
At positive temperature one has both a photon
counter  $N\otimes\one_\cH$ and a hole counter $\one_\cH\otimes N$ .
Subtracting one from the other yields a total photon counter, a \myquote{charge}
operator.
A deformed thermal photon cloud may be very far from the Planck distribution while having finite (total) photon
number. A better measure for controlling the photon/hole content of a bound state
is the sum of photons \emph{and} holes $N\otimes \one_\cH+\one_\cH\otimes N$, which is what we understand by the number operator at positive temperature.
Other relevant regularity
questions pertain to the study of momentum content, in particular in the infrared region, 
of bound states. This amounts to studying regularity of $k\to a(k)\psi$.

 Asymptotic completeness at zero temperature expresses that states in the absolutely continuous subspace
correspond exactly to scattering processes with incoming photons entering the interaction region from spatial infinity,
exciting/relaxing the atom and escaping again to spatial infinity leaving the atom in an altered state.
To express this more concisely set $\cV = \cH_\pp$, the closure of the linear span of all eigenstates
of $H$, which may be a one-dimensional subspace of $\cH$. The spaces of incoming and outgoing states are identical and equal
$\cH_{\pm} = \cV\otimes\cF$. 

We define a scattering identification operator $I \colon\cH_{\pm}\to\cH$ by
\[
I (\psi\otimes a^*(f)^n\vacuum) = a^*(f)^n\psi,
\]
and extended (formally) by linearity to $\cH_{\pm}$.
The wave operators, as maps $W_\pm\colon\cH_\pm \to \cH$, are now given by
\[
W_\pm  = \slim_{t\to\infty} \e^{\ri t H} I \e^{-\ri t H_\pm},
\]
where $H_\pm = H_{|\cV}\otimes\one_\cF + \one_\cV\otimes H_\ph$ is the free dynamics on the incoming and outgoing states.
Asymptotic completeness amounts to establishing existence and unitarity of 
wave operators. The scattering operator $S =  W_+^* W_-\colon \cH_-\to\cH _+$ is then a unitary map
from incoming to outgoing states. Wave operators are known to exist (under reasonable assumptions),
but unitarity has only very recently been established by Faupin and
Sigal \cite{FaupinSigal2012b}, for the
particular case of the spin-boson model and in the weak coupling
regime. Faupin and Sigal made use of a crucial number estimate due to De~Roeck and Kupiainen \cite{DeRoeckKuppiainen2012}. Six months after the appearance of \cite{FaupinSigal2012b}, 
De~Roeck, Griesemer and Kupiainen  announced a different proof \cite{DeRoeckGriesemerKuppiainen2013}. The general problem remains open. We will not be discussing
asymptotic completeness further, 
but the techniques we develop have in other contexts been essential to establishing asymptotic completeness.
For further material we refer the reader to \cite{FaupinSigal2013,FaupinSigal2013b,Gerard2002}.

Return to equilibrium, that is; mixing, or its weaker form, ergodicity, expresses the following
physical situation. When the Pauli-Fierz system is uncoupled, the atomic levels are not 
mixed, and hence the system is really a sum of one-dimensional systems. In particular, the reservoir
can not drive the small system towards the equilibrium vector state given by $\Omega^\PF_\beta$.
However, if the coupling does mix the atomic energy levels, then one expects that 
any normal state on $\gothM_{\beta,\roml}$ should evolve towards the 
interacting $\beta$-KMS vector state given by $\Omega^\PF_{\beta,\coup}$. 
Note that the final inverse temperature is enforced on the small system by the reservoir.

As for asymptotic completeness at positive temperature, this is a completely
unexplored question. As far as the author is aware, there is not even
a concise strategy proposed to address this problem.

We have now arrived at the starting point, substance wise, for the positive temperature
side of these notes.
The following quantum analogue of the classical Koopman theorem \cite{ReedSimonI1980}, 
expresses ergodicity and mixing in terms of spectral properties of generators.
The translation of  \myquote{Return to Equilibrium} into a spectral problem is often referred to
as quantum Koopmanism.

\begin{theorem} We have the following at positive temperature $\beta\in(0,\infty)$:
\begin{Enumerate}
\item The interacting Pauli-Fierz dynamics at positive temperature is
ergodic if and only if $\sigma_\pp(L_\beta)=\{0\}$ and $0$ is a simple eigenvalue.
(The corresponding eigenvector being the $\beta$-KMS vector.)
\item The interacting Pauli-Fierz dynamics at positive temperature is
mixing if and only if $\wlim_{t\to\infty} \e^{\ri t L_\beta} = \vert \Omega^\PF_\beta\ra\la\Omega^\PF_\beta\vert$.
\item If the spectrum of $L_\beta$ is purely absolutely continuous, except for a simple eigenvalue at $0$,
then the interacting Pauli-Fierz dynamics is mixing. 
\end{Enumerate}
\end{theorem}

We have now set the stage for embarking on a systematic spectral analysis of Pauli-Fierz
Hamiltonians and standard Pauli-Fierz Liouvilleans.

\noindent\textbf{Acknowledgments:} The author would first and foremost like to thank
Laurent Bruneau, Vojkan Jak\v{s}i{\'c} and Claude-Alain Pillet for the invitation to lecture
at the summer school on Non-equilibrium Statistical Mechanics, held at CRM in Montreal in the summer of 2011.
These notes are a direct result of their invitation. Furthermore, I would like thank Jan Derezi\'nski and Vojkan Jak\v{s}i{\'c}
for valuable comments regarding the contents of the notes,
part of which were typed during stays at CRM in Montreal and Dokuz Eyl{\"u}l University, Izmir. 
The author would like to thank these two institutions for hospitality.

\newpage

\section{Construction and Properties of Operators}

 In this section we construct the Pauli-Fierz Hamiltonian
 and its positive temperature counterpart, the standard Pauli-Fierz Liouvillean.
 Furthermore, we establish some of their basic properties. We will throughout these notes
 make heavy use of the $C^1(A)$ commutator calculus. For the convenience of the reader 
 we have, in the form of Appendix A, included a condensed presentation 
 of the elements of the calculus that we rely on.

\subsection{The Pauli-Fierz Hamiltonian}

Recall from Subsect.~\ref{Subsec-smallSystem} that as a small quantum system we took a finite dimensional
Hilbert space $\cK=\CC^\nu$ with Hamiltonian $K\in \Mat_\nu(\CC)$,
a self-adjoint $\nu\times \nu$ matrix $K^* = K$. In fact we chose
$K$ to be diagonal with its real eigenvalues sitting on the diagonal.

The dispersion relation for the field is the massless relativistic relation
$k\to |k|$ considered as a multiplication operator on $\gothh = L^2(\RR^3)$.
This gives rise to the second quantized free field energy $H_\ph = \D\Gamma(|k|)$,
as a self-adjoint operator on the bosonic Fock-space $\cF = \Gamma(\gothh) = \oplus_{\ell=0}^\infty \gothh^{\otimes_s \ell}$.
We write $\vacuum = (1,0,0,\dots)$ for the vacuum state in $\cF$.
Our inner products will always be conjugate linear in the first variable, and linear in the second.

We define a class of admissible coupling operators/functions
\[
\coup\in \cB(\cK;\cK\otimes\gothh)= L^2\bigl(\RR^3;\Mat_\nu(\CC)\bigr).
\]
That the two spaces above can be identified can be seen as follows: 
If $\coup\colon \RR^3\to \Mat_\nu(\CC)$ is square integrable
one can define a bounded operator $B_\coup\in\cB(\cK;\cK\otimes \gothh)$ by 
\[
(B_\coup v)(k) = \coup(k)v, 
\]
where we identified $\cK\otimes \gothh$ isometrically with $L^2(\RR^3;\CC^\nu)$.
Then 
\[
\|B_\coup\|^2 = \sup_{|v|\leq 1} \|B_\coup v\|^2 = \sup_{|v|\leq 1}\int_{\RR^3} |\coup(k)v|^2\,\D k \leq \|\coup\|^2
\]
and the linear map $\coup\to B_\coup$ is a contraction, but it is not an isometry.
To see that it is surjective with a bounded inverse, let $B\in \cB(\cK;\cK\otimes \gothh)$ and define
the candidate for an inverse
$\coup$ by $\coup_{ij}(k) = \la (B e_i)(k),e_j\ra$, 
where $e_1,\dots,e_\nu$ is the standard basis for $\CC^\nu$.
Then 
\[
\sum_{j=1}^\nu\int_{\RR^3}|\coup_{ij}(k)|^2\D k =\int_{\RR^3}|(Be_i)(k)|^2\D k =\|Be_i\|_{\cK\otimes\gothh}^2\leq \|B\|^2. 
\] 
Hence
\[
\|\coup\|^2 = \int_{\RR^3} \|\coup(k)\|^2 \D k \leq  \sum_{1\leq i,j\leq \nu}\int_{\RR^3} |\coup_{ij}(k)|^2 \D k\leq \nu \|B\|^2.
\]
From now on we will identify couplings $\coup$ with elements of $L^2(\RR^3;M_\nu(\CC))$, and norms of couplings will be $L^2$-norms.
We remark that the identification of coupling operators as
$\cB(\cK)$-valued functions above is particular to finite
dimensional small systems, cf.~\cite[Remark~5.1]{FaupinMoellerSkibsted2011a}.
Let $\mu>0$ be arbitrary, but fixed.
For the coupling $\coup$ we assume the existence of a constant
$C>0$ such that
\[
\mathrm{\mathbf{(H{\coup}n)}}\quad \begin{aligned}
&\forall k\in\RR^3, |k|\leq 1, \ \textup{and} \ |\alpha|\leq n: 
\qquad \|\partial_k^\alpha G(k)\|\leq C |k|^{n-\frac32+\mu -|\alpha|+\frac{\delta_{n,0}}{2}}\\ 
&\forall k\in\RR^3, |k|\geq 1, \ \textup{and} \ |\alpha|\leq n: 
\qquad \|\partial_k^\alpha G(k)\|\leq C |k|^{-\frac32 -\mu}. 
\end{aligned}
\]
The derivatives are distributional derivatives. We will make use of the condition
\HGCond{n} on $\coup$ with $n=0,1,2$. Note that \HGCond{n+1} implies \HGCond{n}.

The above
conditions reflect that $|k|^{|\alpha|-n}\partial_k^\alpha \coup$ is slightly better than
square integrable near zero, and  the  $\partial_k^\alpha G$'s are slightly better than
square integrable at infinity. For our commutator estimates in
Sect.~\ref{Sec-CommEst} it will not suffice to demand just square
integrability.  We remark that there is nothing special
about three dimensions or the dispersion $|k|$. For some results we
could deal with infinite dimensional small system $\cK$, and more
singular $\coup$'s.
 The above special case however captures the
essentials, and permits us to formulate simple - yet pertinent -
conditions that can be used for all our results at zero
temperature.

We now define the free and coupled Hamiltonians as
\begin{equation*}
H_0   = K\otimes\one_\cF + \one_\cK\otimes H_\ph
\quad \textup{and}\quad
H = H_0 + \phi^\PF(\coup),
\end{equation*}
where
\begin{equation}\label{PF-Coupling}
\phi^\PF(\coup) = \frac1{\sqrt{2}}\int_{\RR^3} \bigl\{ \coup(k)^* a(k) + \coup(k) a^*(k)\bigr\} \,\D k.
\end{equation}
In the following we will drop the superscript $\PF$ to simplify notation. The reader should be able to 
tell from the argument when $\phi$ denotes a Segal field and when it is of the coupling type \eqref{PF-Coupling}.
We remark that $H_0$ is self-adjoint on $\cD(H_0) = \cD(\one_\cK\otimes H_\ph)$ and that
\begin{equation}\label{HCore}
\cC = \cK\otimes \Gamma_\fin(C_0^\infty(\RR^3))
\end{equation}
is a core for $H_0$.  Furthermore, as for Segal fields, by Nelson's analytic vector theorem $\phi(\coup)$
is essentially self-adjoint on $\cC$ as well. In fact Segal fields are a special case, corresponding to $\nu=1$.
 See \cite{Merkli2006} for a proof. 
The notation $\Gamma_\fin(V)$ with $V\subset \gothh$ a subspace, denotes the
algebraic direct sum of $V^{\otimes_\roms n}$ with tensor products of proper
subspaces of Hilbert spaces always being algebraic, whereas
tensor products of Hilbert spaces always denote Hilbert space tensor
products, i.e., completion of algebraic tensor products.

We will as usual use the notation $N$ for the number operator $\D\Gamma(\one_\gothh)$
as an operator on $\cF$, and we will recycle the same notation on $\cH$ 
instead of the more cumbersome $\one_\cK\otimes N$.

Note the easy to verify bounds for field operators

\begin{lemma}\label{Lemma-Field-H-Bounds} Let $\coup\in L^2(\RR^3;\Mat_\nu(\CC))$ and $\psi\in\cC$.
\begin{Enumerate}
\item\label{Item-field-H-N-Bounds} The following bounds hold true for all $\sigma>0$
\begin{align}\label{PhiNumberBound-Norm}
\|\phi(\coup)\psi\| & \leq \sqrt{2}\|\coup\|\|\sqrt{N}\psi\| + \frac1{\sqrt{2}}\|\coup\|\|\psi\| \\ 
\label{PhiNumberBound}
 |\la\psi,\phi(\coup)\psi\ra| & \leq \sigma \la\psi,N\psi\ra
 + (2\sigma)^{-1}\|\coup\|^2 \|\psi\|^2.
\end{align}
\item\label{Item-field-H-E-Bounds} If furthermore $|k|^{-1/2}\coup\in  L^2(\RR^3;\Mat_\nu(\CC))$, then for all $\sigma>0$
\begin{align}\label{PhiEnergyBound-Norm}
\|\phi(\coup)\psi\| & \leq \sqrt{2}\||k|^{-\frac12}\coup\|\|(\one_\cK\otimes\sqrt{H_\ph})\psi\| + \frac1{\sqrt{2}}\|\coup\|\|\psi\|\\  
\label{PhiEnergyBound}
|\la\psi,\phi(\coup)\psi\ra| & \leq \sigma \la\psi,\one_\cK\otimes H_\ph \psi\ra
+ (2\sigma)^{-1}\|\coup/\sqrt{|k|}\|^2\|\psi\|^2.
\end{align}
\end{Enumerate}
\end{lemma}

The bounds in \ref{Item-field-H-N-Bounds} extend by continuity to $\psi \in\cD(\sqrt{N})$,
whereas the bounds in \ref{Item-field-H-E-Bounds} extend by continuity to $\psi\in\cD(\one_\cK\otimes\sqrt{H_\ph})$.
 We note that the bound \eqref{PhiNumberBound-Norm} 
 implies that $\cD(\sqrt{N})\subseteq\cD(\phi(\coup))$, just as for Segal fields.

Suppose \HGCond{0}. By \eqref{PhiEnergyBound-Norm} and Kato-Rellich's theorem,
\cite[Thm.~X.12]{ReedSimonII1975},  $H$ is essentially self-adjoint on $\cC$, bounded from below
and $\cD(H)=\cD(H_0)$. In particular, the domain of $H$ does not depend on $\coup$.

We furthermore observe that if we equip the space of $\coup$'s satisfying \HGCond{0}
with the norm $\|\coup\|_0^2 = \int_{\RR^3} (1+|k|^{-1})\|\coup(k)\|^2\D k$, then
 the resolvent map $(z,\coup)\to (H-z)^{-1}$ is norm continuous. Here $\im z\neq 0$.
We introduce notation for the bottom of $H$'s spectrum
\begin{equation}\label{BottomOfSpec}
\Sigma = \inf\sigma(H) > -\infty.
\end{equation}
The spectrum of $H$ is in fact a half-line starting at $\Sigma$ as we now proceed to prove, using an argument 
from \cite{GeorgescuGerardMoeller2004b}. We pass via two useful results on the way,
the first of which involves the Mourre class of operators introduced in Appendix~\ref{App-MourreClass}. 

\begin{lemma}\label{Lemma-H-is-C1N} Assume \HGCond{0}. Then $H$ is of
  class  $C^1_\Mo(N)$ and
the operator representing the commutator form is $[H,N]^\circ = \ri\phi(\ri
  \coup)$.
\end{lemma}

\begin{proof} We aim to use Proposition~\ref{Prop-MourreEquiv} to establish the lemma.
Note first that the property  \ref{Item-bstability} in 
Proposition~\ref{Prop-MourreEquiv}~\ref{Item-MourreC1} trivially holds true, since $\cD(H)=\cD(H_0)$
and $N$ commutes with $H_0$. 

Secondly, since $H_0 + N$ is essentially self-adjoint on $\cC$ -- being a direct sum of multiplication operators --
we conclude that $\cC$ is dense in $\cD(H)\cap\cD(N) = \cD(H_0 + N)$ with respect to the intersection topology.
The core $\cC$ was introduced in \eqref{HCore}.

We can now compute in the sense of forms on $\cC$
\[
[H,N] = [\phi(\coup),N] = \ri\phi(\ri \coup).
\]
Since $\phi(\ri\coup)$ is $N^{1/2}$-bounded the above form identity now extends
to the intersection domain $\cD(H)\cap\cD(N)$.
By \eqref{PhiEnergyBound-Norm} (applied with $\coup$ replaced by $\ri\coup$) we thus find that property \ref{Item-MourreCommBound} in Proposition~\ref{Prop-MourreEquiv}~\ref{Item-MourreC1} is also satisfied and hence, $H$ is of class $C^1_\Mo(N)$. 
\end{proof}

 The second ingredient is a version of the so-called pull through formula

\begin{proposition}\label{PullthroughH} Suppose \HGCond{0}. For any $z\in \CC\backslash [\Sigma,\infty)$ and
  $\psi\in \cD(\sqrt{N})$ we have as an $L^2(\RR^3;\cH)$-identity
\[
a(k)(H -z)^{-1}\psi = (H+|k| -z)^{-1} a(k)\psi -
\frac1{\sqrt{2}}(H+|k| - z)^{-1}(\coup(k)\otimes \one_\cF)\psi.
\]
\end{proposition}

\begin{remark}\label{RemOnHbeingC1N} The fact that 
  $H$ is of class $C^1_\Mo(N)$, and hence in particular of class $C^1(N)$, 
  implies that $\cD(N)$,
  and by interpolation $\cD(\sqrt{N})$, is preserved
  by resolvents of $H$. See Lemma~\ref{Lemma-Prop-SA-C1}~\ref{Item-Cheap-DomInv}. 
  Hence both sides of the pull through formula
  define elements of $L^2(\RR^3;\cH)$. 
\end{remark}

\begin{proof} Let $\tpsi\in\cC$ and compute
\[
a(k)(H-z)\tpsi = (H+|k| -z)a(k)\tpsi + \frac1{\sqrt{2}}(\coup(k)\otimes\one_\cF)\tpsi
\]
as an $L^2(\RR^3;\cH)$-identity, where the only possibly irregular
contribution is $\coup$ near zero. Since $z-|k|\in\rho(H)$ - the
resolvent set for $H$ - we
obtain the $L^2(\RR^3;\cH)$-identity
\[
a(k)\tpsi = (H+|k|-z)^{-1}a(k)(H-z)\tpsi - \frac1{\sqrt{2}}(H+|k|-z)^{-1}(\coup(k)\otimes\one_\cF)\tpsi.
\]
Let $h\in L^2(\RR^3)$ and $\varphi\in\cC$. Then
\[
\big\la a^*(h)\varphi,\tpsi\big\ra = \big\la
\tvarphi,(H-z)\tpsi\big\ra - \int_{\RR^3} \frac{\overline{h(k)}}{\sqrt{2}}\big\la
\varphi,(H+|k|-z)^{-1}(\coup(k)\otimes\one_\cF)\tpsi\big\ra\,\D k,
\]
where
\[
\tvarphi = \int_{\RR^3} h(k) a^*(k)(H+|k|-z)^{-1}\varphi \,\D k\in\cH.
\]
From this expression, and $H$ being essentially self-adjoint on $\cC$,
we observe that the above identity remains true for $\tpsi\in \cD(H)$.
Inserting $\tpsi = (H-z)^{-1}\psi$, where $\psi\in\cD(\sqrt{N})$
yields the proposition. Here we used that $L^2(\RR^3)\otimes \cC$
(algebraic tensor product) is dense in $L^2(\RR^3;\cH)$.
\end{proof}

For stronger versions of the pull through formula see \cite{BruneauDerezinski2004,Gerard2000}.
We are now ready to show that the spectrum is a half-axis. The
argument goes back to \cite{GeorgescuGerardMoeller2004b}, cf. also \cite{BruneauDerezinski2004}.

\begin{theorem}\label{HVZH} Suppose \HGCond{0}. Then $\sigma(H) = [\Sigma,\infty)$.
\end{theorem}

\begin{proof}
It suffices to show that $\sigma((H-\Sigma+1)^{-1}) \supset (0,1]$.
To see this, let $\lambda>0$, $\epsilon>0$
and choose $\tpsi\in \one[H\leq \Sigma+\epsilon/2]\cH$ to be
normalized. Since $\cC$ is dense in $\cD(H)$, we can pick a normalized
$\psi\in\cC$ such that $\|(H-\Sigma)(\tpsi-\psi)\|\leq \epsilon/2$
and hence we must have $\|(H-\Sigma)\psi\|\leq \epsilon$.

Choose a function $h\in C_0^\infty(\RR)$ real-valued with $\|h\| = 1$ and $\supp h\subseteq [-1,1]$.
Put $h_n(k) =  n^{3/2}h(n(|k|-\lambda))$. Form
$\psi_n = a^*(h_n)\psi$ and compute for $\vphi\in\cD(\sqrt{N})$ using the
pull through formula Proposition~\ref{PullthroughH}
\begin{align*}
& \bigl\la\vphi,\bigl((H-\Sigma + 1)^{-1}-(\lambda + 1)^{-1}\bigr)\psi_n \bigr\ra\\
& = \int_{\RR^3} h_n(k)\bigl\la a(k)\bigl((H-\Sigma + 1)^{-1}-(\lambda + 1)^{-1}\bigr) \vphi,\psi\bigr\ra\,\D k\\
& = \int_{\RR^3} h_n(k)\bigl\la \bigl((H+|k|-\Sigma + 1)^{-1}-(\lambda + 1)^{-1}\bigr) a(k)\vphi,\psi\bigr\ra\,\D k\\
&\quad 
- \int_{\RR^3} \frac{h_n(k)}{\sqrt{2}} \bigl\la \coup(k)\otimes\one_\cF\vphi, (H+|k|-\Sigma + 1)^{-1}\psi\bigr\ra\, \D k.
\end{align*}
Since $h_n$ goes to zero weakly in $L^2(\RR^3)$, the last term is $o(1)\|\vphi\|$
in the limit of large $n$. To deal with the first term on the right-hand side we estimate
using the support properties of $h_n$ and the choice of $\psi$:
\begin{align*}
&\bigl| h_n(k)\bigl\la \bigl((H+|k|-\Sigma + 1)^{-1}-(\lambda + 1)^{-1}\bigr) a(k)\vphi,\psi\bigr\ra\bigr|\\
&\leq \frac{|h_n(k)|}{\sqrt{|k|}} \big\| (H_\ph+|k|+1)^{-1}\sqrt{|k|}a(k)\vphi\big\|\\
&\quad \times
\bigl\|(H_\ph+|k|+1)(H+|k|-\Sigma + 1)^{-1}\bigl((H-\Sigma) + (|k|-\lambda)\bigr)\psi\bigr\|\\
&\leq C\Bigl(\epsilon+\frac1{n}\Bigr)\frac{|h_n(k)|}{\sqrt{|k|}}\bigl\| \sqrt{|k|}a(k)(H_\ph+1)^{-1}\vphi\bigr\|.
\end{align*}
Noting that $\|h_n/\sqrt{|k|}\|\leq (\lambda-1/n)^{-1/2}\|h_n\| = (\lambda-1/n)^{-1/2}$ we conclude from Cauchy-Schwartz that
\[
\bigl|\bigl\la\vphi,\bigl((H-\Sigma + 1)^{-1}-(\lambda + 1)^{-1}\bigr)\psi_n \bigr\ra\bigr|\leq C(\epsilon + o(1))\|\vphi\|,
\]
where $o(1)$ refers to the large $n$ limit. It now remains to prove that $\|a^*(h_n)\psi\|$ is bounded away from zero. But this follows from the computation 
\[
\|a^*(h_n)\psi\|^2 = \|\psi\|^2 + \|a(h_n)\psi\|^2.
\]
Recall that when $h_n$ goes to zero weakly, we have $a(h_n)\psi\to 0$ in norm, whenever 
$\psi\in\cD(\sqrt{N})$. 
\end{proof}

We end this subsection introducing some extra structure that will be used in the next subsection.

We define a conjugate linear involution operator $\Conj$ on $\cH$ as follows.
It is a tensor product of two conjugate linear involutions, one on $\cK$ and
one on $\cF$. On $\cK$ we simply take coordinate wise complex conjugation $(\conj v)_ j = \bar{v}_j$, and on
$\cF$ we take second quantized complex conjugation $\Gamma(\conj)$, acting
on an $n$-particle state by complex conjugation, or equivalently
described by the intertwining $\Gamma(\conj) a^\#(g) \Gamma(\conj) =
a^\#(\conj g)$. In conclusion $\Conj = \conj\otimes\Gamma(\conj)$. Note that
$\la \Conj\psi,\varphi\ra = \la \Conj\varphi,\psi\ra$.

With this choice of conjugation we can define $H^\conj = \Conj H \Conj = H_0
+ \phi(\overline{\coup})$. Note that $H_0^\conj = \Conj H_0\Conj = H_0$. Clearly,
the spectrum, pure point spectrum and absolutely/singular continuous
spectrum of the two operators coincide. Eigenvectors are related by
$\psi^\conj=\Conj\psi$, where $H\psi = \lambda \psi$ and $H^\conj\psi^\conj = \lambda \psi^\conj$.
Finally we observe that the spectral resolutions $E$ and $E^\conj$ of the operators
$H$ and $H^\conj$ are related by $E^\conj_\psi = E_{\Conj\psi}$.

\subsection{The Standard Pauli-Fierz Liouvillean}\label{Subsec-Liovillean}

 The Liouvillean, at inverse temperature $\beta>0$,  is a self-adjoint operator
on the doubled Hilbert space $\cH^\romL := \cH\otimes\cH$.
The zero temperature Liouvillean, corresponding to $\beta=\infty$, is given by
 \[
 L_\infty = H\otimes \one_\cH - \one_\cH\otimes H^\conj,
 \]
which is essentially self-adjoint on algebraic tensor products $\cD\otimes \cD$, where
$\cD\subseteq \cH$ is a core for $H$. See \cite[Thm.~VIII.33]{ReedSimonI1980}. As a choice of core we take 
\begin{equation}\label{LCore}
\cC^\romL = \cC\otimes\cC,
\end{equation}
where $\cC$ was defined in \eqref{HCore}.
Observe that $L_\infty$ is unbounded from below and indeed $\sigma(L_\infty)=\RR$.

We furthermore write
\[
L_0 = H_0\otimes\one_\cH- \one_\cH\otimes H_0
\]
for the uncoupled Liouvillean. Recall that $H_0^\conj = H_0$. With this notation, at least formally,
the zero temperature ($\beta=\infty)$ Liouvillean can be written as
the operator sum
$L_\infty = L_0 + \phi(\coup)\otimes\one_\cH - \one_\cH\otimes
\phi(\overline{\coup})$. 

We will need stronger conditions than \HGCond{n} on the coupling $\coup$
when dealing with the Liouvillean. Let $n\in\NN_0$. We assume that $\coup$ admits $n$ distributional derivatives in $L_{\loc}^1(\RR^3;M_\nu(\CC))$ 
and the existence of a constant $C>0$ such that
\[
\mathrm{\mathbf{(L{\coup}n)}}\quad 
\begin{aligned}
&\forall k\in\RR^3, |k|\leq 1, \quad \textup{and} \quad |\alpha|\leq n: 
\qquad \|\partial_k^\alpha G(k)\|\leq C |k|^{n-1+\mu -|\alpha|}\\ 
&\forall k\in\RR^3, |k|\geq 1, \quad \textup{and} \quad |\alpha|\leq n: 
\qquad \|\partial_k^\alpha G(k)\|\leq C |k|^{-\frac32 -\delta_{\alpha,0}-\mu}. 
\end{aligned}
\]
We will make use of the condition
\LGCond{n} on $\coup$ with $n=0,1,2$. Note that \LGCond{n+1} implies \LGCond{n} and  \LGCond{n} implies \HGCond{n}.
As for the Hamiltonian, there is nothing particular about dimension
$3$. The difference between \HGCond{n} and \LGCond{n} comes from
having to absorb an infrared singularity from the Planck density \eqref{Intro-Planck},
which mix  the left (photons) and right (holes)
field  components  at positive temperature, $\beta<\infty$. 
One could use a different density, modifying \LGCond{n} accordingly. See also Remark~\ref{Rem-SpinBoson}.

We now give an explicit construction of the positive temperature
perturbation denoted by $\phi^\PF_{\beta,\roml}(G)$
and formally introduced in Subsect.~\ref{Subsect-IntPF}.

 For $G_\romL\in L^2(\RR^3;\cB(\cK\otimes\cK))$ we extend the definition
of the (zero temperature) left and right annihilation and creation operators to read 
\begin{equation}\label{ZeroTemplrCCR}
\begin{aligned}
a_\roml(\coup_\romL) &= \shuf^* \Bigl(\int_{\RR^3}
\coup_\romL(k)^*\otimes a(k)\otimes\one_\cF \, \D k \Bigr)\shuf\\
a_\romr(\coup_\romL) &= \shuf^* \Bigl(\int_{\RR^3} \coup_\romL(k)^*
\otimes\one_\cF\otimes\, a(k) \, \D k \Bigr)\shuf.
\end{aligned}
\end{equation}
Here $\shuf$ is the unitary shuffle defined in \eqref{Shuffle}. The creation operators 
$a^*_{\roml/\romr}(\coup_\romL)$, adjoints of the annihilation operators, are represented by similar formulas.

 For use as $\coup_\romL$ we define
 \begin{equation}\label{GlandGr}
 \coup_\roml(k) = \coup(k)\otimes\one_\cK \quad 
\textup{and} \quad  
\coup_\romr(k) = \one_\cK\otimes\, \overline{\coup(k)}.
\end{equation}
With this definition we have $J_\romp \coup_\roml(k) = \coup_\romr(k)
J_\romp$, where $J_\romp$ is the modular conjugation on the particle
system \eqref{Jp}.

 Recalling the form of the Araki-Woods
annihilation and creation operators \eqref{Intro-lr-AW-Fields}, 
we can now define positive temperature left and right annihilation and creation operators
\begin{equation}\label{PF-AW-CCR}
\begin{aligned}
a_{\beta,\roml}(\coup) &= a_{\roml}(\sqrt{1+\rho_\beta}\, \coup_\roml) + a^*_{\romr}(\sqrt{\rho_\beta}\,\coup_\roml^*),\\
a_{\beta,\romr}(\coup) &= a_{\romr}(\sqrt{1+\rho_\beta}\, \coup_\romr) + a^*_{\roml}(\sqrt{\rho_\beta}\,\coup_\romr^*)
\end{aligned}
\end{equation}
and creation operators as their adjoints. Note that as a function of $\coup$, the right annihilation operator 
is linear. This fits the interpretation that the right annihilation
operator annihilates a hole, i.e. creates a photon.
At zero temperature, $\beta=\infty$, this reduces to
$a_{\infty,\roml}^\#(\coup) = a_\roml^\#(\coup_\roml)$ and
$a^\#_{\infty,\romr}(\coup) = a^\#_\romr(\coup_\romr)$.

The self-adjoint operator formally corresponding to $\phi(\coup)$ at positive
temperature under the representation $\pi^\PF_{\beta,\roml}$,
cf.~\eqref{Pi-PF}, 
can now be explicitly
written as
\begin{equation}\label{phiLAWForm}
\phi^\PF_{\beta,\roml}(\coup) = \frac1{\sqrt{2}}\bigl(a_{\beta,\roml}^*(\coup) + a_{\beta,\roml}(\coup) \bigr).
\end{equation}
We proceed to identify a suitable $\coup_\romL$ 
useful for reformulating $\phi^\PF_{\beta,\roml}(\coup)$ in terms of
left and right fields $\phi_\roml(\coup_\romL)$ and $\phi_\romr(\coup_\romL)$.
Here
\[
\phi_{\roml/\romr}(\coup_\romL) = \frac1{\sqrt{2}}\bigl(a^*_{\roml/\romr}(\coup_\romL) +  a_{\roml/\romr}(\coup_\romL) \bigr)
\]
are defined in terms of
left and right annihilation and creation operators \eqref{ZeroTemplrCCR}.

All annihilation and creation operators, having densely defined adjoints, are closable on $\cC^\romL$.
We observe that the left sets of annihilation and creation operators form
representations of CCR, non-Fock at positive temperature. As for the right representations,
they are also representations of CCR, but with the roles of annihilation and creation operators reversed.
The left and right operators, at the same inverse temperature, commute.
 As for the field operators, they are all -- by the usual analytic vector argument -- essentially self-adjoint
 on $\cC^\romL$. See \cite{Merkli2006}. Left and right fields, at the same inverse temperature, commute.

  At finite inverse temperature $\beta$ we introduce $\coupbl,\coupbr\in L^2(\RR^3;\cB(\cK\otimes\cK))$
  by the prescription
\begin{equation}\label{GlandGrbeta}
\coupbl = \sqrt{1+\rho_\beta}\,\coup_\roml -
   \sqrt{\rho_\beta}\,\coup_\romr^* \quad \textup{and} \quad 
 \coupbr = \sqrt{1+\rho_\beta}\,\coup_\romr-\sqrt{\rho_\beta}\,\coup_\roml^*.
\end{equation}
Note that $\coup_\roml$ and $\coup_\romr$ are the zero temperature
limits of $\coupbl$ and $\coupbr$,
and $J_\romp \coupbl(k) = \coupbr(k) J_\romp$.
 The interaction at
positive temperature is $W_\beta(\coup)$ where
\begin{equation}\label{TotalInt}
 W_\beta(\coup) :=  \phi^\PF_{\beta,\roml}(\coup)-\phi^\PF_{\beta,\romr}(\coup) = \phi_\roml(\coupbl)
- \phi_\romr(\coupbr).
\end{equation}
The two expressions can easily be seen to coincide on $\cC^\romL$, a common domain of essential self-adjointness.

The positive temperature Liouvillean is thus densely defined, a priori
on $\cC^\romL$, as
the operator sum
\[
L_\beta = L_0 + W_\beta(\coup) = L_\infty + I_\beta(\coup),
\]
where 
\begin{equation}\label{RelativeInt}
I_\beta(\coup)  = \phi_\roml(\coupbl-\coup_\roml) - \phi_\romr(\coupbr-\coup_\romr).
\end{equation}
That the Liouvilleans $L_\beta$, for $0 < \beta < \infty$, are essentially self-adjoint
on $\cC^\romL$ was proved in \cite[Lemma~3.2]{JaksicPillet1996a},
cf. also \cite{BachFroehlichSigal2000,DerezinskiJaksic2001,Merkli2001}, using Nelson's commutator
theorem \cite[Thm.~X.37]{ReedSimonII1975}.
This requires that $\coup$ can absorb a power of the dispersion $|k|$,
which is the source of the $\delta_{\alpha,0}$ term in the ultraviolet
part of condition \LGCond{n}. We warn the reader that the 
domain of $L_\beta$ may depend on both $\beta$ and $\coup$,
an issue that complicates the analysis of the operator. Proposition~\ref{Prop-BasicLReg}~\ref{Item-NLDomainInv}
below remedies this issue somewhat.

We write $N^\romL = N\otimes \one_\cF + \one_\cF \otimes N$ for the
number operator on $\cF\otimes\cF$, and as for $N$ we use the same
notation to denote $\one_{\cK}\otimes N\otimes\one_\cK\otimes N$.

We recall the modular conjugation $J$ from \eqref{Intro-J}, which we here express
in terms of the conjugation $\Conj$ from the end of the last subsection:
\begin{equation}\label{ModConj}
J = (\Conj\otimes \Conj )\,\cE,
\end{equation}
where $\cE$ is the exchange operator defined on simple tensors by
$\cE(\psi\otimes\varphi) = \varphi \otimes\psi$. Here $\psi,\varphi\in\cH$.
Clearly $JL_\infty J = - L_\infty$. Indeed, the identity holds on $\cC^\romL$
and extends to $\cD(L_\infty)$ since $\cC^\romL$ is an operator core for $L_\infty$.

Computing as an identity first on $\cC^\romL$ we find
\[
J  \phi_\roml(\coupbl) J  = \phi_\romr(\coupbr)
\quad \textup{and}\quad J  \phi^\PF_{\beta,\roml}(G) J  = \phi^\PF_{\beta,\romr}(G),
\]
and hence
\[
J L_\beta J = - L_\beta.
\]
As above one should first verify the identities on $\cC^\romL$ and extend by continuity to $\cD(L_\beta)$.
Consequently, we observe that the spectrum and pure point spectrum of $L_\beta$
is reflection symmetric around $0$. Furthermore the spectral resolution 
$E^\beta$ associated with $L_\beta$ satisfies $E^\beta(B) = JE^\beta(-B)J$ and hence the absolutely and singular continuous spectra
of $L_\beta$ are also reflection symmetric.

\begin{lemma} Let $\coup\in L^2(\RR^3;\Mat_\nu(\CC))$ and  $\psi\in\cC^\romL$.
\begin{align}\label{WbetaBoundNorm}
\|W_\beta(\coup)\psi\| & \leq \bigl(\|\coup\| + \frac{2}{\sqrt{\beta}}\||k|^{-\frac12}\coup\|\bigr) \bigl(2^{\frac32}\|\psi\| + 2^\frac12 \|\sqrt{N^\romL}\psi\|\bigr),\\
\label{WbetaBoundForm}
|\la \psi,W_\beta(\coup)\psi\ra| &  \leq \sigma\la\psi,N^\romL\psi\ra + \sigma^{-1} \bigl(\|\coup\| + \frac{2}{\sqrt{\beta}}\||k|^{-\frac12}\coup\|\bigr)^2\|\psi\|^2.
\end{align}
\end{lemma}

\begin{proof} Use the representation of $W_\beta(\coup)$ as the difference of two zero temperature fields,
cf.~\eqref{TotalInt}, together with Lemma~\ref{Lemma-Field-H-Bounds}~\ref{Item-field-H-N-Bounds} and the bounds
\begin{align*}
\|\coup_{\beta,\roml/\romr}\| & \leq \|\sqrt{1+\rho_\beta}\coup\| + \|\sqrt{\rho_\beta}\coup\|\\
& \leq \|\coup\| + 2 \|\sqrt{\rho_\beta}\coup\|\\
& \leq \|\coup\| + \frac{2}{\sqrt{\beta}}\||k|^{-\frac12}\coup\|.
\end{align*}
The last inequality follows from the estimate
\[
|k|\rho_\beta = \frac{|k|}{\e^{\beta|k|}-1} \leq \beta^{-1}.
\]
\end{proof}

We end the subsection with a proposition that permits us to work
effectively with standard Liouvilleans, despite domain problems.
Its proof follows closely arguments from
\cite{FaupinMoellerSkibsted2011a}, establishing similar statements
for technically related operators. The proposition 
involves the $C^1_\Mo(A)$ class introduced in Appendix~\ref{App-MourreClass}.

\begin{proposition}\label{Prop-BasicLReg} Suppose \LGCond{0}. The following holds
\begin{Enumerate}
\item\label{Item-NC1L} $N^\romL\in C^1_\Mo(L_\beta)$ and the operator
  $[N^\romL,L_\beta]^\circ$ extends from $\cD(N^\romL)$ 
by continuity to an element of $\cB(\cD(\sqrt{N^\romL});\cH^\romL)$.
\item\label{Item-NLDomainInv} $\cD(N^\romL)\cap \cD(L_\beta)$ does not depend on $\beta$, nor on $\coup$.
\item\label{Item-IntersectCore} $\cC^\romL$ is dense in $\cD(N^\romL)\cap \cD(L_\beta)$ with respect to the intersection topology.
\end{Enumerate}
\end{proposition}

\begin{proof} To establish \ref{Item-NC1L}, we argue as in the
  verification of \cite[Cond.~2.1~(2), cf.~Sect.~5.5]{FaupinMoellerSkibsted2011a}.

First observe that $N^\romL$ and $L_0$ commute, such that we can compute as a form on the core $\cC^\romL$
\[
\big[(N^\romL+1)^{-1},L_\beta\big] = (N^\romL+1)^{-1}W_\beta(\coup)  - W_\beta(\coup)(N^\romL+1)^{-1}.
\] 
The right-hand side extends to a bounded operator, and since
$\cC^\romL$ is dense in $\cD(L_\beta)$, 
the form $[(N^\romL+1)^{-1},L_\beta]$ defined on $\cD(L_\beta)$
extends by continuity to a bounded form on $\cH^\romL$, coinciding
with the closure of the right-hand side as a form on $\cC^\romL$. 
Hence $N^\romL\in C^1(L_\beta)$, cf.~Definition~\ref{Def-C1A}.

Having established that $N^\romL$ is of class $C^1(L_\beta)$, we know
that $[N^\romL,L_\beta]$ 
extends from the intersection domain $\cD(N^\romL)\cap \cD(L_\beta)$
to a bounded form on $\cD(N^\romL)$. Hence, to compute this form, it suffices to compute it on a core of $N^\romL$.
Compute as a form on $\cC^\romL$
\begin{equation}\label{CommLandN}
\ri[N^\romL,L_\beta] = W_\beta(\ri \coup),
\end{equation}
which due to \eqref{WbetaBoundNorm}
extends from $\cC^\romL$ to $\cB(\cD(\sqrt{N^\romL});\cH^\romL)$. This proves \ref{Item-NC1L}.

As for \ref{Item-NLDomainInv}, we follow the proof of \cite[Lemma~5.15]{FaupinMoellerSkibsted2011a}. 
Let $T_0= L_0+\ri (N^\romL+1)$. Since $L_0$ and $N^\romL$ commute we clearly have
$\cD(T_0) = \cD(L_0)\cap \cD(N^\romL) =: \cD_0$.

We now construct $L_\beta+ \ri(N^\romL+1)$ in two different ways.
First define $\hL= L_0 + W_\beta(\coup)$ as a symmetric operator on $\cD_0$.
Then $T_0 + W_\beta(\coup)  = \hL+\ri(N^\romL+1) =: T_1$ is by  \cite[Corollary to
Thm.~X.48]{ReedSimonII1975} a closed operator on $\cD_0$.
Here we used \eqref{WbetaBoundNorm} again.
Conversely, we can use Proposition~\ref{Prop-Skibsted} 
to construct $T_2^\pm = L_\beta \pm \ri (N^\romL+1)$ 
as closed operators
on $\cD_\beta := \cD(L_\beta)\cap \cD(N^\romL)$ with $T_2^{+*} = T_2^{-}$.
Since $\cC^\romL\subset \cD_0$, we find that $\hL$ is a symmetric extension of
$L_{\beta_{|\cC^\romL}}$. 
Hence, $\cC^\romL$ being a core for $L_\beta$, 
we find that $\hL\subset L_\beta$. This implies that $T_1\subset T_2^+ =: T_2$. Since $T_2^\pm$ 
are both accretive, we find that $T_2$ generates a contraction semigroup.
It follows from the Hille-Yosida theorem, cf. \cite[Thm.~X.47a]{ReedSimonII1975},
that $\rho(T_1)\cap \rho(T_2)\neq \emptyset$. Let $z\in \rho(T_1)\cap \rho(T_2)$. Then
$(T_1-z)^{-1}\cH\subseteq \cD(T_2)$ and $(T_2-z)(T_1-z)^{-1} = \one_{\cH^\romL}$. Consequently,
$(T_1-z)^{-1} = (T_2-z)^{-1}$ and we conclude $\cD_0=\cD_\beta$. This proves~\ref{Item-NLDomainInv}.

Finally we turn to \ref{Item-IntersectCore}. From what was just proved, together with the closed graph theorem, we conclude that it
suffices to prove that $\cC^\romL$ is dense in $\cD_0$ with respect to the norm
\[
\|\psi\|_0 = \|N^\romL\psi\|_{\cH^\romL} + \|L_0 \psi\|_{\cH^\romL} + \|\psi\|_{\cH^\romL}.
\]
Since $L_0$ and $N^\romL$ commute, it suffices to show that one can
approximate $\psi\in\cD_0$ with 
$\psi = \one[N^\romL\leq n]\psi$, for some $n$.
Similarly, since $L_0$ and $N^\romL$ commute with $\Gamma_R := \one_\cK\otimes
\Gamma(\one[|k|\leq R])\otimes\one_\cK\otimes \Gamma(\one[|k|\leq R]) $, it suffices to approximate 
states $\psi$, non-zero in finitely many particle sectors, and 
satisfying $\Gamma_R\psi = \psi$, for some $R>0$. 

 Let $\{\vphi_n\}_{n\in\NN}\subset \cC^\romL$ be a sequence
with
$\|\psi-\vphi_n\|_{\cH^\romL}\to 0$ for $n\to\infty$. Let $\chi\in
C_0^\infty(\RR^3)$ satisfy $0\leq \chi\leq 1$,  
$\chi(k)=1$ for $|k|\leq R$, and $\chi(k) = 0$ for $|k|\geq R+1$.  
Then $\Gamma_\chi := \one_\cK\otimes\Gamma(\chi)\otimes\one_\cK\otimes
\Gamma(\chi)$ preserves $\cC^\romL$ and $\|\psi -
\Gamma_\chi\vphi_n\|_{\cH^\romL}\to 0$ for $n\to\infty$.

Now that both $\psi$ and $\Gamma_\chi \vphi_n$ only have finitely many
non-zero components, all supported inside a box of side length $R+1$,
one can easily verify that 
\[
\|\psi-\Gamma_\chi\vphi_n\|_0 \to 0 , \quad
\textup{for} \quad n\to\infty.
\]
This completes the proof.
\end{proof}

We remark that $L_\beta$ is presumably \emph{not} of class
$C^1(N^\romL)$, cf. Lemma~\ref{Lemma-H-is-C1N}.

\subsection{\JakPil Gluing}

We proceed to discuss a unitarily equivalent form of the Liouvillean
obtained by the so-called \JakPil gluing procedure, cf.~\cite{DerezinskiJaksic2001,JaksicPillet1996a}.

But first we pass to polar coordinates on the Hamiltonian level.
Define a unitary transform $T_\roml\colon \gothh \to \tgothh_\roml:= L^2([0,\infty))\otimes L^2(S^2)$
by the prescription 
\[
(T_\roml f)(\omega,\Theta) = \omega f(\omega\Theta).
\]
Denote by $\tcF_\roml = \Gamma(\tgothh_\roml)$ the Fock space in polar coordinates.
The subscript $\roml$ is for later use and refers to the left component
in the tensor product $\cH^\romL=\cH\otimes\cH$.  The twiddle indicates an object represented in polar coordinates
for the Hamiltonian, and after gluing for the Liouvillean.

The coupling in polar coordinates becomes 
\[
\tcoup(\omega,\Theta) := \omega\coup(\omega\Theta)
\]
and the Hamiltonian takes the form
\[
\tH = \bigl(\one_\cK\otimes \,\Gamma(T_\roml)\bigr) H \bigl(\one_\cK\otimes\,\Gamma(T_\roml)^*\bigr) = K\otimes\one_{\tcF_\roml}
+\one_\cK\otimes\,\D\Gamma(\omega) + \phi(\tcoup),
\]
a priori as an identity on $\cK\otimes\Gamma_\fin(C_0^\infty([0,\infty)\otimes
C^\infty(S^2))$ 
and extended to $\cD(\tH)=\cD(\one_\cK\otimes\,\D\Gamma(\omega))$ by continuity.

To deal with the standard Liouvillean we similarly need a map $T_\romr \colon\gothh\to \tgothh_\romr :=
L^2((-\infty,0])\otimes L^2(S^2)$,
defined by $(T_\romr f)(\omega,\Theta) = (T_\roml f)(-\omega,\Theta)$. Put $\tcF_\romr = \Gamma(\tgothh_\romr)$.
This sets up a unitary transformation
\[
T\colon \gothh\oplus\gothh \to \tgothh :=  L^2(\RR)\otimes L^2(S^2),
\]
by the construction
\[
(T(f,g))(\omega,\Theta) = \one[\omega\geq 0](T_\roml f)(\omega,\Theta)
+ \one[\omega\leq 0](T_\romr g)(\omega,\Theta).
\]

 Using the canonical identification
$I \colon \Gamma(\gothh\oplus\gothh) \to \cF\otimes\cF$, cf. \eqref{CanonId}, we construct a unitary map
\[
\cU \colon \cH^\romL \to \tcH^\romL := \cK\otimes\cK\otimes\tcF,
\]
where $\tcF = \Gamma(L^2(\RR)\otimes L^2(S^2))$. The map $\cU$ is
defined on simple tensors by
\[
\cU(u\otimes \eta\otimes v \otimes \xi) = u\otimes v \otimes
\Gamma(T)I^*(\eta\otimes \xi)
\]
and extended to $\cH^\romL$ by linearity and continuity.
Here $u,v\in\cK$ and $\xi,\eta\in\cF$.
 As an alternative core we take
\begin{equation}\label{tCL}
\tcC^{\romL} = \cK\otimes\cK \otimes \Gamma_\fin\bigl(C_0^\infty(\RR)\otimes C^\infty(S^2)\bigr).
\end{equation}

In the new coordinate system, we can write the interaction $W_\beta(\coup)$ as a field operator as follows.
First, the zero temperature interaction is
\begin{equation}\label{tGinfty}
\tG_\infty := \one[|\omega\geq
0]\tG_\roml(\omega,\Theta) - \one[\omega\leq 0]\tcoup_\romr(-\omega,\Theta),
\end{equation}
where $\tcoup_{\roml/\romr}(\omega,\Theta) = \omega
\coup_{\roml/\romr}(\omega\Theta)$, 
cf.~\eqref{GlandGr}. With this construction we have the identity
$\cU (\phi_\roml(\coup_\roml) - \phi_\romr(\coup_\romr)) \cU^*  = \phi(\tcoup_\infty)$. The computation is easily done on
$\tcC^\romL$ and extended by continuity to $\cD(\phi(\tcoup))$. At finite temperature, the interaction reads
\begin{equation}\label{tGbeta1}
\tcoup_\beta(\omega,\Theta) := \sqrt{1+\trho_\beta}\,\tcoup_\infty + \sqrt{\trho_\beta}\,\tcoup_{\infty,\cR}^*,
\end{equation}
where $\tcoup_{\infty,\cR}(\omega,\Theta) = \tcoup_\infty(-\omega,\Theta)$ is the
reflected glued coupling, and 
\[
\trho_\beta(\omega,\Theta) =
\rho_\beta(\omega\Theta)= \frac{1}{\e^{\beta|\omega|}-1}.
\]
Recalling \eqref{GlandGrbeta}, we observe that we similarly have 
\[
\cU W_\beta(\coup)\cU ^* = \cU (\phi_\roml(\coupbl) - \phi_\romr(\coupbr)) \cU^*  = \phi(\tcoup_\beta).
\]

As observed in \cite{JaksicPillet1996a} we have the following alternative representation of
$\tcoup_\beta$
\begin{equation}\label{tGbeta2}
\tcoup_\beta = \left(\frac{\omega}{1-\e^{-\beta\omega}}\right)^{\frac12} \hcoup_\roml - 
\left(\frac{\omega}{\e^{\beta\omega}-1}\right)^{\frac12} \hcoup_\romr,
\end{equation}
where
\begin{equation}\label{tGbeta3}
\begin{aligned}
& \hcoup_\roml(\omega,\Theta) = \bigl(\one[\omega\geq 0]\sqrt{\omega}\,\coup(\omega\Theta) + \one[\omega\leq 0]\sqrt{-\omega}\,\coup(-\omega\Theta)^*\bigr)\otimes\one_\cK\\
& \hcoup_\romr(\omega,\Theta) = \one_\cK\otimes \bigl(\one[\omega\geq 0]\sqrt{\omega}\,\overline{\coup(\omega\Theta)} + \one[\omega\leq 0]\sqrt{-\omega}\,\overline{\coup(-\omega\Theta)}^*\bigr).
\end{aligned}
\end{equation}
This form of the interaction mirrors the Araki-Woods representation \eqref{PF-AW-CCR}.

\begin{remark}\label{Rem-DerIsL2}
The representation \eqref{tGbeta1} allows us to easily observe that under the assumption \LGCond{n},
the ultraviolet part of $\tG_\beta$ and its first $n$ derivatives are square integrable,
whereas \eqref{tGbeta2} allows us to conclude the same for the infrared region.
This sets up an application of Lemma~\ref{Lemma-Field-H-Bounds}~\ref{Item-field-H-N-Bounds},
which holds true for any one-particle space $\gothh$, not just $L^2(\RR^3)$.
\remarkQED\end{remark}

We can now write down the standard Liouvillean in the new coordinate system as
\[
\tL_\beta = L_\romp\otimes \one_{\tcF} +\one_{\cK\otimes\cK}\otimes \tH_{\ph} + \phi(\tcoup_\beta),
\]
with $\tH_\ph = \D\Gamma(\omega)$. 
Note that $\omega$ denotes both a real number and multiplication by the identity function in $L^2(\RR)$.
Here $L_\romp$ is the standard Liouvillean 
for the small quantum system, cf.~\eqref{Lp}.

Again, by Nelson's commutator theorem, $\tL_\beta$ is essentially self-adjoint on $\tcC^\romL$.
We observe that $L_\beta$ and $\tL_\beta$ are unitarily equivalent through $\cU$.
As an identity on $\tcC^\romL$ we have $\tL_\beta = \tL_\infty +
\phi(\tcoup_\beta-\tcoup_\infty)$,
with $\tL_\infty = \tL_0 + \phi(\tcoup_\infty)$ and $\tL_0 =
L_\romp\otimes\one_{\tcF} + \one_{\cK\otimes\cK}\otimes\tH_{\ph}$. These operators are
also essentially self-adjoint on $\tcC^\romL $ and their closures are unitarily
equivalent with the appropriate untwiddled objects.

In the glued coordinate system we write $\tN = \cU N^\romL \cU^* =
\D\Gamma(\one_\tgothh)$, where the second quantization is here performed in $\tcF$.

The statements \ref{Item-NC1L-t} and \ref{Item-NLDomainInv-t} in the
following corollary to Proposition~\ref{Prop-BasicLReg} are an immediate consequence of
Proposition~\ref{Prop-BasicLReg}~\ref{Item-NC1L} and~\ref{Item-NLDomainInv}.
The item  \ref{Item-IntersectCore-t} however is not, but it can be proved by
an argument identical to the one employed at the end of the proof above.

\begin{corollary}\label{Cor-BasicLReg} Suppose \LGCond{0}. The following holds
\begin{Enumerate}
\item\label{Item-NC1L-t} $\tN\in C^1_\Mo(\tL_\beta)$ and the operator
  $[\tN,\tL_\beta]^\circ$ extends from $\cD(\tN)$ 
by continuity to an element of $\cB(\cD(\sqrt{\tN});\tcH^\romL)$.
\item\label{Item-NLDomainInv-t} $\cD(\tN)\cap \cD(\tL_\beta)$ does not depend on $\beta$, nor on $\coup$.
\item\label{Item-IntersectCore-t} $\tcC^\romL$ is dense in $\cD(\tN)\cap \cD(\tL_\beta)$ with respect to the intersection topology.
\end{Enumerate}
\end{corollary}

We remark that it is a consequence of Proposition~\ref{Prop-BasicLReg}
and the above corollary that, supposing \LGCond{0}, the resolvents
of $L_\beta$ and $\tL_\beta$ are strongly continuous in $\beta\in (0,\infty]$ and $\coup$, using the norm
$\|\coup\|_0' = \|(1+|k|^{-1/2})\coup\| $.
Indeed, it suffices to prove strong convergence on $\cD(\tN)$
where we compute
\[
(\tL_\beta(\coup)-z)^{-1} - (\tL_{\beta'}(G')- z)^{-1} = 
(\tL_{\beta'}(\coup')-z)^{-1}\phi(\tcoup_{\beta}-\tcoup'_{\beta'})(\tL_{\beta}(G)- z)^{-1}.
\]
Here we used Proposition~\ref{Prop-DomInv}~\ref{Item-ResInv-MoPlus}, 
cf.~Corollary~\ref{Cor-BasicLReg}~\ref{Item-NC1L-t}, which ensures that 
$(\tL_{\beta}(G)- z)^{-1}\colon \cD(\tN)\to \cD(\tN)$.
The result now follows by observing that
\[
\lim_{(\beta',\coup')\to (\beta,\coup)}\big\|\tcoup_{\beta}-\tcoup'_{\beta'}\big\| = 0. 
\]
Norm continuity, even in $\beta$, of the resolvent is probably false but an argument is lacking.
This means that while $\sigma(L_\infty) = \RR$, cf. \cite[Thm.~VIII.33]{ReedSimonI1980},
we cannot a priori exclude that the spectrum of $L_\beta$ could
collapse for $\beta<\infty$ to become a proper subset of $\RR$,
cf. the discussion around \cite[Thm.~VIII.24]{ReedSimonI1980}.

Instead we proceed as for the Hamiltonian, via a pull through formula:

\begin{proposition}\label{PullthroughL} Suppose \LGCond{0}. For any $\beta>0$, $z\in \CC\backslash \RR$ and
  $\psi\in \cD(\sqrt{\tN})$, we have as an $L^2(\RR\times S^2;\tcH^\romL)$-identity
\[
a(\omega,\Theta)(\tL_\beta -z)^{-1}\psi = (\tL_\beta+\omega -z)^{-1} a(\omega,\Theta)\psi -
(\tL_\beta+\omega - z)^{-1}(\tcoup_\beta(\omega,\Theta)\otimes \one_{\tcF})\psi.
\]
\end{proposition}

We omit the proof of Proposition~\ref{PullthroughL}
since it is verbatim the same as for
Proposition~\ref{PullthroughH}, keeping in mind
Corollary~\ref{Cor-BasicLReg}. In particular, cf.~Proposition~\ref{Prop-DomInv}~\ref{Item-ResInv-MoPlus}, the consequence that resolvents of
$\tL_\beta$ preserves $\cD(\tN)$, hence also $\cD(\sqrt{\tN})$, and that $\tcC^\romL$ is an
operator core for $\tL_\beta$. We are now almost ready to prove the following
theorem, which seems to be new and establishes a widely expected result as a fact.

\begin{theorem}\label{HVZL} Suppose \LGCond{0}. For any $\beta>0$, we have
  $\sigma(L_\beta)=\sigma(\tL_\beta) = \RR$.
\end{theorem}

We postpone the proof of the above theorem to Subsect.~\ref{Subsec-ExisNonexis}, where a 
missing ingredient will be introduced.
The two HVZ-type theorems, Theorems~\ref{HVZH} and~\ref{HVZL}, will
play no role in the notes apart from clarifying the general spectral picture.

\begin{remark}\label{Rem-SpinBoson} Our results on the standard Liouvillean
will mostly be proved in the \JakPil glued coordinates.
Since only $\omega$-derivatives will play a role, this allows us to 
formulate slightly weaker assumptions using $\tcoup_\beta$ instead of
$\coup$. This improvement is in general largely irrelevant, hence the present
formulation with \LGCond{n}. 

More importantly, for couplings $\coup$ on a special form, one can
-- due to the representation \eqref{tGbeta2} -- allow for interactions 
at positive temperature far more singular than what is permitted by
\LGCond{n}. To make this precise, assume $\coup$ takes the form
\begin{equation}\label{Ohmic-Form}
\coup(k) =  |k|^{-\frac12} g(k)\coup_0,
\end{equation}
where $\coup_0\in \Mat_\nu(\CC)$ is self-adjoint $\coup_0^*=\coup_0$, and $g\colon\RR^3\to
\RR$. Define
\[
\hg(\omega,\Theta) = \one[\omega\geq 0]g(\omega\Theta) +
\one[\omega\leq 0]g(-\omega\Theta).
\]
Then we can represent $\hcoup_\roml = \hg \coup_0$ and $\hcoup_\romr =
\hg \overline{\coup}_0$, cf. \eqref{tGbeta3}. Hence, we see that differentiability of
$\tcoup_\beta$ is governed by that of $\hg$. For the spin-boson model,
$g$ is a form factor (or ultraviolet cutoff), e.g. constant near $0$ or
perhaps of the form $\e^{-k^2/\Lambda^2}$ to take some popular
choices. Here $\hg$ will be constant across the singularity at
$\omega=0$, or equal to $\e^{-\omega^2/\Lambda^2}$ for the other
choice.

For models of the form \eqref{Ohmic-Form} we can reformulate replacements for
\LGCond{n}. Let $n\in\NN_0$. 
There exists $\hg\colon\RR\times S^2\to\CC$ and $\coup_0\in
\Mat_\nu(\CC)$  self-adjoint, such that $\hg$ admits $n$ distributional $\omega$
derivatives in $L_\loc^1(\RR\times S^2)$ and such that
\[
\mathrm{\mathbf{(L{\coup}n')}}\quad \begin{aligned}
& \coup(\omega\Theta) = |\omega|^{-\frac12}\,\hg(\omega,\Theta)\, \coup_0\\
&\forall \omega\Theta\in\RR^3, |\omega|\leq 1, \quad \textup{and} \quad j\leq n: 
\quad |\partial_\omega^j \hg(\omega,\Theta)|\leq C |\omega|^{n-1+\mu -j}\\ 
&\forall \omega\Theta\in\RR^3, |\omega|\geq 1, \quad \textup{and} \quad j\leq n: 
\quad |\partial_\omega^j \hg(\omega,\Theta)|\leq C |\omega|^{-1 -\delta_{j,0}-\mu}. 
\end{aligned}
\]
This type of condition was used in \cite{DerezinskiJaksic2003,FroehlichMerkli2004b}.

Finally, we remark that it was observed and utilized in \cite{FroehlichMerkli2004b},
that the \JakPil gluing is not canonical in that one can glue the
two reservoirs together at $\omega=0$, twisting one of them with a
phase. This allows one to consider $\hg$ as complex valued 
and then pick the gluing phase such that $\hg(\omega\Theta)$ and $\overline{\hg(-\omega
  \Theta})$ fit together seamlessly across $\omega=0$. In fact, one
can in this way also allow for singular behavior of the form $|k|^{1/2}$
at zero, and not just $|k|^{-1/2}$. This would require an extra twist by the
angle $\pi$ corresponding to a sign change across zero. 
\remarkQED\end{remark}

\subsection{Multiple Reservoirs}

We have made a choice, in the name of concreteness, to focus on finite
dimensional quantum systems coupled to a massless scalar field (in
three dimensions) and their thermal Liouvilleans. 

Our methods, and indeed theorems, however have validity beyond this
particular choice. We single out here the case of multiple reservoirs at
possibly different inverse temperatures $\vec{\beta}=(\beta_1,\dots,\beta_q)$. 

The easiest way to observe that the results of these notes carry over
to the case of multiple reservoirs is to replace $\gothh = L^2(\RR^3)$ by
$L^2(\RR^3\times\{1,\dots,q\})\sim \gothh^q$, $q$ being the number of
reservoirs. The dispersion becomes $\omega(k,j) =  |k|$ (or
$|k|\one_{\CC^q}$). Given $q$ couplings $\coup_1,\dotsc,\coup_q$, all
satisfying the same sets of conditions, one can construct a coupling 
for the multi-reservoir system by setting $\coup(k,j) = \coup_j(k)$. 
 
As for the standard Liouvillean, one 
should replace $\coup_{\beta,\roml/\romr}$ from \eqref{GlandGrbeta} by the functions
$\coup_{\vec{\beta},\roml/\romr}(\cdot,j)= \coup_{\beta_j,\roml/\romr}(\cdot,j)$. Similarly for $\tcoup_{\vec{\beta}}$.

Weak-coupling as well as high and low-temperature results remain valid
if all coupling, respectively temperatures, are taken into the same regime.

Only one type of result here does not extend to the case of multiple
reservoirs, and that is the existence/non-existence results for
eigenvalues of $L_\beta$ discussed in Subsect.~\ref{Subsec-ExisNonexis}, which
make critical use of the modular structure of the thermal Liouvilleans.
In fact, if two inverse temperatures are distinct, at weak coupling and
under a suitable non-triviality condition on $\coup$ one has
$\sigma_\pp(L_{\vec{\beta}})=\emptyset$, cf. \cite[Thm.~7.17]{DerezinskiJaksic2003}.
This reflects the fact that an atom coupled to multiple reservoirs at different temperatures
does not have an invariant state. Instead one should look for non-equilibrium steady states,
describing heat transport from the warm to the cold reservoir through the atomic system.
See \cite{DerezinskiJaksic2003,JaksicPillet2002b,JaksicPillet2002a,MerkliMuckSigal2007,MerkliMuckSigal2007b}.
Three of these papers consider also the so-called $C$-Liouvillean, which is not self-adjoint, and seems
to be a more natural object when considering non-equilibrium steady states.

One could also replace the thermal density $\rho_\beta$ by other
densities, and a number of our results remain valid. However, the reader doing that would have to
reformulate the condition \LGCond{n} where the $1/|k|$ singularity
of $\rho_\beta$ is built in. As for our low-temperature result, these would have to 
be translated into \myquote{small} density statements.
The reader can consult
\cite{DerezinskiJaksic2003,DerezinskiJaksicPillet2003} for discussions
of other models. 

Finally we remark that essentially what we exploit is
the \JakPil glued representation of the standard Liouvillean and,
presumably, one could rephrase everything in this abstract setup.
See also \cite[Sect.~8]{DerezinskiJaksic2003}.

\subsection{Open Problems II}

There are not that many serious problems pertaining to the material from this section.
We did mention two related conjectures regarding the standard Liouvillean, while not in itself of great interest, 
resolving them would serve to clarify the picture: 

\begin{problem} Clarify to what extend the domain of the standard Liouvillean
$L_\beta$ is $\beta$ and $\coup$ dependent.
\end{problem}

\begin{problem} Verify that, as conjectured, the resolvents of the Liouvillean are \emph{not} norm continuous
in $\beta$ and $\coup$.
\end{problem}

As a final topic, we discuss the ultraviolet singularity of the models. For e.g. the spin-boson model (with the ohmic coupling), the coupling $\coup$ goes as $1/\sqrt{|k|}$ for large momenta,
which is more singular than what we can deal with. It is well-known that the Nelson (and the polaron) model
is renormalizable, but this is due to a regularizing effect stemming from the small system, in that
the Laplacian allows for control of the ultraviolet contributions \cite{Ammari2000,LiebThomas1997,Nelson1964}. Indeed, we do not expect
that the spin-boson model has a meaningful ultraviolet limit and it should not be a relevant question since it is a model describing low energy/momentum phenomena only. Having said that, it would still be undesirable if
the choice of (a reasonable) cutoff would influence whether or not the Liouvillean 
has a unique invariant state or admits non-zero eigenvalues,
and if it does have non-zero eigenvalues, will the point spectrum, being related to energy differences, have an ultraviolet limit. 
This is an underlying, and largely unexplored, issue that will not play a role in these notes.

\newpage

\section{Bound States}

In this section we study the basic properties of bound states.
The key is the following formal computation
\[
\la\psi,\ri [H,A]\psi\ra = 0,
\]
whenever $\psi$ is a bound state for $H$ and $A$ is some auxiliary operator.
Choosing $A$ such that the commutator $\ri[H,A]$ contains a positive 
operator $N$ and a remainder controllable either by $H$ or some fractional power of $N$,
will imply - at least formally - that $\psi$ is in the form domain of $N$. 
In our case, the operator $N$ will be the number operator $N$ (or $N^\romL$ for the Liouvillean).

It turns out to be a surprisingly delicate question to
establish such a bound rigorously for the standard Liouvillean, but for the 
Hamiltonian it is fairly straightforward. 
The first argument of this type is for the Hamiltonian and is due to Skibsted \cite{Skibsted1998},
and for the Liouvillean it goes back to Fr{\"o}hlich and Merkli
\cite{FroehlichMerkli2004a}, cf. also \cite{FroehlichMerkliSigal2004}.
The result we present here for the Liouvillean improves on the theorem
of Fr{\"o}hlich and Merkli.

As a consequence of such number bounds, we will be able to establish virial theorems
for the Hamiltonian and the Liouvillean.

\subsection{Number Bounds at Zero Temperature}\label{Subsec-NumberBoundH}

As for $A$, we make the choice 
\[
A=\D\Gamma(a)
\]
with
\begin{equation}\label{aRadTrans}
a = \frac{\ri}2\left\{\frac{k}{|k|}\cdot \nabla_k + \nabla_k\cdot\frac{k}{|k|}\right\},
\end{equation}
the generator of radial translations. Note that $a$ should be viewed as the closure of
$a$ restricted to $C_0^\infty(\RR^3\backslash\{0\})$ and that $a$ is a maximally symmetric operator, 
but not self-adjoint. Since $H$ is not of class $C^1(A)$, we cannot directly make sense out of the formal computation above.

Instead, we introduce a family of regularized conjugate operators
$A_n = \D\Gamma(a_n)$ with
\[
a_n = \frac{\ri}2\left\{\frac{k}{\sqrt{|k|^2+n^{-1}}}\cdot \nabla_k + \nabla_k\cdot\frac{k}{{\sqrt{|k|^2+n^{-1}}}}\right\}.
\]
The $a_n$'s, constructed as closures from $C_0^\infty(\RR^3)$, are self-adjoint
and $H\in C^1(A_n)$ for all $n$, provided \HGCond{1} is assumed. This construction goes back to Skibsted \cite{Skibsted1998} 
and was used also in \cite{GeorgescuGerardMoeller2004b}.

Let $\psi$ be a bound state for $H$, i.e. $H\psi = E\psi$ for some $E\in\RR$.
It is now a consequence of the standard Virial Theorem, cf. Theorem~\ref{Thm-Virial}, that 
$\la\psi,\ri[H,A_n]^\circ\psi\ra =0$.
Computing the commutator, we find
\[
\ri[H,A_n]^\circ = \D\Gamma\Bigl(\frac{|k|}{\sqrt{|k|^2 + n^{-1}}}\Bigr) - \phi(\ri a_n \coup).
\]
Note that assuming \HGCond{1}, we have $a_n \coup\in L^2(\RR^3;\Mat_\nu(\CC))$.
From the estimate \eqref{PhiEnergyBound}, applied with $|k|/\sqrt{|k|^2 + n^{-1}}$ in place of $|k|$,
we get
\[
\ri[H,A_n]^\circ \geq \frac12 \D\Gamma\Bigl(\frac{|k|}{\sqrt{|k|^2 + n^{-1}}}\Bigr)
- C \big\|(|k|^2 + n^{-1})^{\frac14}|k|^{-\frac12}a_n \coup\big\|^2.
\]
To check for the finiteness and uniform boundedness of the norm on the right-hand side,
we write
$a_n = \frac{k}{\sqrt{|k|^2+n^{-1}}}\cdot \ri\nabla_k + \frac{\ri}2\Div(k/ \sqrt{k^2+n^{-1}})$
and estimate
\begin{equation}\label{BoundOn-ancoup}
\big\|(|k|^2 + n^{-1})^{\frac14}|k|^{-\frac12}a_n \coup\big\|\leq \big\|\nabla G\big\| + C \big\|G/|k|\big\|,
\end{equation}
for some $n$-independent constant $C$.
Since $n\to |k|/\sqrt{|k|^2 + n^{-1}}$ is monotonously increasing towards $1$ we conclude from Lebesgue's
theorem on monotone convergence the following:

\begin{theorem}\label{Thm-NumberBoundH}
 Suppose \HGCond{1}. There exists a $C>0$, such that for any normalized bound state $\psi\in\cH$ of $H$, we have
$\psi\in\cD(\sqrt{N})$ and
\[
\big\|\sqrt{N} \psi\big\| \leq C \big(\|\nabla G\| +  \|G/|k|\|\big).
\]
\end{theorem}

Since $N$ commutes with the conjugation $\Conj$, 
we observe that the same theorem holds for bound states of $H^\conj$.

That the constant $C$ in the theorem above can be chosen uniformly in $E$, is a consequence
of $\cK$ being finite dimensional. For e.g. the confined Nelson model, this would be false
since one will need a resolvent of $K$ to bound the relevant $a\coup$. This is however a mute point, since we in
Subsect.~\ref{Subsec-EstZero} will prove that $H$ does not have high energy bound states!

In fact, if one assumes in addition \HGCond{2}, one can do better and get
$\psi\in\cD(N)$ using \cite{FaupinMoellerSkibsted2011a}. This is however a much deeper result
and will not play a role in these notes.

\subsection{Number Bounds at Positive Temperature}\label{Subsec-NumberBoundL}

 For the standard Liouvillean $L_\infty$ at zero temperature, we observe that since eigenstates are of the form
$\psi\otimes\varphi$, with $\psi,\varphi$ eigenstates of $H$ and $H^\conj$ respectively, they are --
due to Theorem~\ref{Thm-NumberBoundH} -- automatically in the domain of $\cD(\sqrt{N^\romL})$. 
Hence, eigenstates of $\tL_\infty$ are in the
domain of $\sqrt{\tN}$. The situation at positive temperature is a good
deal more subtle.

We begin with a key technical lemma, which enables us to compute commutators.
The proof follows closely a similar argument from \cite[Proof of
Cond.~2.1~(3)]{FaupinMoellerSkibsted2011a}.

Before stating the lemma, we need some notation.
Let $m\in C^\infty(\RR)$ be real-valued and  bounded with bounded
derivatives. Put 
\begin{equation}\label{ModTrans}
\ta_m = \frac{\ri}2\left\{m \frac{\D}{\D\omega} + 
\frac{\D}{\D\omega} m \right\}\otimes\one_{L^2(S^2)}
\end{equation}
and 
\begin{equation}\label{dGammaModTrans}
\tA_m =\one_{\cK\otimes\cK}\otimes\, \D\Gamma(\ta_m).
\end{equation}
We leave it to the reader to argue that 
$\ta_m$ and $\tA_m$ are essentially self-adjoint on
$C_0^\infty(\RR)\otimes C^\infty(L^2(S^2))$ and $\tcC^\romL$, respectively.
Recall from \eqref{tCL} the form of the core $\tcC^\romL$.

\begin{lemma}\label{ComputeFirstComm} Suppose \LGCond{1}. Then
\[
\bigl\la \psi, \ri [\tL_\beta,(\tA_m-z)^{-1}] \vphi \bigr\ra = - \bigl\la \psi,(\tA_m-z)^{-1}
\tL_\beta'(\tA_m-z)^{-1} \vphi\bigr\ra,
\]
for all $\psi,\vphi\in \cD(\sqrt{\tN})\cap \cD(\tL_\beta)$ and $z\in
\CC$ with $\im z\neq 0$. Here 
\[
\tL_{\beta,m}' = \one_{\cK\otimes\cK}\otimes \,\D\Gamma(m) - \phi(\ri \ta_m \tcoup_\beta)
\]
defined as a form on $\cD(\sqrt{\tN})$.
\end{lemma}

\begin{remark} We first observe that the expression in the lemma makes
  sense. Since $\tA_m$ and $\tN$ commute, $(\tA_m-z)^{-1}$ preserves the
domain of $\sqrt{\tN}$. By boundedness of $m$ and $m'$, together with Remark~\ref{Rem-DerIsL2}, we see that
$\tL_{\beta,m}'$ is well-defined as a form on $\cD(\sqrt{\tN})$.
\end{remark}

\begin{proof} By Corollary~\ref{Cor-BasicLReg} and 
Lemma~\ref{Lemma3-Prop-Mourre-C1}~\ref{Item-InterDomDense} it suffices to establish the desired
form-identity on $\cD(\tN)\cap\cD(\tL_\beta)$.

Recall from Corollary~\ref{Cor-BasicLReg} that 
$\cD(\tN)\cap\cD(\tL_\beta) =  \cD(\tN)\cap \cD(\tL_0)$. On this
domain $\tL_\beta$ can be
written as the operator sum $\tL_0+\phi(\tcoup_\beta)$. Hence it
suffices to prove that
\begin{equation}\label{NoughToProve}
\begin{aligned}
\big\la \psi, \ri [\tL_0,(\tA_m-z)^{-1}] \vphi \big\ra & = - \big\la \psi,(\tA_m-z)^{-1}
\one_{\cK\otimes\cK}\otimes \D\Gamma(m)(\tA_m-z)^{-1}\vphi \big\ra,\\
\big\la \psi, \ri [\phi(\tcoup_\beta),(\tA_m-z)^{-1}] \vphi\big\ra 
& = \big\la \psi,(\tA_m-z)^{-1}
\phi(\ri \ta_m \tcoup_\beta)(\tA_m-z)^{-1}\vphi \big\ra,
\end{aligned}
\end{equation}
for all $\psi,\vphi\in  \cD(\tN)\cap\cD(\tL_0)$ and $z\in \CC$ with
$\im z\neq 0$.
 The second identity in \eqref{NoughToProve} can easily be verified for
$\psi,\vphi\in\tcC^\romL$ from which it extends by density since
$\phi(\tcoup_\beta)$ and $\phi(\ri \ta_m \tcoup_\beta)$ are
$\sqrt{\tN}$-bounded, cf. Remark~\ref{Rem-DerIsL2}.

As for the first identity in \eqref{NoughToProve}, one should first
observe that all objects preserve particle sectors, i.e. sectors with
$\tN = n$ for some $n$. Hence, it suffices to establish the identity for
$\psi,\vphi$ being $n$-particle states. Observe that
$\D\Gamma^{(n)}(\omega)$ is of class $C^1_\Mo(\D\Gamma^{(n)}(\ta_m))$,
indeed; $\ri [\D\Gamma^{(n)}(\omega),\D\Gamma^{(n)}(\ta_m)]^\circ =
\D\Gamma^{(n)}(m)$ is a bounded operator on $\cF^{(n)}$, the
$n$-particle sector. 
Hence, the identity 
\begin{align*}
& \big\la \psi,\one_{\cK\otimes\cK}\otimes \ri [\D\Gamma^{(n)}(\omega),(\D\Gamma^{(n)}(\ta_m)-z)^{-1}] \vphi
\big\ra \\
\qquad &  = - \big\la \psi,\one_{\cK\otimes\cK}\otimes (\D\Gamma^{(n)}(\ta_m)-z)^{-1}
\D\Gamma^{(n)}(m)(\D\Gamma^{(n)}(\ta_m)-z)^{-1} \vphi\big\ra
\end{align*}
holds for $z$ with $|\im z| \geq
\sigma_{n}$ for some $\sigma_{n}$ chosen such that
$(\D\Gamma^{(n)}(\ta_m)-z)^{-1}$ preserves $\cD(\D\Gamma^{(n)}(\omega))$
inside the $n$-particle sector, cf. Lemma~\ref{Lemma1-Prop-Mourre-C1}. 
By the unique continuation theorem,
the identity then holds for all $z$ with $\im z\neq 0$.
\end{proof}

\begin{lemma}\label{LprimeC1} Suppose $\inf m(\omega)>0$. Then $\tL_{\beta,m}'$ is of class $C^1_\Mo(\tA_m)$,
with commutator $ [\tL_{\beta,m}',\tA_m]^\circ = \ri\phi(\ta_m^2\tcoup_\beta)$.
\end{lemma}

\begin{proof} We aim to use Proposition~\ref{Prop-MourreEquiv} to establish that $\tL_{\beta,m}'$ is of class $C^1_\Mo(\tA_m)$.
Note that $\tN$ and $\tA_m$ commute, and by Remark~\ref{Rem-DerIsL2} both $\phi(\ri\ta_m\tcoup_\beta)$ and
$\phi(\ta_m^2\tcoup_\beta)$ are $\sqrt{\tN}$-bounded. 
The first of the two criteria, cf. Proposition~\ref{Prop-MourreEquiv}~\ref{Item-MourreC1}~\ref{Item-bstability},
follows from the equality $\cD(\tL_{\beta,m}')= \cD(\tN)$.
Since $\tcC^\romL$ is a core for both the commuting operators $\tN$ and $\tA_m$,
it is dense in the intersection domain $\cD(\tN)\cap\cD(\tA_m)$. 
Hence, to verify \ref{Item-MourreCommBound} and the desired form of the commutator,
one simply has to verify the commutator identity in the sense of forms on $\tcC^\romL$. But this is straightforward.
\end{proof}

We are now ready to state and prove our improvement
of the Fr{\"o}hlich-Merkli number bound. For comparison, we require
one less commutator reflected in an improvement by one power of $|k|$
in the infrared behavior of $\coup$. It is however still one
commutator  more than what was needed
for the Hamiltonian. It is unclear if this is just a technical issue.

\begin{theorem}\label{NumberBoundL} Suppose \LGCond{2}. Let $\psi$ be an eigenstate of $L_\beta$. Then
  $\psi\in\cD(\sqrt{N^\romL})$.
\end{theorem}

\begin{proof} Let $\psi\in\cD(\tL_\beta)$ be an eigenstate for
  $\tL_\beta$.
It suffices to prove that $\psi\in \cD(\tN^{1/2})$. We can assume
  without loss of generality that the eigenvalue is zero,
  i.e. $\tL_\beta\psi = 0$.

Denote by $\ta = \ta_{1}$ the generator of translations, cf. \eqref{ModTrans}.
Similarly, we abbreviate  $\tA = \tA_{1} = 
\one_{\cK\otimes\cK}\otimes \,\D\Gamma(\ta)$. Note that $\tA$ commutes with $\tN$. 

Abbreviate $I_n(\tN) = n(\tN+1)^{-1}$ as in Lemma~\ref{Lemma-In}.
Since $\tN$, by Corollary~\ref{Cor-BasicLReg}, is of class $C^1_\Mo(\tL_\beta)$, we have
$I_n(\tN)\psi\in\cD(\tL_\beta)\cap \cD(\tN)$, for all $n\geq 1$.
See Lemma~\ref{Lemma-Prop-SA-C1}~\ref{Item-Cheap-DomInv}. 

Put $\psi_n = I_n(\tN) \psi$. With the choice $m=1$ we have
\[
L_{\beta,1}' = \tN - \phi(\ri \ta\tcoup_\beta),
\]
as a self-adjoint operator with domain $\cD(\tN)$.
We abbreviate $L_\beta' = L_{\beta,1}'$.
We can thus compute using Lemma~\ref{ComputeFirstComm} for $m\in\NN$ and  $z\in\CC$, with $\im z\neq 0$,
\begin{align*}
& \big\la\psi_n,\ri [\tL_\beta,(\tA/m-z)^{-1}]\psi_n \big\ra \\
& \qquad  = - \frac1{m}\big\la \psi_n,
(A/m-z)^{-1} \tL_\beta'(\tA/m-z)^{-1}\psi_n\big\ra\\
&\qquad  =   -\frac1{m}\big\la \psi_n,
 \tL_\beta'(\tA/m-z)^{-2}\psi_n\big\ra\\ 
&\qquad \quad - \frac{\ri}{m^2} \big\la \psi_n, (\tA/m-z)^{-1}\phi(\ta^2\tcoup_\beta)(\tA/m-z)^{-2} \psi_n\big\ra.
\end{align*}
In the last equality we used Lemma~\ref{LprimeC1}. 
On the other hand, we can undo the commutator on the left-hand side and
commute $\tL_\beta$ through $I_n(\tN)$ to get
\begin{align*}
& \big\la\psi_n,\ri [\tL_\beta,(\tA/m-z)^{-1}]\psi_n\big\ra = \\
&  - \big\la \psi_n,
\big\{ \phi(\ri\tcoup_\beta)(\tN +n)^{-1} (\tA/m-z)^{-1} + (\tA/m-z)^{-1}  (\tN+n)^{-1} \phi(\ri\tcoup_\beta)\big\} \psi_n\big\ra.
\end{align*} 
Here we used Corollary~\ref{Cor-BasicLReg} and a twiddled version of \eqref{CommLandN}.

Let $g\in C^\infty(\RR)$ be identical to $t$ for $|t|\leq 1$, monotonously
increasing and constant outside a ball of radius $2$. Suppose in
addition that $\sqrt{g'}$ is smooth. We will furthermore require that
\begin{equation}\label{Monotonegprime}
\forall t\in \RR: \quad t g^\dprime(t) \leq 0.
\end{equation}
 Let $\tg$ denote an almost
analytic extension of $g$, cf.~\cite{Moeller2000}. Abbreviating $g_m(t) = m g(t/m)$, we get
\begin{align*}
& - \big\la \psi_n,
\big\{\phi(\ri\tcoup_\beta) (\tN +n)^{-1} g_m(\tA) + g_m(\tA) (\tN+n)^{-1} \phi(\ri\tcoup_\beta)\big\} \psi_n\big\ra \\
& \quad  =  
\big\la\psi_n, \tL_\beta' g_m'(\tA)\psi_n\big\ra\\
&\quad \quad -\frac{\ri}{m\pi}\int_\CC \bar{\partial} \tg(z) \big\la \psi_n, (\tA/m-z)^{-1}
 \phi( \ta^2\tcoup_\beta) (\tA/m-z)^{-2}\psi_n\big\ra \D z.
\end{align*}
We estimate the left-hand side  to be $O(m/\sqrt{n})$ and the second term on the
right-hand side is $O(\sqrt{n}/m)$, cf. Remark~\ref{Rem-DerIsL2}.  Hence, we arrive at the estimate
\begin{equation}\label{LNumberBoundA}
\bigl|\big\la\psi_n, \tL_\beta' g_m'(\tA)\psi_n\big\ra\bigr| \leq
C\Bigl(\frac{m}{\sqrt{n}} + \frac{\sqrt{n}}{m} \Bigr),
\end{equation}
for some $C>0$. Put $h(t) = \sqrt{g'(t)}$. Then $h\in C_0^\infty(\RR)$
and defining $h_m(t) = h(t/m)$ we find $h_m(\tA)^2 = g_m'(\tA)$.
Let $\tih$ be an almost analytic extension of $h$. Then
\begin{align*}
\big\la\psi_n, \tL_\beta' g_m'(\tA)\psi_n\big\ra & = \big\la \psi_n, h_m(\tA) \tL_\beta'
h_m(\tA)\psi_n\big\ra + \big\la\psi_n,[\tL_\beta',h_m(\tA)]h_m(\tA)\psi_n\big\ra.
\end{align*}
Observe that $[h_m(\tA),\tL_\beta'] (\tN + 1)^{-1/2} =
[h_m(\tA), \phi(\ri \ta \tcoup_\beta)](\tN+1)^{-1/2}$ is of
the order $1/m$, and therefore 
\begin{equation}\label{LNumberBoundB}
\big\la\psi_n, \tL_\beta' g_m'(\tA)\psi_n\big\ra = \big\la \psi_n, h_m(\tA) \tL_\beta'
h_m(\tA)\psi_n\big\ra + O\Bigl(\frac{\sqrt{n}}{m}\Bigr).
\end{equation}
Finally, using that $\tL_\beta'\geq \tN/2 - C'$, for some $C'>0$,
we get from \eqref{LNumberBoundA} and \eqref{LNumberBoundB} that
\[
\big\la\psi_n, h_m(\tA) \tN h_m(\tA)\psi_n\big\ra \leq
C\Bigl(\frac{m}{\sqrt{n}} + \frac{\sqrt{n}}{m}\Bigr),
\]
for some $C>0$ and all $n,m\geq 1$. 

We now pick $n = m^2$, such that we obtain the bound
\[
\big\la\psi, h_m(\tA) I_{m^2}(\tN)^2\tN h_m(\tA)\psi\big\ra \leq 2C,
\]
uniformly in $m$.
 
Let $E^{(\tN,\tA)}$ be the joint spectral resolution on $\NN_0\times\RR$,
induced by the two commuting operators $\tN$ and $\tA$. 
Then
\[
\big\la\psi, h_m(\tA) I_{m^2}(\tN)^2\tN h_m(\tA)\psi\big\ra  = \int_{\NN_0\times \RR}
h_m(t)^2 \frac{n m^4}{(n + m^2)^2}\, \D E^{(\tN,\tA)}_\psi(n,t).
\]
Since $h_m(t)^2 \frac{n m^4}{(n + m^2)^2}\to n$ monotonously, as $m\to\infty$,
we conclude using the monotone convergence theorem  
that $\int_{\NN_0\times\RR} n \,\D E^{(\tN,\tA)}_\psi(n,t) <\infty$. 
Here we used \eqref{Monotonegprime} to ensure that $m\to h_m(t)$ is
monotonously increasing towards $1$.  
Being a joint spectral resolution, we have 
\[
\int_{\NN_0\times\RR} n \,\D E^{(\tN,\tA)}_\psi(n,t) = \int_{\NN_0} n \,
\D
E^{\tN}_\psi(n),
\]
where $E^\tN$ is the spectral resolution for $\tN$. Hence
$\la \psi, \tN\psi\ra <\infty$ and we are done.
\end{proof}

\subsection{Virial Theorems}

Having established the number bounds, we can now formulate and prove
two virial theorems.

Let $m\in C^1(\RR)$ be real-valued and bounded, with bounded derivative as in the previous subsection.
Given such a function $m$, we can construct a maximally symmetric operator by the prescription
\[
a_m = \frac{\ri}2\left\{\frac{m(|k|)k}{|k|}\cdot \nabla_k + \nabla_k\cdot \frac{m(|k|)k}{|k|}\right\},
\]
extended by continuity from its core $C_0^\infty(\RR^3\backslash\{0\})$. 
This gives rise to a maximally symmetric operator $A_m =
\D\Gamma(a_m)$ on $\cH$.  We write, supposing \HGCond{1},
\begin{equation}\label{Hmprime}
H'_m = \D\Gamma(m(|k|)) - \phi(\ri a_m\coup),
\end{equation}
as a form on $\cD(\sqrt{N})$, cf. \eqref{PhiNumberBound}. If $\inf m(\omega)>0$, then $H_m'$ is self-adjoint on
$\cD(N)$.  The form $H_m'$ formally equals the commutator $\ri [H,A_m]$ and we have

\begin{theorem}\label{Thm-VirialH} Suppose \HGCond{1}. Let $\psi\in\cH$ be a bound state for the Hamiltonian $H$.
Then $\la \psi, H'_m\psi\ra= 0$.
\end{theorem}

\begin{proof}
Note that the expectation value is meaningful due to the number bound in Theorem~\ref{Thm-NumberBoundH}.

Replace $m$ by a regularizing function $m_n(r) = m(r)r/\sqrt{r^2+n^{-1}}$, as in the proof of the number bound in Subsect.~\ref{Subsec-NumberBoundH}. Then the associated $a_{m_n}$ is self-adjoint and so is $A_{m_n} = \D\Gamma(a_{m_n})$. Furthermore,  $H$ is of class $C^1(A_{m_n})$ for all $n$. We can compute the commutator
$\ri [H,A_{m_n}]^\circ = \D\Gamma(m_n(|k|)) - \phi(\ri a_{m_n}\coup)$.
By the usual virial theorem, cf.~Theorem~\ref{Thm-Virial}, together with \eqref{BoundOn-ancoup}, Theorem~\ref{Thm-NumberBoundH} and
Lebesgue's dominated convergence theorem, we conclude the proof.
\end{proof}

To deal with the standard Liouvillean, we use the observables $\ta_m$ and $\tA_m$ from \eqref{ModTrans}
and~\eqref{dGammaModTrans}.
Write, supposing now \LGCond{1},
\begin{equation}\label{tLbetamprime}
\tL'_{\beta,m} = \one_{\cK\otimes\cK}\otimes\,\D\Gamma(m(\omega))  - \phi(\ri \ta_m\tcoup_\beta),
\end{equation}
which by Remark~\ref{Rem-DerIsL2} is a well-defined form on $\cD(\sqrt{\tN})$. Again, if $\inf m(\omega)>0$, then
$\tL_{\beta,m}'$ is self-adjoint on $\cD(\tN)$. Define 
\begin{equation}\label{Lbetamprime}
L'_{\beta,m} = \cU^* \tL'_{\beta,m}\,  \cU ,
\end{equation}
which is a well-defined form on $\cD(\sqrt{N^\romL})$. Under the
\LGCond{1} assumption, we can compute
\[
L_{\beta,m}' = \one_{\cK\otimes\cK}\otimes\,\big(\D\Gamma(m(|k|)\otimes \one_\cF +
\one_\cF\otimes\,\D\Gamma(m(|k|)\big) -
\phi_\roml(\ri a_m \coupbl) + \phi_\romr(\ri a_m\coupbr).
\]
We warn the reader that if one follows Remark~\ref{Rem-SpinBoson} and
imposes an \LGCondp{1} assumption instead of \LGCond{1}, then
$a_m\coup_{\beta,\roml/\romr}$ may not be well-defined. 

\begin{theorem}\label{Thm-VirialL} Suppose \LGCond{2}. Let
  $\psi\in\cH^\romL$ be a bound state 
for the standard Liouvillean $L_\beta$, at inverse temperature
$0<\beta\leq \infty$.
Then $\la \psi, L'_{\beta,m}\psi\ra= 0$.
\end{theorem}

\begin{proof} First of all, we note that the expectation value is meaningful due to Theorem~\ref{NumberBoundL}.
Secondly, it suffices to prove the theorem in the glued coordinates,
where $L_{\beta,m}'$ is replaced by $\tL_{\beta,m}'$ and $\psi$ is an
eigenstate for $\tL_\beta$.

Let $\psi$ be a bound state for $\tL_\beta$.
Using the notation $I_n(\tA_m) = \ri n (\tA_m + \ri n)^{-1}$, cf. Lemma~\ref{Lemma-In},
we write 
\[
B_n = \tA_m I_n(\tA_m) = \ri n\one_{\tcH^\romL} + n^2(\tA_m + \ri n)^{-1}.
\]
 Then $B_n$ is bounded for all $n$.
 We compute using Lemma~\ref{ComputeFirstComm}
as a form on $\cD(\tL_\beta)\cap \cD(\sqrt{\tN})$:
\[
0 = \big\la\psi,\ri [\tL_\beta,B_n]\psi\big\ra =  \big\la\psi,I_n(\tA_m)\tL'_{\beta,m}I_n(\tA_m)\psi\big\ra.
\]
Since $\tN$ commutes with $I_n(\tA_m)$, we can - keeping Theorem~\ref{NumberBoundL} in mind - take the limit
$n\to\infty$, using \eqref{Reg-ToId} and \eqref{Reg-ToA}, and conclude the theorem. 
\end{proof}

The theorem of course remains true if we pass to the \JakPil glued operator $\tL_\beta$.
While the proof given above is at least formally identical to a standard proof of the usual virial theorem,
cf. the proof of Theorem~\ref{Thm-Virial}, the reader should keep in mind that it relies on the non-trivial Lemma~\ref{ComputeFirstComm} and Theorem~\ref{NumberBoundL}.

The virial theorem's are the tools that will allow us to deduce statements about
non-existence, local finiteness and finite multiplicity for eigenvalues, given
a so-called positive commutator estimate. This is the subject of Sect.~\ref{Sec-CommEst}.

\subsection{A Review of Existence and Non-existence Results}\label{Subsec-ExisNonexis}

The first theorem we highlight is due to G{\'e}rard \cite[Thm.~1]{Gerard2000} and establishes existence of a
ground state for the Hamiltonian $H$ under an \HGCond{1} condition.
Subsequently some improvements appeared in \cite{BruneauDerezinski2004,Ohkubo2009}.

\begin{theorem}\label{Thm-Gerard} Assume \HGCond{1}. Then the bottom of the spectrum
  $\Sigma$ of $H$ is an eigenvalue. 
\end{theorem}

In a somewhat surprising recent development Hasler and Herbst  proved  
that the Spin-Boson model, cf.~Remark~\ref{Rem-SpinBoson},
admits a ground state if the coupling is sufficiently weak \cite{HaslerHerbst2011a}. They used
the renormalization group method of Bach, Fr{\"o}hlich and Sigal \cite{BachFroehlichSigal1999}, which also yielded analyticity of the 
ground state energy as a function of the (small) coupling constant.
Analyticity had previously been established by Griesemer and Hasler \cite{GriesemerHasler2009},  
under a stronger \HGCond{1} assumption. Abdesselam subsequently gave a new powerful proof of analyticity  of the ground state energy, using cluster expansion techniques \cite{Abdesselam2011}.  
See also Problems~\ref{Prob-HH1} and~\ref{Prob-HH2} in the following subsection.

The following beautiful theorem, due to Derezi{\'n}ski, Jak\v{s}i{\'c} and Pillet 
establishes the existence of a $\beta$-KMS vector, which is in
particular an eigenstate of $L_\beta$ with eigenvalue zero. See
\cite[Thm.~7.3]{DerezinskiJaksic2003} and \cite[Appendix~B]{DerezinskiJaksicPillet2003}.
This improves on an earlier result of Bach, Fr{\"o}hlich and Sigal
\cite[Thm.~IV.3]{BachFroehlichSigal2000}, who required more infrared regularity. 
For the particular case of the spin-boson model, the result goes back to \cite{FannesNachtergaleVerbeure1988}.

\begin{theorem}\label{Thm-betaKMS} Suppose \LGCond{0}. Then for any inverse temperature
  $0<\beta < \infty$, we have $\Omega^\PF_\beta\in\cD(\e^{-\beta(L_0+\phi^\PF_{\beta,\roml}(\coup))})$
  and 
  \[
    \Omega^\PF_{\beta,\coup} :=  \frac{\e^{-\beta(L_0+\phi^\PF_{\beta,\roml}(\coup))}\Omega^\PF_\beta}
    {\bigl\|\e^{-\beta(L_0+\phi^\PF_{\beta,\roml}(\coup))}\Omega^\PF_\beta\bigr\|} \in \cP^\PF_\beta\cap\ker(L_\beta),
  \]
  where $\cP^\PF_\beta$ is the standard cone \eqref{PFCone}. Furthermore, $\Omega^\PF_{\beta,\coup}$
  is a faithful $\beta$-KMS vector for the standard Pauli-Fierz Liouvillean $L_\beta$.
\end{theorem}

It is worth noting that although the above theorem mirrors G{\'e}rard's result for the
Hamiltonian, it holds true for more singular interactions. In
particular, one can not rule out a situation where $H$ has no ground
state, but $L_\beta$ has a $\beta$-KMS vector in its kernel.
Indeed, this situation actually occurs in the $\nu=1$ case.
Here the Pauli-Fierz Hamiltonian is of the type considered by Derezi{\'n}ski in \cite{Derezinski2003},
where it is referred to as a van Hove Hamiltonian. Consider
\[
\coup(k) = |k|^{-\frac12}\hg(k),
\]
with $\hg\in C_0^\infty(\RR^3)$ real-valued playing the role of an ultraviolet cutoff.
We put $\hg(0) = 1$ such that the infrared behavior is captured by $|k|^{-1/2}$.
 It satisfies \LGCond{0} needed for Theorem~\ref{Thm-betaKMS}, 
but not \HGCond{1} needed for Theorem~\ref{Thm-Gerard}.

With this coupling the Hamiltonian becomes of infrared type II, again referring to the terminology of \cite{Derezinski2003},
and does not admit a ground state. The ground state should be the coherent state $\e^{\ri \phi(\ri|k|^{-3/2}\hg)}\vacuum$,
but this is not in the Fock-space, since $|k|^{-3/2}\hg\not\in L^2(\RR^3)$.
To see what happens with the standard \myquote{van Hove} Liouvillean, we observe that for $\nu=1$ (and real $\hg$) we have 
\[
\coupbl = \coupbr = \big(\sqrt{1+\rho_\beta} - \sqrt{\rho_\beta}\big) |k|^{-\frac12}\hg.
\]
Expanding $\rho_\beta$ around $k=0$, we see that 
$\sqrt{1+\rho_\beta} - \sqrt{\rho_\beta} \sim \sqrt{\beta|k|}/2$. Hence
\begin{equation}\label{vanHoveCoupAtZero}
\coup_{\beta,\roml/\romr} \sim \frac{\sqrt{\beta}}{2}
\end{equation}
at $k=0$. Hence, we can diagonalize the Liouvillean with a tensor product of Weyl operators as follows.
Put
\[
V = \e^{\ri\phi(\ri |k|^{-1}\coupbl)}\otimes  \e^{\ri\phi(\ri |k|^{-1}\coupbr)},
\]
which due to \eqref{vanHoveCoupAtZero} is a well-defined unitary operator. Then
$V^* L_\beta V = L_0$ and $V\dvacuum$ is the only eigenstate, and in particular
the $\beta$-KMS state. Note that the energy shift one gets for the Hamiltonian
does not occur here, since the shift from the left and right components cancel 
each other out. 

\begin{remark} In order to invoke the general results of
  \cite{DerezinskiJaksicPillet2003},
one must establish first that $\phi_{\beta,\roml}(G)$ is a
perturbation affiliated with the von Neumann algebra
$\gothW_{\beta,\roml}$. The key is to observe that by the Trotter
product formula \cite[Thm.~VIII.31]{ReedSimonI1980}:
\[
\slim_{n\to\infty}\bigl(\e^{\ri\phi_{\beta,\roml}(\coup_1)/n}\e^{\ri\phi_{\beta,\roml}(\coup_2)/n}\bigr)^n
= \e^{\ri\phi_{\beta,\roml}(G_1+G_2)},
\]
such that the set of $\coup$'s for which
$\e^{\ri\phi_{\beta,\roml}(\coup)}$ is affiliated with
$\gothW_{\beta,\roml}$ is a  vector space over $\CC$.

If $\coup$ is of the form $\coup(k) = \coup_0 g(k)$, with $g\in L^2(\RR^\nu)$
and $\coup_0\in \Mat_{\nu}(\CC)$ self-adjoint, then one can find a
basis $u_1,\dots,u_\nu$ for $\CC^\nu$, consisting of eigenvectors with
$\coup_0 u_\ell = \lambda_\ell u_\ell$ and $\lambda_\ell\in\RR$. 
Since, for such $\coup$, we have $\phi_{\beta,\roml}(\coup) =
\shuf^*(\coup_0\otimes\one_\cK\otimes \phi^\AW_{\beta,\roml}(g))\shuf$,
we get 
\[
\e^{\ri \phi_{\beta,\roml}(\coup)} = \sum_{\ell=1}^\nu \shuf^*\bigl( \vert u_\ell\ra \la
u_\ell\vert \otimes\one_\cK\otimes
W^\AW_{\beta,\roml}(\lambda_\ell g)
\bigr)\shuf \, \in\, \tgothM_{\beta,\roml}\subseteq \gothM_{\beta,\roml}.
\]
Any $\coup_0 \in\Mat_\nu(\CC)$ can be written as a linear combination
of self-adjoint matrices $\coup_0 = \coup_1 + \ri \coup_2$, and hence
-- by the Trotter argument --
$\phi_{\beta,\roml}(\coup) = \phi_{\beta,\roml}(G_1 g) +
\phi_{\beta,\roml}(G_2\ri g)$, with 
$\coup(k)= \coup_0 g(k)$, is also affiliated with $\gothW_{\beta,\roml}$.

In order to deal with the general case, we take $\coup\in
L^2(\RR^\nu;\Mat_\nu(\CC))$. Assume $\{\coup_n\}\subset
L^2(\RR^\nu;\Mat_\nu(\CC))$ is a sequence converging to $\coup$ in
norm and such that $\phi_{\beta,\roml}(\coup_n)$ is affiliated with
$\gothM_{\beta,\roml}$. Such a sequence exist, since one can choose
$\coup_n$ to be simple functions, i.e. a finite linear combination of
coupling functions of the form just handled above. Now one can show
that on the analytic vectors $\psi\in \cC^\romL$, cf.~\eqref{LCore},
we have $\lim_{n\to\infty} \e^{\ri \phi_{\beta,\roml}(\coup_n)}\psi = \e^{\ri
  \phi_{\beta,\roml}(\coup)}\psi$. Hence, by density of $\cC^\romL$ in $\cH^\romL$, we get
convergence in the strong operator topology, and this concludes the proof of the perturbation
$\phi_{\beta,\roml}(\coup)$ being affiliated with $\gothW_{\beta,\roml}$. 
\remarkQED\end{remark}

 Theorem~\ref{Thm-betaKMS} is also the missing ingredient in the 

\begin{proof}[Proof of Theorem~\ref{HVZL}] It suffices to prove the theorem for $\tL_\beta$
and show that $\sigma((\tL_\beta + \ri)^{-1}) = (\RR+\ri)^{-1}$.
Fix $\lambda \in \RR\backslash \{0\}$ and put $\psi_m = m(\tN+m)^{-1}\Omega^\PF_{\beta,\coup}$, for $m\in\NN$.

 Since $\tN$ is of class $C^1_\Mo(\tL_\beta)$, cf.~Corollary~\ref{Cor-BasicLReg}~\ref{Item-NC1L-t}, we find that
$\psi_m\in\cD(\tN)\cap\cD(\tL_\beta)$. Furthermore, 
\[
\tL_\beta \psi_m = m \tL_\beta(\tN+m)^{-1}\Omega^\PF_{\beta,\coup} = 
(\tN+m)^{-1}\tphi(\ri \tG_\beta)\psi_m.
\]
Since the right-hand side is an element of $\cD(\tN)$, we can conclude that $\tL_\beta \psi_m\in\cD(\tN)$
and
\[
(\tN+1)^\frac12\tL_\beta \psi_m  =
(\tN+1)^\frac12(\tN+m)^{-1}\tphi(\ri \tG_\beta)\psi_m.
\]
We thus get
\begin{align*}
\bigl\|(\tN+1)^\frac12\tL_\beta \psi_m \bigr\| & \leq C \bigl\|m^\frac12 (\tN+1)^\frac12 (\tN+m)^{-1}\Omega^\PF_{\beta,\coup}\bigr\|\\
&\leq C \bigl\|m(N+m)^{-1}\Omega^\PF_{\beta,\coup}\bigr\|^\frac12  \bigl\|(N+1)(N+m)^{-1}\Omega^\PF_{\beta,\coup}\bigr\|^\frac12,
\end{align*}
where the last factor on the right-hand side goes to zero in the large $m$ limit.
Hence, we can pick $m$ large enough such that with $\psi = \psi_m$ 
we have
\begin{equation}\label{StrongHVZBound}
\bigl\|(\tN+1)^\frac12\tL_\beta\psi\bigr\| \leq\epsilon.
\end{equation}

Choose $h_n\in C_0^\infty(\RR)$ real-valued, with $\|h\| = 1$ ($L^2$-norm)
and $\supp(h)\subseteq [-1,1]$. Put $h_n(\omega,\Theta) =h_n(\omega) =  n^{1/2}h(n(\omega-\lambda))$. 
Define $\psi_n = a^*(h_n)\psi$.
Using now the pullthrough formula from Proposition~\ref{PullthroughL}, cf.~the proof of 
Theorem~\ref{HVZH},
we get for $\varphi\in \cD(\tN)$
\begin{align*}
& \bigl\la\vphi,\bigl((\tL_\beta +\ri)^{-1}-(\lambda + \ri)^{-1}\bigr)\psi_n \bigr\ra\\
& = \int_{\RR\times S^2} h_n(\omega)\bigl\la \bigl((L_\beta+\omega - \ri)^{-1}-(\lambda - \ri)^{-1}\bigr) a(\omega,\Theta)\vphi,\psi\bigr\ra\,\D \omega\D\Theta\\
&\quad 
- \int_{\RR\times S^2} \frac{h_n(\omega)}{\sqrt{2}} \big\la \tcoup_\beta(\omega,\Theta)\otimes\one_{\tcF}\vphi, 
(\tL_\beta+\omega + \ri)^{-1}\psi\bigr\ra\, \D \omega\D\Theta.
\end{align*}
Inserting $\one= (\tN+2)^{1/2}(\tN+2)^{-1/2}$ to handle the annihilation operator,
we estimate the first integrand on the right-hand side
\begin{align*}
\bigl|h_n(\omega) & \bigl\la \bigl((\tL_\beta+\omega - \ri)^{-1}-(\lambda - \ri)^{-1}\bigr) a(k)\vphi,\psi\bigr\ra\big|\\
&\leq |h_n(\omega)| \bigl\| (\tN+2)^{-\frac12}a(k)\vphi\bigr\|\\
&\quad \times
\bigl\|(\tN+2)^{\frac12}(\tL_\beta+\omega + \ri)^{-1}\bigl(\tL_\beta + (\omega-\lambda)\bigr)\psi\bigr\|\\
&\leq C |h_n(\omega)|\bigl\|(\tN+1)^\frac12\bigl(\tL_\beta + (\omega-\lambda)\bigr)\psi\bigr\|\bigl\|a(k)(\tN+1)^{-\frac12}\vphi\bigr\|\\
& \leq  C |h_n(\omega)| \Bigl(\epsilon + \frac1{n}\Bigr) \bigl\|a(k)(\tN+1)^{-\frac12}\vphi\bigr\|.
\end{align*}
Here we used \eqref{StrongHVZBound} in the last step,
and in  the second step we employed  Proposition~\ref{Prop-DomInv}~\ref{Item-ResInv-MoPlus},
to move the root of the number operator through $(\tL_\beta + \omega+\ri)^{-1}$.
We now conclude the bound
\[
\bigl| \big\la \vphi,\bigl((\tL_\beta+\ri)^{-1}-(\lambda+\ri)^{-1}\bigr)\psi_n\bigr\ra\bigr|
\leq C(\epsilon+o(1))\|\vphi\|.
\]
As in the proof of Theorem~\ref{HVZH}, it remains to observe that $\|\psi_n\|$ is bounded from below uniformly in $n$.
\end{proof}

The final result we discuss in this subsection, is a consequence of
Theorem~\ref{Thm-betaKMS} and a theorem of Jadczyk \cite{Jadczyk1969}, which has as a
consequence that existence and simplicity of the $0$ eigenvalue for the standard
Liouvillean implies non-existence of non-zero eigenvalues! We refer
the reader to the short and very elegant paper \cite{JaksicPillet2001c}
for details, which are entirely operator algebraic in nature.

\begin{theorem}\label{Thm-Jadczyk} Suppose \LGCond{0}. Let $0<\beta<\infty$ and suppose
  that $0$ is a simple eigenvalue for $L_\beta$. Then $\sigma_\pp(L_\beta)=\{0\}$.
\end{theorem}

\subsection{Open Problems III}

As the reader may have observed, the bottleneck for applying the 
virial theorem to the standard Liouvillean is the number bound
in Theorem~\ref{NumberBoundL}, 
where we -- compared with the Hamiltonian case in Theorem~\ref{Thm-NumberBoundH} --
need much stronger assumptions.
This is in particular unfortunate, since the positive commutator estimates we establish in the following
section hold under an \LGCond{1} assumption, not the \LGCond{2} assumption needed for the number bound.

\begin{problem} Can the number bound in Theorem~\ref{NumberBoundL} be established under an
\LGCond{1} condition, or some other condition truly weaker than \LGCond{2}.
\end{problem}

The author does not know one way or the other what the answer may be to this problem. We remark that, although
the number bound is a bottleneck viz a viz the structure of the point spectrum, the \LGCond{2} condition
is what one would expect for a limiting absorption principle to hold, given a positive commutator estimate.
Hence, from a broader perspective, the \LGCond{2} condition will appear anyway.

The proof of the number bound in Theorem~\ref{NumberBoundL} did not make
essential use of the small system being finite dimensional. Hence, we expect the theorem
to remain true also for confined small systems, like the standard Liouvillean for the confined Nelson model.

\begin{problem} 
 Extend Theorem~\ref{NumberBoundL} to the case where the small system $\cK$ is not necessarily finite dimensional.
\end{problem}

As mentioned in Subsect.~\ref{Subsec-ExisNonexis}, Hasler and Herbst established in \cite{HaslerHerbst2011a}
the existence of an interacting ground state for the spin-boson model
with physical infrared singularity $|k|^{-1/2}$, provided the coupling is
sufficiently weak. This result came as a complete surprise to the author,
since it is contrary to the solvable model with $\cK=\CC$ and the
confined Nelson model
\cite{BruneauDerezinski2004,Derezinski2003,Hirokawa2006,LorincziMinlosSpohn2002,Panati2009}. 
Furthermore, it goes beyond what was considered
the natural borderline established in \cite{Gerard2000}, cf. also
\cite{AraiHirokawaHiroshima2007,BruneauDerezinski2004,Ohkubo2009}. 
In fact, there has been speculation that gauge invariance
of the minimally coupled model was responsible for the existence
result of Griesemer-Lieb-Loss 
\cite{GriesemerLiebLoss2001,LiebLoss2003}, something that was however
debunked by Hasler-Herbst \cite{HaslerHerbst2011c,HaslerHerbst2012}, 
who proved that existence of a ground state, at weak coupling, remains true even after
dropping the quadratic term in the minimally coupled model thus breaking gauge invariance.

The $|k|^{-1/2}$ infrared behavior of $\coup$ is sometimes called the ``ohmic case'', a
terminology we use below.

\begin{problem}\label{Prob-HH1} Does there exists a critical coupling at which the
  ground state seize to exist for the spin-boson model considered by
  Hasler and Herbst? Or does a
  ground state exist for all couplings?
\end{problem}

\begin{problem}\label{Prob-HH2}
 Characterize the properties of ohmic $\coup$
  that ensures existence of a ground state for $H$ in the weak
  coupling regime. As a simpler problem, consider $\coup$'s of the form
  $\coup(k) = |k|^{-1/2}\hg(k) G_0$ as
  discussed in Remark~\ref{Rem-SpinBoson}.
\end{problem}

For the thermal standard Liouvillean, one has existence of a $\beta$-KMS vector in the kernel of $L_\beta$
at all values of $\beta$, cf. Theorem~\ref{Thm-betaKMS}, 
and furthermore the modular structure ensures that a simple $0$-eigenvalue implies
absence of non-zero eigenvalues, cf. Theorem~\ref{Thm-Jadczyk}. These results 
were derived from the underlying algebraic structure of standard Liouvilleans, 
and may not have natural operator theoretic proofs. It would be 
natural to ask if it is not possible to extract even more information
from the underlying algebraic framework.

\begin{problem}\label{MainProblem} Can one exploit the underlying algebraic structure to infer more
information on the point spectrum and pertaining eigenstates, than what is afforded by Theorem~\ref{Thm-Jadczyk}?
In particular, can one use algebraic arguments to conclude that zero
is in fact a simple eigenvalue of $L_\beta$?
\end{problem}

It is well known that establishing instability -- or outright absence -- of embedded eigenvalues away from
zero coupling, or some other explicitly solvable regime, is a daunting
task. It is for example not known if embedded
(necessarily negative) eigenvalues of $N$-body Schr{\"o}dinger
operators are unstable under perturbations of pair-potentials. One can
only show generic instability under perturbations by external potentials,
cf.~\cite{AgmonHerbstSasane2011,AgmonHerbstSkibsted1989}. In
\cite{FaupinMoellerSkibsted2011b} a Fermi Golden Rule was established
at arbitrary coupling for the Hamiltonian, but to conclude instability
one needs better control of eigenstates beyond the ground state (where
Perron-Frobenius theory applies). 
The case of perturbation around zero coupling is far better understood 
\cite{BachFroehlichSigal1999,BachFroehlichSigal2000,BachFroehlichSigalSoffer1999,
DerezinskiJaksic2001,DerezinskiJaksic2003,
FroehlichMerkli2004b,Golenia2009,Merkli2001}. 
Hence, whether or not the kernel of the standard Liouvillean is
generically one-dimensional beyond the weak-coupling regime, is not 
a question one is likely to answer using perturbation theory of embedded eigenvalues only.

There is perhaps an unexplored avenue available for investigating Problem~\ref{MainProblem},
which we now discuss. Let $T = \e^{-L_\beta^2/2}$. Then $T$ is a bounded self-adjoint operator, with
an eigenvalue sitting at the top of its spectrum. Appealing to the Fourier transform, we observe that
\[
T = (2\pi)^{-\frac12} \int_\RR \e^{-s^2/2}\e^{\ri s L_\beta}\,\D s.
\]
Since $L_\beta$ is a standard Liouvillean, $\e^{\ri s L_\beta}$ preserves the self-dual 
standard cone $\cP^\PF_\beta$, cf.~\eqref{PFCone}. Hence $T$ is positivity preserving, and it thus seems natural 
to apply Perron-Frobenius theory in an abstract form. There is however an obstacle to this approach, namely
the fact that the order on $\cH^\romL$ induced by the standard cone is not a lattice.
Hence, one can not map the problem into a function space, where the standard cone goes into
the positive functions. See \cite{Penney1976}. If that had been possible one could have made use of
ergodicity arguments as in \cite{HaslerHerbst2011a}.

We stress that we consider Problem~\ref{MainProblem} to be the most
important problem highlighted in these notes. 
The reason being that, due to Theorem~\ref{Thm-Jadczyk}, it reduces the question of
establishing mixing a.k.a. \myquote{return 
to equilibrium} beyond the weak coupling regime,
to positive commutator estimates and limiting absorption principles. 
Something we see no fundamental obstacle to obtaining, 
although the picture is not yet entirely clear beyond small temperatures. See Sect.~\ref{Sec-LAP}.

Finally, it would be natural, in the spirit of \cite{Derezinski2003}, to
investigate the types of ultraviolet and infrared behavior of the standard Liouvillean
when $\nu=1$, which is a solvable case. See also the discussion on ground 
states versus $\beta$-KMS states when $\nu=1$ in the previous subsection, which indicates that the infrared  
type  II property, cf. \cite{Derezinski2003}, 
characterizes existence of $\beta$-KMS states. 

\begin{problem}
Classify possible types of ultraviolet and infrared behavior of the ``van Hove Liouvillean'', i.e. when $\nu=1$.
\end{problem}

\newpage

\section{Commutator Estimates}\label{Sec-CommEst}

\subsection{The Weak Coupling Regime}

 The weak coupling regime is very well understood. To explore it, we replace
 $\coup$ by $\lambda\coup$, where the absolute value of $\lambda\in\RR$ is small.  In fact, obtaining
 positive commutator estimates in this regime is an easy exercise.
 Indeed, choosing $a$ to be generator of radial translation \eqref{aRadTrans},
 we get using \eqref{PhiNumberBound}
 \[
 H' = N - \lambda\phi(\ri a \coup)\geq  \frac12 N - \lambda^2\|a\coup\|^2.
 \]
 Here we abbreviated $H' = H'_1$, cf.~\eqref{Hmprime}.
 Choosing $\lambda$ such that  $\lambda^2\|a\coup\|^2\leq 1/4$ yields
 \begin{equation}\label{WeakCoupBoundH}
 H'\geq \frac14\one_\cH - \frac14 \one_\cK\otimes \vacuum \tvacuum.
 \end{equation}
  We can now prove
  
  \begin{corollary}\label{WeakCoupPPH} Suppose \HGCond{1} and let $\lambda_0 = \|a\coup\|^{-1}$. 
  For $\lambda\in (-\lambda_0,\lambda_0)$
  the pure point spectrum $\sigma_\pp(H)$ is finite and all eigenvalues have finite multiplicity.
  Indeed, $\dim\Ran(P) \leq \nu$, where $P=E(\sigma_\pp(H))$ is the projection onto the pure point subspace.
  (Here $H$ is defined with $\coup$ replaced by $\lambda\coup$.)
  \end{corollary}
  
  \begin{proof} Let $\{\psi_j\}_{j=1}^n$ be an orthonormal set of eigenvectors for $H$.
  We can use the virial theorem, Theorem~\ref{Thm-VirialH}, together with \eqref{WeakCoupBoundH}
  to estimate
  \[
  0 = \sum_{j=1}^n \la \psi_j,H'\psi_j\ra \geq \frac{n}4 - \frac{\Tr(\one_\cK\otimes \vacuum\tvacuum)}4
  = \frac{n-\nu}4.
  \]
  This implies $n\leq \nu$ and concludes the proof.
  \end{proof}
 
 Similarly for the Liouvillean, where we can again choose
 $\ta = \ri \frac{\D}{\D\omega}\otimes\one_{L^2(S^2)}$ to be the generator of translations in the 
 glued variable. Then
 \[
 \tL_\beta' = \tN^\romL - \lambda \phi(\ri\ta\tcoup_\beta) \geq \frac12  \tN^\romL - \lambda^2\big\|\ta \tcoup_\beta\big\|^2,
 \] 
 where we again abbreviated $\tL_\beta'= \tL_{\beta,1}'$, cf.~\eqref{tLbetamprime}.
 Hence, we arrive at 
 
 \begin{corollary}\label{WeakCoupPPL} Suppose \LGCond{2} and let $\lambda_0=\|\ta\tcoup\|^{-1}$. 
  For $\lambda\in (-\lambda_0,\lambda_0)$
  the pure point spectrum $\sigma_\pp(L_\beta)$ is finite and all eigenvalues have finite multiplicity.
  Indeed, $\dim\Ran(P_\beta)\leq\nu^2$, where $P_\beta = E^\beta(\sigma_\pp(L_\beta))$ 
  is the projection onto the pure point subspace.
  (Here $L_\beta$ is defined with $\coup$ replaced by $\lambda\coup$.) 
  \end{corollary} 
  
   \begin{proof} The proof is identical to the proof of
     Corollary~\ref{WeakCoupPPH}, except we make use of Theorem~\ref{Thm-VirialL}
   instead of Theorem~\ref{Thm-VirialH}. 
   \end{proof}
 
 This theorem improves on a result of Merkli \cite{Merkli2001}, due to the improvement in the number bound
 Theorem~\ref{NumberBoundL}. See also \cite{FroehlichMerkli2004b,FroehlichMerkli2004a,FroehlichMerkliSigal2004}.
 
\subsection{Conjugate Operators}

Let $\chi\in C_0^\infty(\RR)$ satisfy: $\chi(-\omega) = \chi(\omega)$,
$0\leq \chi\leq 1$, $\chi(\omega)=1$ for $|\omega|\leq 1/2$ and
$\chi(\omega)=0$ for $|\omega|>1$. We furthermore assume that $\chi'(\omega) \leq 0$ for $\omega\geq 0$.

Let $\mu>0$ be the constant used to define the class of couplings we
can treat, cf. \HGCond{n} and \LGCond{n}. We use it to construct an auxiliary function 
$d\colon (0,\infty)\to [1,\infty)$ as follows
\[
d(\omega) = \chi(\omega)\omega^{-\mu/4} + \chi(\omega/2)-\chi(\omega)
+ (1-\chi(\omega/2))\omega^{\mu/4}.
\]
We leave it to the reader to verify the following properties of $d$
\begin{enumerate}[label=\textup{\bf (d\arabic*)},ref=\textup{\bf (d\arabic*)}]
\item $(\omega-1)d'(\omega)\geq 0$.
\item $\lim_{\omega\to 0_+} d(\omega) = \lim_{\omega\to +\infty} d(\omega) = +\infty$.
\item $\exists C>0$ s.t. $|d'(\omega)|\leq C d(\omega)/\omega$ for all $\omega>0$.
\end{enumerate}
We extend $d$ to $\RR\backslash\{0\}$ by setting $d(\omega) =
d(-\omega)$ for $\omega<0$.

For a given 
\begin{equation}\label{DeltaZero}
\udelta = (\delta_0,\delta_\infty)\in  (0,1]\times [1,\infty) =: \Delta_0,
\end{equation}
 we define
a smooth positive function $m_\udelta\colon\RR\to [1,\infty)$ by
\begin{align*}
m_\udelta(\omega) & = d(\delta_0)\chi(\omega/\delta_0)\\
 & \quad +
d(\omega)\big(\chi(\omega/(2\delta_\infty))-
\chi(\omega/\delta_0)\big)\\
&\quad  +d(\delta_\infty)\big( 1- \chi(\omega/(2\delta_\infty))\big).
\end{align*}
Observe that $m_\udelta$ has compactly supported derivatives.

Our conjugate operator on the one-particle level, for the Hamiltonian
at zero temperature, is defined as the modified generator of radial translations
\[
a_\udelta = \frac{\ri}2\left\{m_\udelta(k)\frac{k}{|k|}\cdot \nabla_k +
  \nabla_k\cdot \frac{k}{|k|} m_\udelta(k)\right\}.
\]
Note that $a_\udelta$ a priori defined on $C_0^\infty(\RR^3\backslash\{0\})$ is closable,
and its closure is a maximally symmetric operator.
The conjugate operator is obtained through second quantization
\[
A_\udelta = \one_{\cK}\otimes\, \D\Gamma(a_\udelta),
\]
and is again a maximally symmetric operator, closable on $\cK\otimes \Gamma_\fin(C_0^\infty(\RR^3\backslash\{0\}))$.

To get a conjugate operator for the Liouvillean,  we do the construction
after gluing and define the modified generator of translations
\[
\ta_\udelta := \frac{\ri}2\left\{m_\udelta(\omega)\frac{\D}{\D\omega} 
+ \frac{\D}{\D\omega} m_\udelta(\omega)\right\}\otimes\one_{L^2(S^2)},
\]
which is essentially self-adjoint on $C_0^\infty(\RR)\otimes C^\infty(S^2)$.
We second quantize to obtain
\[
\tA_\udelta :=  \one_{\cK\otimes\cK}\otimes\, \D\Gamma(\ta_\udelta),
\]
which is essentially self-adjoint on $\tcC^\romL$.
Note that we have simplified the notation a bit, writing $\ta_\udelta$ and $\tA_\udelta$
 instead of
the more cumbersome $\ta_{m_\udelta}$ and $\tA_{m_\udelta}$, cf.~\eqref{ModTrans} and~\eqref{dGammaModTrans}. As for the \myquote{commutators}, we similarly write
$H'_\udelta$, $L'_{\beta,\udelta}$ and $\tL_{\beta,\udelta}'$ instead of $H'_{m_\udelta}$,
$L'_{\beta,m_\udelta}$ and $\tL'_{\beta,m_\udelta}$, cf.~\eqref{Hmprime},~\eqref{Lbetamprime} and~\eqref{tLbetamprime}.

As an identity on $C_0^\infty(\RR\backslash\{0\})\otimes\, C^\infty(S^2)$, we have
$T(a_\udelta\otimes\one_\gothh - \one_\gothh\otimes a_\udelta )T^* = \ta_\udelta$ and furthermore
\begin{equation}\label{TransOfComm}
\cU \tL_{\infty,\udelta}' \,\cU^* = L_{\infty,\udelta}' =  H_\udelta'\otimes\one_\cH + \one_\cH\otimes {H^\conj_\udelta}'.
\end{equation}
The latter being an operator identity on $\cD(N^\romL)$.

\subsection{Estimates at Zero Temperature}\label{Subsec-EstZero}

Throughout this section we will for $\udelta'\in\Delta_0$,
cf. \eqref{DeltaZero}, use the notation
\begin{equation}\label{Deltadelta}
\Delta(\udelta') := \bigset{\udelta\in\Delta_0}{\delta_0\leq
  \delta_0', \delta_\infty\geq \delta_\infty'}. 
\end{equation}

For $\udelta\in\Delta_0$, we write $N_\udelta$ for
$\D\Gamma(m_\udelta)$, the modified number operator
appearing in $H'_\udelta = N_\udelta - \phi(\ri a_\udelta \coup)$.
The reason for introducing the modified generator of radial
translation is that $N_\udelta$ is large in the infrared and
ultraviolet regimes, which allows us to handle very soft and very hard photons.

We shall make use of geometric localization, cf. Appendix~\ref{Subsect-GeomLoc} three times in this subsection. 
It will be used in two different forms, the first of which will appear in two different proofs, and for this reason we introduce the notation here. 
For $R>1$, we perform a partition of unity in momentum space as follows.
Let
\begin{equation}\label{MomentumIMS1}
F^R = \begin{pmatrix} \one[|k| < R] \\ \one[|k|\geq R]\end{pmatrix} \colon \gothh \to L^2(B(0,R))\oplus L^2(B(0,R)^c) 
=: \gothh_<^R\oplus\gothh_>^R
\end{equation}
and observe that $F^R$ is unitary. We lift to $\cF$, cf.~Appendix~\ref{Subsect-GeomLoc}, and get
the unitary
\begin{equation}\label{MomentumIMS2}
\cGamma(F^R) \colon \cF \to \Gamma(\gothh_<^R)\otimes\Gamma(\gothh_>^R) =: \cF_<^R\otimes\cF_>^R.
\end{equation}
Using \eqref{dGammacGammaIntertwine}, we conclude the intertwining relation
\begin{equation}\label{GeomLocMomentum}
\cGamma(F^R)\D\Gamma(m) 
= \bigl(\D\Gamma(m_{|B(0,R)})\otimes \one_{\cF^R_>}
+ \one_{\cF^R_<}\otimes \D\Gamma(m_{|\RR^3\backslash B(0,R)})\bigr) \cGamma(F^R),
\end{equation}
where $m$ is (multiplication by) an $L^2_\loc(\RR^3)$ function.
For simplicity, we write $\cGamma(F^R)$ also for $\one_\cK\otimes \cGamma(F^R)$ acting on $\cH$. 

We begin with a new high-energy estimate, which is particular to
the case of finite dimensional small systems. It will not hold,
e.g., for (confined) atomic small systems. 

\begin{theorem}\label{Thm-ME-H-LargeE} Suppose \HGCond{1}. 
Let $e>0$ be given. There exists $\delta'_\infty\geq 1$, $c>0$
and $E_0\in \RR$ such that for all
$\udelta\in \Delta((1,\delta_\infty'))$, we have
\[
H_\udelta' \geq e\one_\cH - c \one[H\leq E_0],
\]
in the sense of forms on $\cD(N)$.
\end{theorem}

\begin{proof} The first step we take is to estimate from below
\begin{equation}\label{PF0-BoundInt}
H'_\udelta\geq \frac12 N_\udelta -C\one_\cH,
\end{equation}
exploiting the $N^{1/2}$-boundedness of $\phi(\ri a_\udelta \coup)$,
cf. \eqref{PhiNumberBound},  and the
inequality $N\leq N_\udelta$. Here $C$ is some positive number.

Put $H_0^\ext = K\otimes\one + \one_\cK\otimes {H_\ph}_{|\cF_<^R}\otimes\one + \one\otimes {H_\ph}_{|\cF_>^R}$ and
abbreviate $\lambda_\rommax = \max\sigma(K)$.
We now compute for $\tE>\lambda_\rommax$ using \eqref{GeomLocMomentum}
\begin{align*}
& N_\udelta  \geq N_\udelta \one[H_0 > \tE]\\
 & = \cGamma(F^R)^* \left\{\one_\cK\otimes {N_\udelta}_{|\cF_<^R}\otimes\one_{\cF_>^R} + \one_{\cK\otimes \cF_<^R}\otimes {N_\udelta}_{|\cF_>^R}\right\}\one[H_0^\ext > \tE]\cGamma(F^R)\\
& \geq \Gamma(\one[|k|<R])N_\udelta \one[H_0>\tE] + 
m_\delta(R)\cGamma(F^R)^* \one_{\cK\otimes \cF_<^R}\otimes \bP_\Omega \one[H_0^\ext>\tE]\cGamma(F^R)\\
& \geq \frac{\tE-\lambda_\rommax}{R} \Gamma(\one[|k|<R])\one[H_0>\tE]\\
& \qquad  + m_\udelta(R)\cGamma(F^R)^*\bigl( \one_{\cK\otimes \cF_<^R}\otimes \bP_\Omega\bigr) \one[H_0^\ext>\tE]\cGamma(F^R)\\
& \geq \min\left\{\frac{\tE-\lambda_\rommax}{R},m_\udelta(R)\right\}\one[H_0>\tE].
\end{align*}
We thus get
\begin{equation}\label{PF0-BoundInt2}
\frac12 N_\udelta - C\one_\cH\geq
\frac12\min\left\{\frac{\tE-\lambda_\rommax}{R},m_\udelta(R)\right\}\one[H_0
>\tE] - C\one_\cH.
\end{equation}

To pass from $H_0$ to $H$; we estimate, recalling that $\Sigma$ denotes the bottom of the spectrum of $H$  \eqref{BottomOfSpec},
\begin{align*}
\one[H_0\leq \tE]&\leq (\tE+1)(H_0+1)^{-1}\\
&= (\tE+1)(H-\Sigma+1)^{-\frac12}\\
 & \quad \times\left\{(H-\Sigma+1)^{\frac12}(H_0+1)^{-1}(H-\Sigma+1)^{\frac12} \right\}(H-\Sigma+1)^{-\frac12}\\
& \leq (\tE+1) \tC (H-\Sigma+1)^{-1}\\
& \leq (\tE+1) \tC\one[H\leq E] + \frac{(\tE+1) \tC}{E-\Sigma+1}.
\end{align*}
Combining with \eqref{PF0-BoundInt} and \eqref{PF0-BoundInt2},  we arrive at the bound
\[
H'_\udelta\geq
\frac12\min\left\{\frac{\tE-\lambda_\rommax}{R},m_\udelta(R)\right\}\left(1-\frac{\tC(\tE+1)}{E-\Sigma+1}\right)
- C - \frac{(\tE+1)\tC}{2}\one[H\leq E].
\]
We are now in a position to pick our constants. First choose $R$ large enough such that
$d(R)/2\geq e +C +1$. Then choose $\delta'_\infty\geq R$ such that for
$\udelta\in (0,1)\times(\delta'_\infty,\infty)$, we have $m_\udelta(R) = d(R)$.
Subsequently, we fix $\tE$ large enough such that
$(\tE-\lambda_\rommax)/(2R)\geq e+C+1$.
With these choices of $R,\delta'$ and $\tE$ (in that order), we get
\[
H_\udelta'\geq (e+C+1)\left(1-\frac{\tC(\tE+1)}{E-\Sigma+1}\right)
- C - \frac{(\tE+1)\tC}{2}\one[H\leq E].
\]
Finally, we can take $E_0$ large enough, such that with $E = E_0$, the right-hand side
is bounded from below by $e\one_\cH - \frac12 \tC(\tE+1)\one[H\leq E_0]$.
\end{proof}



Before we establish a positive commutator estimate at zero temperature, we pause to formulate and prove a technical lemma. This type of lemma is a standard tool used in inductive proofs
of Mourre estimates. See, e.g., \cite{GeorgescuGerardMoeller2004b,HunzikerSigal1995,MoellerSkibsted2004}. 
Recall from \eqref{DeltaZero} and \eqref{Deltadelta}, the definition of the sets $\Delta_0$ and $\Delta(\udelta')$.

\begin{lemma}\label{Lemma-FromPointMEToUnifME} Suppose \HGCond{1}. Let $J\subset \RR$ be a compact set and $\udelta'\in \Delta_0$. Suppose that for any $E\in J$, $\udelta\in\Delta(\udelta')$ and $\epsilon>0$, there exists $\kappa>0$, $C>0$ and a compact self-adjoint operator $K$, such that the \myquote{Mourre estimate}
\begin{equation}\label{LemBasicME}
H'_\udelta\geq (1-\epsilon)\one_\cH - C\one[|H-E|\geq \kappa] - K,
\end{equation}
is satisfied in the sense of forms on $\cD(N)$.
Then, for any $\udelta\in\Delta(\udelta')$ and $\epsilon>0$, there exist $\kappa'>0$ and $C'>0$, such that as a form on $\cD(N)$
\begin{equation}\label{LemDerivedME}
 H'_\udelta \geq -\epsilon \one_\cH - C'\one[|H-E|\geq \kappa'],
\end{equation}
for all $E\in J$.
\end{lemma}


\begin{proof} Let $J$ and $\udelta'$ be fixed as in the statement of the lemma.
 In the following, $\udelta\in\Delta(\udelta')$ is arbitrary. Note that $C,C'$, $\kappa,\kappa'$
 and $K$ may depend on $\udelta$. 
First note that by the Virial Theorem, Theorem~\ref{Thm-VirialH}, the point spectrum in an open neighborhood of $J$
is locally finite and eigenvalues in $J$ have finite multiplicity.

Let $\epsilon>0$.
We begin by verifying the estimate \eqref{LemDerivedME} for a fixed $E$, at which the Mourre estimate \eqref{LemBasicME} is satisfied.

If $E\not\in\sigma_\pp(H)$ we proceed as follows:  Let $\kappa>0$, $C>0$ be the constants and $K$ the compact operator from \eqref{LemBasicME}, associated with the given $\epsilon$. Write for $0<\kappa'<\kappa$ the compact error as
$K = K\one[|H-E|<\kappa'] + K\one[|H-E|\geq\kappa']$. Pick $\kappa'$ small enough
such that $\|K\one[|H-E|\leq\kappa']\|\leq 1/2$. Then
\[
H'_\udelta\geq -\epsilon\one - \bigl(C+\frac12\|K\|^2\bigr)\one[|H-E|\geq\kappa'],
\]
where we used that $\re\{K\one[|H-E| \geq \kappa']\} \geq - \frac12\one_\cH - \frac12 \|K\|^2 \one[|H-E|\geq \kappa']$.

If on the other hand $E\in\sigma_\pp(H)$, we proceed differently. Extract a Mourre estimate \eqref{LemBasicME} with $\epsilon$ replaced by $\epsilon/5$. Denote by $\kappa_1>0$, $C_1>0$ and $K$ the associated objects.
Write $P_E$ for the finite rank orthogonal projection
on the eigenspace associated with $E$. Abbreviate $\bP_E= \one-P_E$. Since $\Ran(P_E)\subset \cD(N^{1/2})$, we can compute and use \eqref{LemBasicME}
\begin{align}\label{HprimeAtEigenvalue}
\nonumber H'_\udelta & = P_E H'_\udelta P_E + 2\re\bigl\{P_E H'_\udelta \bP_E\bigr\} + \bP_E H'_\udelta \bP_E\\
\nonumber & = 2\re\{P_E H'_\udelta \bP_E\} + \bP_E H_\udelta' \bP_E\\
\nonumber & \geq \bigl(1-\frac{\epsilon}5\bigr)\bP_E - C_1\one[|H-E|\geq \kappa_1] -\bP_E K \bP_E + 2\re\bigl\{P_E H_\udelta' \bP_E\bigr\}\\
& \geq -\frac{\epsilon}5\one_\cH - C_2\one[|H-E|\geq \kappa_2]+ 2\re\bigl\{P_E H_\udelta' \bP_E\bigr\},
\end{align}
where we used Theorem~\ref{Thm-VirialH} in the second equality.
In the last step, we argued as above to get rid of the compact error $-\bP_E K \bP_E$ by passing to a smaller (positive) $\kappa_2<\kappa_1$.

As for the cross term in \eqref{HprimeAtEigenvalue},
we write $H_\udelta' = N_\udelta - \phi(\ri a_\udelta \coup)$ as a form sum on $D(N^{1/2})$.
Recalling Theorem~\ref{Thm-NumberBoundH}, we decompose for an $r>0$ to be fixed later
\begin{equation}\label{MissLabel0}
P_E H'_\udelta\bP_E = P_E N_\udelta \one[N_\udelta>r]\bP_E + \tK\bP_E,
\end{equation}
with $\tK= P_E N_\udelta\one[N_\udelta\leq r]-P_E \phi(\ri a_\udelta \coup)$ being compact.
Estimate first for $\sigma>0$
\begin{align*}
2\re\bigl\{P_E N_\udelta \one[N_\udelta>r]\bP_E\bigr\} &= 2\re\bigl\{P_E N_\udelta \one[N_\delta>r]\bigr\}-2P_E  N_\udelta \one_{N_\udelta>r}P_E\\
&\geq -\sigma N_\udelta  -(2+\sigma^{-1}) P_E  N_\udelta \one[N_\udelta>r]P_E.
\end{align*}
Fix $\sigma$ small enough such that
\begin{equation}\label{ChoiceOfSigma}
2\sigma \|a_\udelta \coup\|^2\leq \frac{\epsilon}5.
\end{equation}
Fix $r$ large enough such that $\|P_E  N_\udelta \one[N_\udelta>r]P_E\|\leq \epsilon/(5(2+\sigma^{-1}))$.
We then have
\begin{equation}\label{MissLabel1}
2\re\bigl\{P_E N_\udelta \one[N_\udelta>r]\bP_E\bigr\} \geq -\sigma N_\udelta -\frac{\epsilon}{5}\one_\cH.
\end{equation}
Secondly, we estimate for $\sigma'>0$
\begin{align}\label{MissLabel2}
\nonumber 2\re\bigl\{\tK \bP_E\bigr\} & =  2\re\bigl\{\tK \one[|H-E|<\kappa_3] \bP_E\bigr\} + 2\re\bigl\{\tK \one[|H-E|\geq\kappa_2]\bigr\}\\
& \geq -\bigl(\frac{\epsilon}5 +\sigma'\bigr)\one_\cH - \frac{\|\tK\|^2}{\sigma'}\one[|H-E|\geq\kappa_3].
\end{align}
Here we chose $\kappa_3>0$ small enough such that $\|\tK \one[|H-E|<\kappa_3] \bP_E\|\leq \epsilon/5$.
We may now pick $\sigma' = \epsilon/5$, $\kappa' = \min\{\kappa_2,\kappa_3\}$  and complete the estimate \eqref{HprimeAtEigenvalue}, using
\eqref{MissLabel0}, \eqref{MissLabel1} and \eqref{MissLabel2}, to arrive at
\[
H'_\udelta \geq -\frac{4\epsilon}5 \one_\cH - \sigma N_\udelta - C_3 \one[|H-E|>\kappa'].
\]
To get rid of the extra $\sigma N_\udelta$, we estimate using \eqref{PhiNumberBound} (with $\sigma = 1/2$) and \eqref{ChoiceOfSigma}
\begin{align*}
(1+2\sigma)H'_\udelta &\geq -\frac{4\epsilon}5 \one_\cH - C_3 \one[|H-E|>\kappa] + \sigma N_\udelta -2\sigma \phi(\ri a_\udelta \coup)\\
&\geq  -\epsilon \one_\cH - C_3 \one[|H-E|>\kappa].
\end{align*}

It remains to establish that for the given $\epsilon>0$ and $\udelta\in\Delta(\udelta')$, one can choose $\kappa$ and $C$ such that the desired bound \eqref{LemDerivedME} holds for all $E\in J$.
We proceed by assuming, aiming for a contradiction, that given 
$\kappa_n = 1/n$ and $C_n = n$, there exists an energy $E_n$  
such that \eqref{LemDerivedME} fails. By compactness of $J$, we may assume that $E_n$ converges to some $E_\infty\in J$.
Let $\kappa_\infty$ and $C_\infty$ be the constants just established to exist, such that 
the bound \eqref{LemDerivedME} holds true at $E_\infty$.
Picking $n$ large enough such that 
\[
|E_\infty-E_n| < \kappa_\infty/2,\quad \kappa_n<\kappa_\infty/2, \quad C_n\geq C_\infty, 
\]
sets us up with a contradiction, since we have
\[
H_{\udelta}' \geq -\epsilon\one - C_\infty\one[|H-E_\infty|\geq \kappa_\infty]  \geq -\epsilon\one - C_n\one[|H-E_n|\geq \kappa_n].
\]
\end{proof}

The following theorem, which appeared originally in \cite[Thm.~7.12]{GeorgescuGerardMoeller2004b}, 
states that a Mourre estimate holds at any fixed $E\in\RR$. It holds
also for confined small systems, not necessarily finite dimensional,
but the proof simplifies slightly here. Furthermore, since we do not need resolvents of $H$ to control
$\phi(a_\udelta\coup)$ but can do with resolvents of $N$, the version here in fact holds under slightly 
weaker infrared assumptions on $\coup$. 

Another special feature of finite dimensional small systems, is that we can choose $\delta_0'$ uniformly in energy.
Indeed, we pick $\delta_0'\in(0,1]$ such that
\begin{equation}\label{ChoiceOfdelta00}
 d(\delta_0')\geq \sup_{\udelta\in\Delta_0}\|a_{\udelta}\coup\|^2 +1.
\end{equation}
With this choice, we have for all $\udelta=(\delta_0,\delta_\infty)\in\Delta((\delta_0',1))$ and $|\omega|\leq \delta_0$ that
\begin{equation}\label{ChoiceOfdelta0}
m_\udelta(\omega)\geq m_\udelta(\delta_0) = d(\delta_0)\geq d(\delta_0')\geq \sup_{\udelta\in\Delta_0}\|a_{\udelta}\coup\|^2 +1.
\end{equation}

\begin{theorem}\label{Thm-GGM} Suppose \HGCond{1}. For any $\epsilon>0$, $\udelta\in\Delta((\delta_0',1))$ and $E\in\RR$, there exist $\kappa>0$, $C>0$ and $K$,
a compact and self-adjoint operator, such that the form estimate on $\cD(N)$
\[
H_\udelta' \geq (1-\epsilon) \one_\cH  - C \one[|H-E|\geq \kappa] - K,
\]
holds true. Here $\delta_0'$ is chosen such that \eqref{ChoiceOfdelta00} is satisfied.
\end{theorem}

\begin{proof} Fix $\udelta\in\Delta((\delta_0',1))$. 
We only have something to prove if $E\geq \Sigma$. The proof goes by induction in energy, and we assume the theorem holds true
for all $\epsilon'>0$ and $E'\leq E_0$. We now fix $\epsilon>0$ and $E \in (E_0,E_0+\delta_0']$ and proceed to argue that one may choose $C,\kappa$ and $K$, such that the form estimate in the theorem is satisfied. 

In this proof we make use of geometric localization in two guises,
cf.~Appendix~\ref{Subsect-GeomLoc}. We shall in particular make use of momentum localization,
as in the proof of Theorem~\ref{Thm-ME-H-LargeE}. The other application of geometric localization is in position space, which we now introduce. Here $\gothh =\gothh_0=\gothh_\infty= L^2(\RR^3)$. Let $j_0,j_\infty\in C^\infty(\RR^3)$
with $0\leq j_0\leq 1$, $j_0(\lambda) = 1$ for $|\lambda|\leq 1$, 
$j_0(\lambda)=1$ for $|\lambda|\geq 2$, and finally; $j_0^2 + j_\infty^2 = 1$.
In the following, we write $j_i^R$ for $j_i(x/R)$, where $R>0$ and $x = \ri \nabla_k$.
Then $j^R=(j_0^R,j_\infty^R)\colon \gothh\to \gothh\oplus\gothh$ is an isometry (not unitary)
and the map $\cGamma(j^R)\colon \cF\to \cF\otimes\cF$
is also an isometry. For simplicity, we write
$\cGamma(j^R)$ for $\one_\cK\otimes \cGamma(j^R)$. 

We begin by recalling the notation $N_\udelta = \D\Gamma(m_\udelta)$ and
observing two estimates.
By pseudo differential calculus $[m_\udelta,j_i^R] = O(R^{-1})$, $i=0,\infty$. 
Abbreviate $q = (q_0,q_\infty) = ([m_\udelta,j_0^R],[m_\udelta,j_\infty^R])$
and observe  using \eqref{BasicNumberBound} that
\[
 \cGamma(j^R)^*\D\cGamma(j^R,q) = \D\Gamma(j_0^R q_0+ j_\infty^R q_\infty) \geq 
  - \frac{C}{R}\bigl(N+ \one_{\cF}\bigr)
\]
for some $C>0$. Hence, by \eqref{dGammacGammaIntertwine}, we find that
\begin{align}\label{GeomLocForm}
\nonumber   N_\udelta  & = \cGamma(j^R)^*\cGamma(j^R)N_\udelta\\ 
\nonumber  & =  \cGamma(j^R)^*\bigl(N_\udelta\otimes\one_{\cF} + \one_{\cF}\otimes N_\udelta\bigr)\cGamma(j^R) + \cGamma(j^R)^*\D\cGamma(j^R,q)\\
  & \geq \cGamma(j^R)^*\bigl(N_\udelta\otimes\one_{\cF} + \one_{\cF}\otimes N_\udelta\bigr)\cGamma(j^R) - \frac{C}{R}\bigl(N+ \one_{\cF}\bigr).
  \end{align}
The second estimate we need is
\begin{align}\label{GeomLocForm2}
\nonumber \phi(\ri a_\udelta\coup) & = \frac1{\sqrt{2}}\bigl(\cGamma(j^R)^*\cGamma(j^R) a^*(\ri a_\udelta\coup) + a(\ri a_\udelta\coup)\cGamma(j^R)^*\cGamma(j^R)  \bigr)\\
\nonumber & = \cGamma(j^R)^*\bigl(\phi(\ri j_0^R a_\udelta \coup)\otimes\one_\cF + \one_\cF\otimes \phi(\ri j_\infty^R a_\udelta\coup) \bigr)\cGamma(j^R)\\
&\geq \cGamma(j^R)^*\bigl(\phi(\ri a_\udelta \coup)\otimes\one_\cF\bigr)\cGamma(j^R) 
-o_R(1)(N+\one_\cH),
\end{align}  
where $\lim_{R\to \infty} o_R(1) = 0$. Here we used \eqref{Gamma-a-Inter}, \eqref{PhiEnergyBound} and the fact that 
$\|j_\infty^R a_\udelta \coup\|$ and $\|(1-j_0^R)a_\udelta \coup\|$ both converge to zero for $R\to\infty$.

Write $P = \vacuum\tvacuum$ and $P^\perp = \one_\cF-P$ as projection operators 
on $\cF$ or $\cH$ (read as e.g. $\one_\cK\otimes P$).
 In order to use geometric localization, we need
the extended Hilbert space $\cH^\ext = \cH\otimes\cF$  and the extended Hamiltonian
$H^\ext = H\otimes\one_\cF + \one_\cH\otimes H_\ph$. The extended commutator is
\[
H^{\ext'}_\udelta = H_\udelta'\otimes\one_\cF + \one_\cH \otimes N_\udelta,
\]
as a self-adjoint operator on $\cD(N^\ext)$, where $N^\ext = N\otimes\one_\cF + \one_\cH\otimes N$.

Observe that if $S\colon \cD(N^{1/2})\to \cH^\ext$ is bounded, then for any $\sigma>0$, we have
\begin{align}\label{ObsLemma}
\nonumber & \re\bigl\{\cGamma(j^R)^*(\one_\cH\otimes P) S\big\}  = \re\big\{\Gamma(j_0^R) S_0\bigr\}\\
\nonumber & \quad = \re\bigl\{\one[|H-E|\geq 1]\Gamma(j^R_0)S_0(N+\one_\cH)^{-\frac12} (N+\one_\cH)^{\frac12}\bigr\}  +  \re\bigl\{K_1 (N+\one_\cH)^\frac12 \bigr\}\\
& \quad \geq -\frac{\sigma}4(N+\one_\cH) - \frac{2}{\sigma}C_1\one[|H-E|\geq 1] - \frac2{\sigma} K_1 K_1^*,
\end{align}
where $S_0 = (\one_\cK\otimes P)S\colon\cD(N^{1/2})\to \cH\otimes\CC$ is the vacuum component of $S$ (in the second tensor factor),
$C_1 = \|S_0(N+\one_\cH)^{-1/2}\|^2$ and the compact operator  $K_1 = \one[|H-E| < 1]\Gamma(j_0^R)S_0(N+\one_\cH)^{-1/2}$.

 We fix the constant $\sigma>0$, such that it satisfies:
\begin{equation}\label{MainChoiceOfSigma}
\sigma \leq \frac{\epsilon}{5},\quad  2\sigma\|a_\udelta\coup\|^2\leq \frac{\epsilon}{5}\quad \textup{and}\quad
\frac{1-\frac{4\epsilon}{5}}{1+2\sigma} \geq 1-\epsilon.
\end{equation}
The observation above, together with geometric localization -- in the form of \eqref{GeomLocForm} and \eqref{GeomLocForm2} -- implies that we can pick $R_0 = R_0(\sigma)>0$ large enough, such that for $R\geq R_0$, we have as a form on $\cD(N)$
\begin{align}\label{TempEstimOfB}
\nonumber H_\udelta' & \geq  \cGamma(j^R)^* H^{\ext'}_\udelta \cGamma(j^R) -\frac{\sigma}4 (N+\one_\cH)\\
\nonumber & \geq  \cGamma(j^R)^* (\one_\cH\otimes P^\perp)H^{\ext'}_\udelta(\one_\cH\otimes P^\perp) \cGamma(j^R)\\
& \qquad  -\frac{\sigma}2 (N+\one_\cH)
-C_2\one[|H-E|>1] -K_2.
\end{align}
In the second inequality, we employed \eqref{ObsLemma} with the operator 
$S = -(\phi(\ri a_\udelta\coup)\otimes\one_\cF)\cGamma(j^R)$, 
for which  $S_0 = -\phi(\ri a_\udelta\coup)\Gamma(j_0^R)$. We may in particular take
 $K_2 = 2 K_1 K_1^*/\sigma$ and 
 $C_2 = 2 C_1/\sigma$. Note that both $C_2$ and $K_2$ depend on $R$, which will be fixed at the end of the proof.

We now employ again the geometric localization in momentum space, 
cf.~\eqref{MomentumIMS1} and  \eqref{MomentumIMS2}.
Let $F^{\delta_0} = (\one[|k|\geq\delta_0],\one[|k|< \delta_0])$ and
recall that the map
$\cGamma(F^{\delta_0})\colon \cF\to \cF_{>}^{\delta_0}\otimes\cF_<^{\delta_0}$ is unitary.
(For notational convenience below, we have switched the order of the interior and exterior regions.)
Abbreviate $\hcH^\ext = (\one_\cH\otimes\cGamma(F^{\delta_0}))\cH^\ext = \cH\otimes  \cF_{>}^{\delta_0}\otimes\cF_<^{\delta_0}$.
Compute using \eqref{GeomLocMomentum} the intertwining relations
\begin{align}
\label{Inter1}
& \cGamma(F^{\delta_0})P^\perp 
=\big(\one_{\cF_>^{\delta_0}}\otimes P^\perp_<  + P_>^\perp\otimes P_<\big)\cGamma(F^{\delta_0})\\
\label{Inter2}
& \big(\one_\cH\otimes \cGamma(F^{\delta_0})\big)H^{\ext\prime}_\udelta = \hH^{\ext\prime}_\udelta \big(\one_\cH\otimes \cGamma(F^{\delta_0})\big)\\
\label{Inter3}
& \big(\one_\cH\otimes \cGamma(F^{\delta_0})\big)H^\ext = \hH^\ext  \big(\one_\cH\otimes \cGamma(F^{\delta_0})\big).
\end{align}
Here $P_{>/<}$ denote the orthogonal projections onto the vacuum sectors inside $\cF^{\delta_0}_{>/<}$,
and 
\begin{align*}
& \hH^{\ext\prime}_\udelta = H^{\ext\prime}_{\udelta,>}\otimes \one_{\cF_<^{\delta_0}} 
+ \one_\cH\otimes\one_{\cF_>^{\delta_0}}\otimes N_{\udelta|\cF_<^{\delta_0}},\\
& H^{\ext\prime}_{\udelta,>} = H'_\udelta\otimes\one_{\cF_>^{\delta_0}}+ \one_\cH\otimes N_{\udelta|\cF_>^{\delta_0}},\\
& \hH^\ext = H\otimes \one_{\cF_>^{\delta_0}} \otimes \one_{\cF_<^{\delta_0}} + \one_\cH \otimes H_{\ph |\cF_>^{\delta_0}}\otimes
  \one_{\cF_<^{\delta_0}} + \one_\cH\otimes \one_{\cF_>^{\delta_0}}\otimes  H_{\ph |\cF_<^{\delta_0}}.
\end{align*}
Using that $H'_\udelta\geq -\|a_\udelta \coup\|^2\one_\cH$, cf. \eqref{PhiNumberBound}, we estimate
\begin{equation}\label{Estim1}
\bigl(\one_\cH\otimes \one_{\cF_>^{\delta_0}}\otimes P_<^\perp\bigr)
\hH^{\ext\prime}_\udelta
\geq \bigl(m(\delta_0) - \|a_\udelta \coup\|^2 \bigr)\one_\cH\otimes  \one_{\cF_>^{\delta_0}}\otimes P_<^\perp
\end{equation}
and observe the identity
\begin{equation}\label{Id1}
\bigl(\one_\cH\otimes P_>^\perp\otimes P_<\bigr)
\hH^{\ext\prime}_\udelta
 = \bigl( (\one_\cH\otimes P_>^\perp)H_{\udelta,>}^{\ext\prime}(\one_\cH\otimes P_>^\perp)\bigr)\otimes P_<.
\end{equation}

Using the intertwining relations \eqref{Inter1} and \eqref{Inter2}, together with \eqref{Estim1}, \eqref{Id1}  and the choice of $\delta_0'$, cf. \eqref{ChoiceOfdelta0}, we get
\begin{align}\label{GGMStep1}
& (\one_\cH\otimes P^\perp)H^{\ext\prime}_\udelta(\one_\cH\otimes P^\perp) \geq
\bigl(\one_\cH\otimes \cGamma(F^{\delta_0})^*\bigr)\\
\nonumber & \quad \times 
\Bigl\{\one_\cH\otimes  \one_{\cF_>^{\delta_0}} \otimes P_<^\perp
+\bigl( (\one_\cH\otimes P_>^\perp)H_{\udelta,>}^{\ext\prime}(\one_\cH\otimes P_>^\perp)\bigr)\otimes P_<\Bigr\}
\bigl(\one_\cH\otimes \cGamma(F^{\delta_0})\bigr).
\end{align}
To deal with the term in the brackets, we note that
\begin{equation}\label{GGMStep2}
(\one_\cH\otimes P_>^\perp)H_{\udelta,>}^{\ext\prime}(\one_\cH\otimes P_>^\perp)
 \geq \one_\cH\otimes P_>^\perp + H'_\udelta\otimes P_>^\perp
\end{equation}
and estimate using Lemma~\ref{Lemma-FromPointMEToUnifME} with $\epsilon$ replaced by $\epsilon/5$, and the induction assumption
\begin{align}\label{GGMStep3}
&\nonumber  H_\udelta'\otimes P_>^\perp  \otimes P_<  =
\Bigl\{\bigoplus_{\ell=1}^\infty \int^\oplus_{(\RR^3\backslash B(\delta_0))^\ell} H'_\udelta\, \D k_1 \cdots \D k_\ell\Bigr\}\otimes P_<\\
\nonumber & \geq  - \Bigl\{\bigoplus_{\ell=1}^\infty \int^\oplus_{(\RR^3\backslash B(\delta_0))^\ell} 
\left(\frac{\epsilon}5 \one_\cH + C\one[|H+\textstyle\sum_{j=1}^\ell |k_j|-E|\geq\kappa]\Bigr)  \D k_1 \cdots \D k_\ell\right\}\otimes P_<\\
\nonumber & = -\frac{\epsilon}{5}\one_\cH\otimes P_>^\perp\otimes P_< - C\one[|\hH^\ext-E|\geq\kappa]\one_\cH\otimes P_>^\perp\otimes P_<\\
&\geq -\frac{\epsilon}{5}\one_{\hcH^\ext}  - C\one[|\hH^\ext-E|\geq\kappa].
\end{align}
Here $\kappa$ and $C$ are coming from Lemma~\ref{Lemma-FromPointMEToUnifME}.
Combining \eqref{GGMStep1}--\eqref{GGMStep3}, cf.~also \eqref{Inter1} and \eqref{Inter3}, we find
\begin{align}\label{TempEstOfB2}
\nonumber& (\one_\cH\otimes P^\perp)H^{\ext\prime}_\udelta(\one_\cH\otimes P^\perp) \geq
\bigl(\one_\cH\otimes \cGamma(F^{\delta_0})^*\bigr)\\
\nonumber& \qquad \times 
\Bigl\{
\one_\cH\otimes  \one_{\cF_>^{\delta_0}} \otimes P_<^\perp + \one_\cH \otimes P_>^\perp \otimes P_<  
-\frac{\epsilon}{5}\one_{\hcH^\ext}  - C\one[|\hH^\ext-E|\geq\kappa]\Bigr\}\\
\nonumber &\qquad \times \bigl(\one_\cH\otimes \cGamma(F^{\delta_0})\bigr)\\
&\quad  = \one_\cH\otimes P^\perp - \frac{\epsilon}{5} \one_{\cH^\ext} - C\one[|H^\ext-E|\geq\kappa].
\end{align}

Pick a non-negative $f\in C_0^\infty(\RR)$ with $\supp(f)\subseteq [-\kappa,\kappa]$ and
$f=1$ on the interval $[-\kappa/2,\kappa/2]$.
Inserting \eqref{TempEstOfB2} into \eqref{TempEstimOfB}, we estimate for $R\geq R_0$
\begin{align}\label{TempEstOfB3}
\nonumber H'_\udelta & \geq \cGamma(j^R)^* \bigl\{\one_\cH\otimes P^\perp - \frac{\epsilon}{5} \one_{\cH^\ext} - C\one[|H^\ext-E|\geq\kappa]\bigr\} \cGamma(j^R)\\
\nonumber & \qquad  -\frac{\sigma}2 (N+\one_\cH)
-C_2\one[|H-E|\geq 1] -K_2\\
\nonumber &\geq  \cGamma(j^R)^* \Bigl\{\bigl(1 - \frac{\epsilon}{5}\bigr) \one_{\cH^\ext} - Cf(H^\ext-E)\Bigr\} \cGamma(j^R)\\
& \qquad  -\frac{\sigma}2 (N+\one_\cH)
-C_3\one[|H-E|\geq 1] -K_3.
\end{align}
 In the last inequality, we made us of the estimate
\begin{align*}
&-\cGamma(j^R)(\one_\cH\otimes P)\cGamma(j^R)  = -\Gamma(j_0^R)^2 \\
&\quad  \geq - \one[|H-E|\geq 1]
-\re\bigl\{(1+\one[|H-E|\geq 1])\Gamma(j_0^R)^2\one[|H-E| < 1]\bigr\},
\end{align*} 
the last term on the right-hand side being compact, such that we may take $C_3 = C_2+1$ and $K_3 = K_2 + \re\{(1+\one[|H-E| < 1])\Gamma(j_0^R)^2\one[|H-E| < 1]\}$.

We proceed to argue that for $R > R_0$ sufficiently large,  we have
\begin{equation}\label{TempEstOfB4}
\cGamma(j^R)f(H^\ext-E) \cGamma(j^R)\leq f(H) + \frac{\epsilon}{5}\one_\cH + \frac{\sigma}2 (N+\one_\cH),
\end{equation}
in the sense of forms on $\cD(N)$.
To see this, we estimate first for $\psi\in\tcC^\romL$ and $\varphi\in\cC$, using \eqref{BasicNumberBound} (with $\rho=1/2$), \eqref{dGammacGammaIntertwine} and \eqref{phicGammaIntertwine} (with $f_\infty = 0$)
\[
\bigl|\bigl\la \psi,\bigl(H^\ext\cGamma(j^R)-\cGamma(j^R)H\bigr)\varphi\bigr\ra\bigr|
\leq o_R(1)\bigl\|(N^\ext+\one_{\cH})^{\frac12}\psi\bigr\|\bigl\|(N+\one_\cH)^\frac12 \varphi\bigr\|.
\]
Here $N^\ext = \one_\cK\otimes N\otimes \one_\cF + \one_\cH\otimes N$, $\lim_{R\to\infty} o_R(1) = 0$ and we used that $[|k|,j_i^R]$ extends by continuity from a form
on $\cD(|k|)$ to a bounded form on $\gothh$, bounded by $c/R$ for some $c>0$. 
The estimate extends by continuity to $(\psi,\varphi)\in \cD(N^\ext)\cap\cD(H^\ext) \times \cD(N)\cap\cD(H)$. Secondly, using Proposition~\ref{Prop-DomInv}~\ref{Item-ResInv-MoPlus}, we find an $n\in\NN$ such that
\begin{align*}
& \bigl|\bigl\la \psi,\bigl(\cGamma(j^R)(H-z)^{-1}-(H^\ext-z)^{-1}\cGamma(j^R)\bigr)\varphi\bigr\ra\bigr|\\
& \qquad 
\leq o_R(1)\bigl(1+|\im z|^{-n}\bigr)\bigl\|(N^\ext+\one_{\cH})^{\frac12}\psi\bigr\|\bigl\|(N+\one_\cH)^\frac12\varphi\bigr\|
\end{align*}
for all $\psi\in\cD(N^\ext)$, $\varphi\in\cD(N)$ and $z\in\CC$ with $\im z\neq 0$.
The estimate \eqref{TempEstOfB4} now follows from an almost analytic extension argument for an $R>R_0$ sufficiently large, which we now fix. 
For almost analytic extensions, we refer the reader to \cite{Moeller2000}.

Inserting \eqref{TempEstOfB4} into \eqref{TempEstOfB3} with the $R$ fixed above yields
\begin{align*}
H_\udelta' & \geq \bigl(1 - \frac{2\epsilon}{5}\bigr) \one_{\cH} -\sigma (N+\one_\cH) - Cf(H-E)-C_3\one[|H-E|\geq 1] - K_3\\
&\geq \bigl(1-\frac{2\epsilon}{5}-\sigma\bigr)\one_\cH -\sigma N - C_4\one[|H-E|\geq\kappa/2]- K_3,
\end{align*}
where $C_4 = C+C_3$. 

The proof is now completed, as in the proof of Lemma~\ref{Lemma-FromPointMEToUnifME}, by the bound
\begin{align*}
(1+2\sigma)H'_\udelta & \geq  H'_\udelta + \sigma N - 2 \sigma \|a_\udelta\coup\|^2\one_\cH\\
&  \geq 
\bigl(1-\frac{4\epsilon}{5}\bigr)\one_\cH - C_4\one[|H-E|\geq\kappa/2]- K_3,
\end{align*}
where we used the choice of $\sigma$, cf.~\eqref{MainChoiceOfSigma}. This concludes the proof.
\end{proof}

We arrive at the following structure result for the pure point spectrum of $H$.

\begin{corollary}\label{Cor-Finite-pp-H} Suppose \HGCond{1}. The operator $H$ has a finite number of eigenvalues, all of finite multiplicity.
\end{corollary}

\begin{proof}
 Assume towards a contradiction that there exists an enumerable sequence $\psi_n$
 of mutually orthogonal normalized eigenstates. Let $\mu_j$ denote the 
 corresponding eigenvalues. Due to Theorems~\ref{Thm-VirialH} and~\ref{Thm-ME-H-LargeE} we know that $\{\mu_j\}_{j=1}^\infty$
 is a bounded sequence. Hence we can assume that it is convergent towards an energy $E$.
 
 Now Theorems~\ref{Thm-VirialH} and~\ref{Thm-GGM}, applied with $\epsilon=1/2$, yield the estimate
 \[
 0\geq \frac12 - \la \psi_j,K\psi_j\ra,
 \] 
 for $j\geq j_0$, where $j_0$ is such that $|\mu_j-E|<\kappa$ for $j\geq j_0$.
 Since $K$ is compact and $\wlim \psi_j = 0$, we conclude that $\lim_{j\to\infty}\la \psi_j,K\psi_j\ra=0$.
 This establishes the sought after contradiction.
\end{proof}

In the following, we denote by $P$ the finite rank projection that projects onto the subspace consisting of eigenstates for $H$,
and we write $\bP=\one_\cH-P$. 

\begin{corollary}\label{Cor-FromPointMEToUnifME} Suppose
  \HGCond{1}. There exists $\udelta'\in\Delta_0$, such that: For any $\epsilon>0$ and $\udelta\in\Delta(\udelta')$, 
there exist $\kappa>0$ and $C>0$, such that
the following two estimates hold for all $E\in\RR$
\begin{align}
\label{UnifHBound1} H'_\udelta &\geq -\epsilon \one_\cH - C\one[|H-E|\geq \kappa],\\
\label{UnifHBound2} H'_\udelta &\geq (1-\epsilon)\one_\cH - C(\one[|H-E|\geq \kappa] + P),
\end{align}
in the sense of forms on $\cD(N)$.
\end{corollary}

\begin{proof} Put $\udelta' = (\delta_0',\delta_\infty')$, where $\delta_\infty'$ 
and $\delta_0'$ come from Theorems~\ref{Thm-ME-H-LargeE} and~\ref{Thm-GGM}, respectively.
The estimate \eqref{UnifHBound1} is now a direct consequence of
these two theorems together with Lemma~\ref{Lemma-FromPointMEToUnifME}.

We proceed to the second bound \eqref{UnifHBound2}. This bound is obviously true for $E>E_0+1$ (cf. Theorem~\ref{Thm-ME-H-LargeE})
and for $E<\Sigma-1$, so what remains is to prove
the estimate uniformly in $E \in [\Sigma-1,E_0+1] =: J$, which is a compact interval. Let $\udelta\in\Delta(\udelta')$ and $\epsilon>0$. We first argue that the estimate is correct
for fixed $E\in J$. Apply Theorem~\ref{Thm-GGM} with $\epsilon$ replaced by $\epsilon/4$. 
For the resulting compact operator $K$, write
\[
K = P K P + \re\{(\one_\cH+P) K\bP\}\geq -\|K\| P + \re\{(\one_\cH+P) K\bP\} .
\]
Decompose
\[
K\bP = K\bP \one[|H-E|\geq \kappa] + K\bP\one[|H-E| < \kappa],
\]
where one can choose $\kappa$ small enough, such that $\|K\bP\one[|H-E| < \kappa]\|\leq\epsilon/4$.
We estimate, for any $\sigma>0$, 
\begin{align*}
\re\{(\one_\cH+P) K\bP\} &\geq -\frac{2\epsilon}{4}\one_\cH  + \re\bigl\{(\one_\cH+P) K\bP\one[|H-E|\geq \kappa] \bigr\}\\
&  \geq  -\frac{2\epsilon}{4}\one_\cH -\sigma \one_\cH - \frac{\|K\|^2}{\sigma}\one[|H-E|\geq \kappa].
\end{align*}
Choosing $\sigma=\epsilon/4$, we get -- for some $C>0$ -- the estimate
\[
K \geq -\frac{3\epsilon}4  \one_\cH - C \bigl(\one[|H-E|\geq \kappa]+ P\bigr).
\]
This completes the argument that for a fixed $E$, one can find $\kappa$ and $C$, such that the commutator estimate \eqref{UnifHBound2} holds true.

 Suppose the estimate \eqref{UnifHBound2} is not correct uniformly in $E$.
That is, for any $\kappa>0$ and $C>0$, there exists $E\in J$ such that estimate fails to hold.

Put $\kappa_n = 1/n$ and $C_n = n$. This gives a sequence $E_n\in J$, for which the estimate \eqref{UnifHBound2} is false.
We may assume, due to compactness of $J$, that $E_n$ converges to an energy $E_\infty\in J$.
Recalling that we have just verified that \eqref{UnifHBound2} holds for a fixed $E\in J$,
we get a $\kappa_\infty>0$ and $C_\infty>0$, such that \eqref{UnifHBound2} holds true
at $E_\infty$. Pick $n$ large enough, such that $1/n<\kappa_\infty/2$, $C_n> C_\infty$ and $|E_\infty-E_n|<\kappa_\infty/2$.
Then
\begin{align*}
H'_\udelta & \geq (1-\epsilon)\one_\cH - C_\infty(\one[|H-E|\geq \kappa_\infty] + P) \\
& \geq (1-\epsilon)\one_\cH - C_n(\one[|H-E_n|\geq \kappa_n] + P),
\end{align*}
contradicting the choice of $E_n$.
\end{proof}

\subsection{Estimates at Positive Temperature}\label{Subsec-CommBoundsL}

In this subsection, we use the notation $\tN_\udelta$ for
$\D\Gamma(m_\udelta)$, which is the analogue of $N_\udelta$ from the
previous subsection. We can then write $\tL_{\beta,\udelta}' = \one_{\cK\otimes\cK}\otimes\tN_\udelta-
\phi(\ri \ta_\udelta\tcoup_\beta)$.

\begin{theorem}\label{Thm-ME-L-LargeE} Suppose \LGCond{1}. Let $e>0$
  be given. 
There exists $E_0>0$, $\delta_\infty'>0$  and $C>0$ such that the following
form bound holds on  $\cD(N^\romL)$ for all $E\geq E_0$ and $\udelta\in \Delta((1,\delta'_\infty))$
\[
L_{\beta,\udelta}' \geq e\one - C\one[|L_\beta|\leq E].
\]
\end{theorem}

\begin{proof} It suffices to prove the theorem with $L_\beta$ and $N^\romL$
  replaced by $\tL_\beta$ and $\tN$.
The proof is divided into two steps. First we consider the
uncoupled glued Liouvillean $\tL_0$.
The reader should not confuse the subscript $0$ with infinite temperature (zero inverse temperature).
We proceed as in the proof of Theorem~\ref{Thm-ME-H-LargeE}.

Observe that
\begin{equation}\label{FormOfL0prime}
\tL_{0,\udelta}' = \one_{\cK\otimes\cK}\otimes \tN_\udelta
\end{equation}
and the estimate
\begin{equation}\label{HE-CoupEasy}
\tL_{\beta,\udelta}' = \tL_{0,\udelta}' - \tphi(\ri \ta_\udelta \tcoup_\beta)\geq \frac12 \tL_{0,\udelta} - \bigl\| \ta_\udelta\tcoup_\beta\bigr\|^2\one_{\tcH}.
\end{equation}
For $R>1$ we again perform a partition of unity in momentum space as follows.
Let
\begin{align*}
& \tF^R  = \begin{pmatrix} \one[|\omega| < R] \\ \one[|\omega|\geq R]\end{pmatrix}\colon \tgothh \to \tgothh_<\oplus\tgothh_>\\
& \tgothh_<^R := L^2((-R,R))\otimes L^2(S^2),\quad  \tgothh_>^R:= L^2((-\infty,R]\cup[R,\infty))\otimes L^2(S^2).
\end{align*}
Compare with \eqref{MomentumIMS1} and  \eqref{MomentumIMS2}.
Put $\tcF_<^R = \Gamma(\tgothh_<^R)$, $\tcF_>^R = \Gamma(\tgothh_>^R)$ and
\[
\tL_0^\ext = L_\romp\otimes \one_{\tcF_<^R\otimes\tcF_>^R} 
 + \one_{\cK\otimes\cK}\otimes\,
 {\D\Gamma(\omega)}_{|\tcF_<^R}\otimes\one_{\tcF_>^R} + \one_{\cK\otimes\cK\otimes\tcF_<^R}\otimes
\, {\D\Gamma(\omega)}_{|\tcF_>^R}.
\]
Note that $\cGamma(\tF^R)$ is unitary and, as usual, we simply write $\cGamma(\tF^R)$ instead of 
$\one_{\cK\otimes\cK}\otimes\cGamma(\tF^R)$ acting on $\tcH$.
Abbreviate $\lambda_\rommax = \max\sigma(K)$ and $\lambda_\rommin = \min\sigma(K)$.
We estimate first for $\tE>2\lambda_\rommax-\lambda_\rommin$, using \eqref{FormOfL0prime} and the obvious analogue of \eqref{GeomLocMomentum}:
\begin{align*}
\tL_{0,\udelta}' & \geq \bigl(\one_{\cK\otimes\cK}\otimes\tN_\udelta\bigr) \one\bigl[\bigl|\tL_0\bigr|+\tN > \tE\bigr]\\
 & = \cGamma(\tF^R)^* \left\{\one_{\cK\otimes\cK}\otimes \tN_{\udelta_{|
       \tcF_<}}\otimes\one_{\tcF_>} + 
 \one_{\cK\otimes\cK\otimes \tcF_<}\otimes
 \tN_{\udelta_{|\tcF_>}}\right\}\\
 &\qquad \times 
 \one\bigl[\bigl|\tL_0^\ext\bigr|+\tN^\ext > \tE\bigr]\cGamma(\tF^R)\\
& \geq \Gamma(\one[|\omega|<R])\tN_\udelta \one\bigl[\bigl|\tL_0\bigr|+\tN>\tE\bigr] \\
&\quad  + m_\delta(R)\cGamma(\tF^R)^* \bigl(\one_{\cK\otimes \cK\otimes
   \tcF_<}\otimes \bP_\Omega\bigr) \one\bigl[\bigl|\tL_0^\ext\bigr|
  +\tN^\ext>\tE\bigr]\cGamma(\tF^R)\\
& \geq \frac{\tE-\lambda_\rommax+\lambda_\rommin}{R+1}
\Gamma(\one[|\omega|<R])\one\bigl[\bigl|\tL_0\bigr|
 +\tN>\tE\bigr]\\
 &\quad  + m_\udelta(R)\cGamma(\tF^R)^* \bigl(\one_{\cK\otimes \cK\otimes
    \tcF_<}\otimes \bP_\Omega\bigr) \one\bigl[\bigl|\tL_0^\ext\bigr|+\tN^\ext\geq\tE\bigr]\cGamma(\tF^R)\\
& \geq \min\Bigl\{\frac{\tE-\lambda_\rommax+\lambda_\rommin}{R+1},m_\udelta(R)\Bigr\}
\one\bigl[\bigl|\tL_0\bigr|+\tN>\tE\bigr].
\end{align*}
Secondly, we estimate
\begin{align*}
& \one\bigl[\bigl|\tL_0\bigr|+\tN \leq\tE\bigr] \leq  \bigl(\tE+1\bigr)\bigl(\bigl|\tL_0\bigr|+ \tN  +1\bigr)^{-1}\\
 &\qquad  = \bigl(\tE+1\bigr)\bigl(\bigl|\tL_\beta\bigr|+1\bigr)^{-\frac12}\\
 & \qquad \qquad \times
 \Bigl\{\bigl(\bigl|\tL_\beta\bigr|+1)^{\frac12}\bigl(\bigl|\tL_0\bigr|+\tN+1\bigr)^{-1}\bigl(\bigl|\tL_\beta\bigr|+1\bigr)^{\frac12}\Bigr\}
\bigl(\bigl|\tL_\beta\bigr|+1\bigr)^{-\frac12}\\
 &\qquad \leq C\bigl(\tE+1\bigr)\one\bigl[\bigl|\tL_\beta\bigr| \leq E\bigr] + C\frac{\tE+1}{E+1}\one_{\tcH}.
\end{align*}
Here we used Proposition~\ref{Prop-BasicLReg}~\ref{Item-NLDomainInv} to establish that
\[
C = \bigl\|\bigl(\bigl|\tL_\beta\bigr|+1\bigr)^{\frac12}
\bigl(\bigl|\tL_0\bigr|+\tN+1\bigr)^{-1}\bigl(\bigl|\tL_\beta\bigr|+1\bigr)^{\frac12}\bigr\|<\infty.
\]
Inserting these two estimates into \eqref{HE-CoupEasy}, we arrive at 
\begin{align*}
\tL_{\beta,\udelta} & \geq  \Bigl(\frac12\min\Bigl\{\frac{\tE-\lambda_\rommax+\lambda_\rommin}{R+1},m_\udelta(R)\Bigr\}
-  \frac12 C\frac{\tE+1}{E+1}- \bigl\| \ta_\udelta\tcoup_\beta\bigr\|^2\Bigr)\one_{\tcH}\\
&\qquad -\frac12 C\bigl(\tE+1\bigr)\one\bigl[\bigl|\tL_\beta\bigr| \leq E\bigr].
\end{align*}
We may now pick $R$, $\delta'_\infty$, $\tE$ and $E$, 
in that order, as in the proof of Theorem~\ref{Thm-ME-H-LargeE} to conclude the proof.
Here we used that $\sup_{\udelta\in\Delta_0}\|\ta_\udelta\tcoup_\beta\|<\infty$.
\end{proof}

It is now an immediate consequence of Theorem~\ref{Thm-VirialL} that

\begin{corollary} Suppose \LGCond{2}. The set of eigenvalues $\sigma_\pp(L_\beta)$ is bounded.
\end{corollary}

From now on we assume at least \HGCond{1} and  fix $\udelta'$ such that 
Corollary~\ref{Cor-FromPointMEToUnifME} holds true. Recall that
\LGCond{1} implies \HGCond{1}.

\begin{proposition}\label{LUnifBoundTempZero} Suppose \HGCond{1}. Let $\epsilon>0$ and  $\udelta\in\Delta(\udelta')$ be given. There exist $\kappa>0$ and $C>0$,  
such that for all $E\in \RR$:
\[
L_{\infty,\udelta}'\geq (1-\epsilon)\one_{\cH^\romL} - C\bigl(\one[|L_\infty-E|\geq \kappa] +P\otimes P^\conj\bigr),
\]
in the sense of forms on $\cD(N^\romL)$.
\end{proposition}

\begin{remark} Note that at zero temperature, we do not need Nelson's
  commutator theorem to build $L_\infty$, nor do we have any
  singularities from $\rho_\beta$ to absorb. Hence, we may work under an
  \HGCond{1} condition instead of an \LGCond{1} condition.
\remarkQED\end{remark}

\begin{proof}
The starting point is the identity
\[
L_{\infty,\udelta}' = H'_\udelta\otimes \one_\cH + \one_\cH\otimes {H_\udelta^\conj}'.
\]
Denote by $P\in \cB(\cH)$ the projection onto the span of all eigenstates of the operator $H$.
This is a finite range projection and hence compact. Put $P^\conj = \Conj P\Conj$
to be the eigen projection onto
the span of the eigenstates of $H^\conj$. We write $\bP = \one-P$ and
$\bP^\conj = \one-P^\conj$.
We deal with $H'\otimes \one$ only since bounds on $\one\otimes {H^\conj}'$ can be obtained
by conjugation with $\cE\Conj$, where $\cE$ is the exchange map that sends $\psi\otimes\varphi$ to
$\varphi\otimes\psi$. Here $\psi,\varphi\in\cH$.

We write
\begin{equation}\label{PosTemp-Step1}
H'_\udelta = PH'_\udelta P + 2\re\bigl\{PH'_\udelta\bP\bigr\} + \bP H'_\udelta \bP,
\end{equation}
which makes sense as forms on $\cD(N^{1/2})$, since $P$ maps into $\cD(N^{1/2})$ by Theorem~\ref{Thm-NumberBoundH}.
We estimate each term differently. For the first and last term we use \eqref{UnifHBound1}
and \eqref{UnifHBound2} from Corollary~\ref{Cor-FromPointMEToUnifME} (applied with $\epsilon/9$ instead of $\epsilon$)
and find
\begin{equation}\label{PosTemp-Step2}
\begin{aligned}
P H'_\udelta P & \geq -\frac{\epsilon}{9} P -C\one[|H-\lambda - E|\geq\kappa]\geq -\frac{\epsilon}{9}\one_\cH -C \one[|H-\lambda - E|\geq \kappa]\\
\bP  H'_\udelta \bP & \geq \bigl(1-\frac{\epsilon}{9}\bigr) \bP - C\one[|H-\lambda - E|\geq \kappa].
\end{aligned}
\end{equation}
As for the cross term $P H'_\udelta \bP$, we proceed in a fashion similar to what was done 
in the proof of Lemma~\ref{Lemma-FromPointMEToUnifME}. Write for an $r>0$
\[
P H'_\udelta\bP  = P N_\udelta \one[N_\udelta>r]\bP  + K\bP,
\]
with $K= P N_\udelta\one[N_\udelta\leq r]-P \phi(\ri a_\udelta \coup)$ being compact.
We can now fix first $\sigma$ small enough, and subsequently $r$ large enough, such that
\[
2\re\bigl\{P N_\udelta \one[N_\udelta>r]\bP\bigr\}\geq -\sigma N_\udelta - \frac{\epsilon}{9} \one_\cH
\]
and
\begin{equation}\label{AlmostLastChoiceofSigma}
2\sigma \|a_\udelta\coup\|\leq \frac{\epsilon}{9}, \quad \frac{1-\frac{8\epsilon}{9}}{1+2\sigma}>1-\epsilon.
\end{equation}
To deal with the term $K\bP$ we note that we can choose $\kappa$ small enough such that
$2\|K \bP\one[|H-\lambda|<\kappa]\|\leq \epsilon/18$ uniformly in $\lambda$. Indeed,
there exists $\Lambda$ such that $2\|K \bP\one[|H|>\Lambda]\|\leq \epsilon/18$ and hence
by a covering argument there exists $\kappa>0$ such that $2\|K \bP\one[|H-\lambda|<\kappa]\|\leq \epsilon/18$
uniformly in $\lambda\in\RR$. We thus get for all $\lambda,E\in\RR$:
\begin{align*}
2\re\bigl\{K\bP\bigr\}  &= 2\re\bigl\{K\bP\one[|H-\lambda-E|\geq\kappa]\bigr\}+2\re\bigl\{K\bP\one[|H-\lambda-E|<\kappa]\bigr\}\\
&\geq -\frac{\epsilon}{18}\one_\cH - \frac{C}{\epsilon}\one[|H-\lambda-E|>\kappa] - \frac{\epsilon}{18}\one_\cH.
\end{align*}
Inserting this together with  \eqref{PosTemp-Step2} into \eqref{PosTemp-Step1}, we arrive at the bound
\begin{equation}\label{PosTemp-Step3}
H'_\udelta\geq \big(1-\frac{\epsilon}{9}\big)\bP -\frac{3\epsilon}{9}\one_\cH - \sigma N_\udelta - C\one[|H-\lambda-E|\geq\kappa].
\end{equation}

From the spectral theorem in multiplication operator form, we get a 
measure space $(\cM,\Sigma,\mu)$, a measurable real function $f$ on $\cM$ and a unitary map $U\colon \cH\to L^2(\cM)$ such that
$UH U^* = M_f$, multiplication by $f$. Put $U^\conj = U C$ such that $U^\conj H^\conj {U^\conj}^* = M_f$ as well.
Here ${U^\conj}^* = CU^*$.  The combined map $U^\romL = U\otimes U^\conj\colon \cH^\romL\to L^2(\cM\times\cM)$
(with product $\sigma$-algebra and measure) now sets up the correspondence
$U^\romL L_\infty {U^\romL}^* = M_{f_1 - f_2}$, where $f_j(q_1,q_2) = f(q_j)$.

Then, under the identification
$L^2(\cM\times\cM) = L^2(\cM;L^2(\cM))$, we get
\[
\one[|M_{f_1-f_2}-E|\geq \kappa] = \int_\cM^\oplus\one[|M_f-f(q)-E|\geq\kappa]\,\D\mu(q).
\]
Hence, we conclude from \eqref{PosTemp-Step3} the estimate
\begin{align*}
&U^\romL H'_\udelta\otimes\one_\cH \,{U^\romL}^* = \int^\oplus_\cM U H'_\udelta U^*\, \D\mu(q) \\
& \geq
\int^\oplus_\cM  \bigl(1-\frac{\epsilon}{9}\bigr)U\bP U^* -\frac{3\epsilon}{9}\one_{L^2(\cM)} - \sigma U N_\udelta U^* - C\one[|M_f-f(q)-E|\geq\kappa] \D\mu(q)\\
&  = U^\romL \Bigl(\bigl(1-\frac{\epsilon}{9}\bigr)\bP\otimes\one_\cH  -\frac{3\epsilon}{9}\one_{\cH^\romL} - \sigma  N_\udelta\otimes\one_\cH  - C\one[|L_\infty-E|\geq\kappa] \Bigr){U^\romL}^*
\end{align*}
in the sense of forms on $U^\romL\cD(N^\romL)$. 
Adding to the above a similar bound for $\one_\cH\otimes {H^\conj}'$ yields
\begin{align*}
L_\infty' &\geq \bigl(1-\frac{\epsilon}{9}\bigr)\bigl[\bP\otimes\one_\cH + \one_\cH\otimes \bP^\conj\bigr]-\frac{6\epsilon}{9}\one_{\cH^\romL}   -\sigma N_\udelta^\romL - 2C\one[|L_\infty-E|>\kappa]\\
&\geq \bigl(1-\frac{7\epsilon}{9}\bigr)\one_{\cH^\romL} - \sigma N_\udelta^\romL- 2(C+1)\bigl(\one[|L_\infty-E|>\kappa] +P\otimes P^\conj\bigr).
\end{align*}
Here we abbreviated $N^\romL_\udelta = N_\udelta\otimes\one_\cH+\one_\cH\otimes N_\udelta$, and used that
\[
\bP\otimes\one_\cH + \one_\cH\otimes \bP^\conj = 2\one_{\cH^\romL} - P\otimes \bP^\conj - \bP\otimes P^\conj
- 2P\otimes P^\conj\geq \one_{\cH^\romL}- 2P\otimes P^\conj.
\]

We now complete the proof, cf. \eqref{AlmostLastChoiceofSigma}, by estimating
\[
(1+2\sigma) L_{\infty,\udelta}' \geq
\bigl(1-\frac{8\epsilon}{9}\bigr)\one_{\cH^\romL} 
- \tC\bigl(\one[|L_\infty-E|\geq\kappa] +P\otimes P^\conj\bigr),
\]
as at the end of the proofs of Lemma~\ref{Lemma-FromPointMEToUnifME} and Theorem~\ref{Thm-GGM}.
\end{proof}

In order to perturb around zero temperature, we first need to control
the difference $\tcoup_\beta-\tcoup_\infty$. 

\begin{lemma}\label{Lemma-IntBoundL} Suppose \LGCond{n}, for some $n\geq 0$. For any
$\beta_0>0$ there exists $C>0$ such that for all $\beta\geq \beta_0$
we have
\begin{equation}\label{ContInbeta}
\bigl\|\tcoup_\beta-\tcoup_\infty\bigr\|\leq C \beta^{-\frac12}.
\end{equation}
If $n\geq 1$, we have furthermore that for all $\udelta\in\Delta_0$
\begin{equation}\label{DerContInbeta}
\bigl\|\ta_\udelta\bigl(\tcoup_\beta-\tcoup_\infty\bigr)\bigr\|\leq C\beta^{-\frac12}.
\end{equation}
\end{lemma}

\begin{proof} We begin with \eqref{ContInbeta}.
For simplicity we only consider the term
$(\sqrt{1+\trho_\beta}-1)\tcoup_\infty$ in the expression for $\tcoup_{\beta}$, cf.~\eqref{tGbeta1}. The other term
$\sqrt{\trho_\beta}\tcoup_{\infty,\cR}^*$ can be dealt with in a
similar fashion.

Suppose an \LGCond{n} condition, with $n\geq 0$.
We split into the infrared and ultraviolet regimes and estimate first
for $|\omega|\leq 1$:
\[
\bigl(\sqrt{1+\trho_\beta(\omega)}-1\bigr)^2\bigl|\tG_\infty(\omega,\Theta)\bigr|^2 \leq
  C \bigl(\sqrt{1+\trho_\beta(\omega)}-1\bigr)^2 |\omega|^{2n+2\mu}.
\]
Hence we can bound the $L^2$-norm squared of the contribution by a multiple of
\begin{equation*}
\int_{0}^1 \bigl(\sqrt{1+\trho_\beta(\omega)}-1\bigr)^2 \omega^{2n+2\mu}
  \,\D\omega  
 \leq \beta^{-1}\int_0^1 (1+\omega)\omega^{2n-1+2\mu}\,\D \omega,
\end{equation*}
where we simply discarded the $-1$ term coming from $\tcoup_\infty$.
The integral is finite for all $n\geq 0$. 
 In fact, the effect of subtracting $\tcoup_\infty$ sits in the ultraviolet part where
$|\omega|\geq 1$. Here we estimate the $L^2$-norm squared by 
\begin{equation*}
\int_1^\infty \bigl(\sqrt{1+\trho_\beta(\omega)}-1\bigr)^2 \omega^{-1-2\mu}\,
\D\omega\leq \frac{\bigl(\sqrt{1+\trho_\beta(1)}-1\bigr)^2}{2\mu}. 
\end{equation*}
Since $\sqrt{1+\trho_\beta(1)}-1 = \sqrt{1/(1-\e^{-\beta})}-1 \sim
\e^{-\beta/2}$ in the limit of large $\beta$, we get for a fixed
$\beta_0>0$
a constant $C=C(\beta_0)$ such that for all $\beta>\beta_0$ we have
\eqref{ContInbeta} satisfied.

To establish \eqref{DerContInbeta} we observe that
\[
\ta_\udelta \tcoup_\beta = \bigl(\sqrt{1+\trho_\beta}-1\bigr)\ta_\udelta\tcoup_\infty +
m_\udelta \tcoup_\infty\frac{\partial \sqrt{1+\trho_\beta}}{\partial\omega}.
\]
The first contribution can be estimate exactly as above, using that
$n \geq 1$, and yields an
$O(\beta^{-1/2})$ term. For the second term we compute
\[
\frac{\partial \sqrt{1+\trho_\beta}}{\partial\omega} = -\frac{\beta}2 \trho_\beta
\sqrt{1+\trho_\beta}.
\]
In the infrared regime this can be dealt with easily since  
$\beta \trho_\beta \leq 1/|\omega|$ and the extra inverse power of $\omega$
can be absorbed into $\tcoup_\infty$. Recall that we assume $n\geq 1$.
For the ultraviolet regime we get exponential decay in $\beta$ from
$\trho_\beta(1)$ and we are done.
\end{proof}

 We remark that a similar bound holds for
 $\ta_\udelta^2(\tcoup_\beta-\tcoup_\infty)$ under an $\LGCond{2}$
 condition but we do not need this. Recall that $\udelta'$ was chosen such that 
 Corollary~\ref{Cor-FromPointMEToUnifME} holds true.

\begin{theorem}\label{Thm-ME-LT} Suppose \LGCond{1}.
 Let $\epsilon>0$ and $\udelta\in\Delta(\udelta')$ be given. There exist $\beta_0>0$, $\kappa>0$ and $C>0$,
 such that for all $E\in \RR$ and $\beta\geq\beta_0$:
\[
L_{\beta,\udelta}'\geq (1-\epsilon)\one_{\cH^\romL} - C\bigl(\one[|L_\beta-E|\geq \kappa] + P\otimes P^\conj\bigr),
\]
in the sense of forms on $\cD(N^\romL)$.
\end{theorem}

\begin{proof} From \eqref{TransOfComm}, Proposition~\ref{LUnifBoundTempZero}, applied with $\epsilon/4$ instead of $\epsilon$, and Lemma~\ref{Lemma-IntBoundL} we get as a form bound on $\cD(\tN)$
\begin{align*}
\tL_{\beta,\udelta}' & = \tL_{\infty,\udelta}' -\phi\bigl(\ri \ta_\udelta(\tcoup_\beta-\tcoup_\infty)\bigr)\\
& \geq \bigl(1-\frac{\epsilon}4-\frac{C_1}{\sigma\beta} \bigr)\one_{\tcH^\romL}
 - \sigma \tN- C\bigl(\one[|\tL_\infty-E|\geq\kappa] + P_\infty\bigr),
\end{align*}
valid for all $\sigma>0$.
Here $P_\infty = \cU(P\otimes P^\conj)\cU^*$ is the projection onto the eigenstates of $\tL_\infty$.

Pick a non-negative $f\in C_0^\infty(\RR)$ with $\supp(f)\subseteq [-\kappa,\kappa]$ and
$f=1$ on the interval $[-\kappa/2,\kappa/2]$. Let $\tf$ be an almost analytic extension of $f$. Write
\[
f(\tL_\infty-E) - f(\tL_\beta-E) = \frac1{\pi}\int_\CC \bar{\partial}\tilde{f}(\eta)
\bigl((\tL_\infty-\eta)^{-1}-(\tL_\beta-\eta)^{-1}\bigr) \,\D\eta.
\]
Since $\tN$ is of class $C^1(\tL_\beta)$ with $[\tN,\tL_\beta]^\circ$
being $\sqrt{\tN}$-bounded, cf.~Corollary~\ref{Cor-BasicLReg},
we conclude from Proposition~\ref{Prop-DomInv}~\ref{Item-ResInv-MoPlus}  (and interpolation)
that $(\tL_\beta-\eta)^{-1}$ preserves $\cD(\sqrt{\tN})$ and 
that there exists $n$ and $C$ such that
\[
\bigl\|(\tN + 1)^{\frac12}(\tL_\beta -\eta)^{-1}(\tN+1)^{-\frac12}\bigr\|\leq 
C\bigl(1 + |\im\eta|^{-n}\bigr).
\]
It follows that
\begin{align*}
&\bigl\|\bigl(f(\tL_\infty-E) - f(\tL_\beta-E)\bigr)(\tN + 1)^{-\frac12}\bigr\|\\
&\quad\leq \frac1{\pi}\int_\CC | \bar{\partial}\tilde{f}(\eta)|\,|\im\eta|^{-1}\bigl\|\phi\bigl(\tcoup_\beta -\tcoup_\infty\bigr)(\tL_\beta -\eta)^{-1}(\tN+1)^{-\frac12}\bigr\|\, \D\eta\\
&\quad\leq C \bigl\|\phi\bigl(\tcoup_\beta -\tcoup_\infty\bigr)(\tN+1)^{-\frac12}\bigr\|.
\end{align*}
Appealing to Lemma~\ref{Lemma-IntBoundL}, we thus get
\begin{align*}
\one[|\tL_\infty-E|\geq\kappa] & \leq
\one_{\tcH^\romL}-f(\tL_\infty-E)\\
& \leq \one_{\tcH^\romL} - f(\tL_\beta-E) + \sigma \tN +
\frac{C_2}{\sigma\beta}\\
& \leq \one[|\tL_\beta-E|\geq\kappa/2] + \sigma \tN +
\frac{C_2}{\sigma\beta}.
\end{align*}
Choose first $\sigma>0$ small enough such that
\begin{equation}\label{LastChoiceofSigma}
3\sigma\sup_{\beta\geq 1,\udelta\in\Delta_0}\big\|\ta_\udelta \tcoup_\beta\big\| < \frac{\epsilon}{4} \quad \textup{and} \quad \frac{1-3\epsilon/4}{1+3\sigma}>1-\epsilon,
\end{equation}
 and subsequently $\beta_0\geq 1$ large enough such that
\[
\frac{C_1}{\sigma\beta_0 }+\frac{C_2}{\sigma\beta_0}<\frac{\epsilon}4.
\]
With these choices we arrive at the bound
\[
\tL_{\beta,\udelta}' \geq \bigl(1-\frac{\epsilon}{2}\bigr)\one_{\tcH^\romL} - 2\sigma \tN- C\bigl(\one[|\tL_\beta-E|\geq \kappa/2] + P_\infty\bigr).
\]
We conclude the proof by the usual argument, i.e. bounding $(1+3\sigma)\tL_{\beta,\udelta}'$ from below, cf.
 \eqref{LastChoiceofSigma} and the previous proof.
\end{proof}

We conclude, repeating the proof of Corollary~\ref{Cor-Finite-pp-H},

\begin{corollary} Suppose \LGCond{2}. There exists $\beta_0>0$ such that for all $\beta\geq \beta_0$,
the Liouvillean $L_\beta$ has finitely many eigenvalues, all of finite multiplicity.
\end{corollary}

We remark that in a $(\beta,\coup)$-regime where a positive commutator estimate holds, we can under the \LGCond{2} condition 
conclude that eigenstates $\psi$ of the standard Liouvillean $L_\beta$ satisfy that 
$\psi\in\cD(N^\romL)$. This is a consequence of \cite{FaupinMoellerSkibsted2011a} and improves the basic number bound Theorem~\ref{NumberBoundL}, without imposing further conditions on $\coup$. 

\subsection{Open Problems IV}\label{Subsec-OpenIII}

The by far most central open question relevant for this section is whether or not one can establish a 
positive commutator estimate for the standard Liouvillean for arbitrary inverse temperature $\beta$
and coupling $\coup$. We have an unsubstantiated inkling that it should be possible to use $\tA_\udelta$.

\begin{problem}\label{Prob-AllTempME} Establish, for arbitrary $\beta$
and $\coup$, a positive commutator estimate for the \JakPil glued standard Liouvillean $\tL_\beta$, possibly making use of
the conjugate operator $\tA_\udelta$. It would be natural to work under the assumption \LGCond{2}, and indeed we expect that this assumption should suffice.
\end{problem}

We remark that we have not in this section made use of the modular conjugation $J$, cf. \eqref{ModConj}, which takes $L_\beta$ to $-L_\beta$.
This may be an extra ingredient to make use of.

When establishing positive commutator estimates in this section, either at weak coupling, high energy or low temperature, we did not attempt to determine a joint $(\beta,G,E)$-regime in which one can get a positive commutator. Without a positive answer to Problem~\ref{Prob-AllTempME}, investigating the interplay between the different approaches 
above would be natural. 

\begin{problem} Determine a joint $(\beta,\coup,E)$-regime where one can derive a positive commutator estimate.
\end{problem}

While one can establish positive commutator estimates for the Hamiltonian also
for infinite dimensional small systems, cf. \cite{GeorgescuGerardMoeller2004b}, the situation is fundamentally
different for standard Liouvilleans. To see this, consider as the small system a one-dimensional Harmonic Oscillator.
Here the uncoupled Liouvillean $L_0$ will have point spectrum (a multiple of) $\ZZ$, with each eigenvalue
having infinite multiplicity. Hence, one should not expect a positive commutator estimate with compact error terms,
barring some mechanism to lift the infinite degeneracy by other means. 
However, in the dipole approximation this model is explicitly solvable
\cite{Arai1981a,Arai1981c} and K{\"o}nenberg in his thesis managed to handle perturbations
of the Harmonic Oscillator potential \cite{Koenenberg2011a}. 
Note that one can construct a small system where the Hamiltonian $K$
has compact resolvent and $L_0$ has point spectrum which is dense in
$\RR$! See also \cite{FroehlichMerkli2004a,FroehlichMerkliSigal2004}, where an atomic small
system is considered, and positive commutator methods are applied in
the weak coupling regime.

\begin{problem}
What can be said about the general structure of the point spectrum of $L_\beta$, without
the assumption of small coupling or a finite dimensional small system. Are positive commutator estimates useful at all?
\end{problem}

We emphasize that all the proofs from
Subsect.~\ref{Subsec-CommBoundsL} make essential use of 
$\cK$ being finite dimensional.

\newpage

\section{Absence of Singular Continuous Spectrum}\label{Sec-LAP}

\newcommand{\SC}{\textup{\textbf{(SC)}}}

The aim of this section is to establish the following two theorems about 
absence of singular continuous spectrum of Pauli-Fierz Systems at zero and positive temperature. 

In this section we impose an \HGCond{2} assumption at zero temperature and an \LGCond{2} assumption at positive temperature.
We will need an extra ultraviolet assumption, which we found inconvenient to include in \HGCond{2} and \LGCond{2}.
It reads
\[
\SC \qquad (|k|+1)\partial_j \coup\in L^2(\RR^3;\Mat_\nu(\CC)),
\]
for $j=1,2,3$.

\begin{theorem}\label{Thm-SCH} Suppose \HGCond{2} and \SC. Then $\sigma_{\sico}(H) = \emptyset$.
\end{theorem}

The above theorem is due to \cite{GeorgescuGerardMoeller2004b}, 
but our reproduction here establishes the theorem under slightly weaker assumptions, 
a consequence of our choice of small system as finite dimensional. 

The following positive temperature analogue, however, is new 
and improves on a small coupling result going back to \cite{DerezinskiJaksic2001,Merkli2001}. 

\begin{theorem}\label{Thm-SCL} Suppose \LGCond{2} and \SC. Then the following holds
\begin{Enumerate}
\item There exists $\beta_0>0 $ such that $\sigma_{\mathrm{sc}}(L_\beta) = \emptyset$, for all $\beta >\beta_0$.
\item For any $\beta$ there exists $\Lambda>0$ such that $\sigma_{\mathrm{sc}}(L_\beta)\subset [-\Lambda, \Lambda]$. 
\end{Enumerate}
\end{theorem}

Together with Theorem~\ref{Thm-Jadczyk}, this reduces return to equilibrium at 
low temperature and arbitrary coupling strength, to establishing that the zero eigenvalue of $L_\beta$ is simple.

As usual, the road we take to establish Theorems~\ref{Thm-SCH} and~\ref{Thm-SCL}. 
passes through a Limiting Absorption Principle (LAP).
While we can employ the LAP from \cite{GeorgescuGerardMoeller2004a} to deal with 
the zero-temperature case, there is no LAP available in the 
literature which can deal with the positive temperature Liouvillean, 
outside the weak coupling regime where \cite{DerezinskiJaksic2001,Merkli2001} apply. 
In the following two subsections, we establish a new LAP, which applies at both zero and positive temperature.

\subsection{A priori Resolvent Estimates}\label{Subsec-AprioriRes}

 In this subsection we work under the following assumptions on two self-adjoint operators $T$ and $T'$ acting on a Hilbert space $\cH$.
 
\begin{enumerate}
\item[\MCond{1}] $T'$ is of class $C^1_\Mo(T)$, with $[T,T']^\circ\in\cB(\cD(|T'|^{1/2});\cH)$.
\item[\MCond{2}] There exist $e,C_M>0$ and $J\subset \RR$, an open interval, 
such that $T'\geq e\one - C_M \one[T\in J]$. 
\end{enumerate}

The condition \MCond{2}, in particular, ensures that $T'$ is
semibounded. We fix in the following a real number $\eta$ such that
$T'+\eta\geq \one$.

Let $J'$ be a compact subinterval of $J$. Pick $\kappa>0$ such that $J'_\kappa = J'+[-\kappa,\kappa]\subset J$.  
From Urysohn's lemma we get an 
$f\in C_0^\infty(\RR;[0,1])$, with $\supp(f)\subset J$ 
and $f(t) = 1$ for $t\in J'_\kappa$. Then, by \MCond{2},
\begin{equation}\label{Op-M}
M:= T' + C_M f^\perp(T)\geq e \one,
\end{equation}
where $f^\perp = 1- f$.
 Note that $\cD(M) = \cD(T')$.
Put $\cM = \cD(M^{1/2}) = \cD((T'+\eta)^{1/2})$ 
and equip $\cM$ with the norm $\|u\|_\cM = \sqrt{\la u, M u\ra}$, with respect to which it is complete.
Note that with this notation, \MCond{1} implies $[T,T']^\circ\in\cB(\cM;\cH)$ and hence, by duality, we have $[T,T']^\circ\in \cB(\cH;\cM^*)$ as well.


\begin{remarks}\label{Remark-T-Tprime} \begin{enumerate}
\item 
Put $\cD = \cD(T)\cap \cD(T')$. For $\epsilon\in \RR$ the operator
$T_\epsilon = T - \ri \epsilon T'$ is a priori defined as an operator on $\cD$. 
By Proposition~\ref{Prop-Skibsted}, $T_\epsilon$ is in fact closed 
for $\epsilon\neq 0$ and  $T_\epsilon^* = T_{-\epsilon}$. 
A fact also exploited in the proof of Proposition~\ref{Prop-BasicLReg}. 
\item\label{Item-Rem-ImprMourre} An application of Proposition~\ref{Prop-DomInv}~\ref{Item-FPS-Invariance}, 
with $A=T$, $S = T'+\eta$, and $\rho=1/2$, shows that
for all $f\in C_0^\infty(\RR)$ we have $f(T)\colon \cD(T')\to\cD(T')$ continuously. 
By interpolation we also find that $f(T)\in\cB(\cM)$.
\end{enumerate}
\end{remarks}

This section is devoted to the study of the resolvent set of $T_\epsilon$,
and to establish bounds on the corresponding resolvents. 
It is precisely operators of the form $T_\epsilon$ that enter into Mourre's differential inequality technique
to establish Limiting Absorption Principles. Our particular 
construction appeared first in \cite{GeorgescuGerardMoeller2004a}, 
and is closely related to constructions from \cite{MoellerSkibsted2004,Skibsted1998}.

\begin{lemma}\label{LAP-Lemma1} There exists $C_1,C_2>0$  such that for any $\epsilon\in\RR$ and $z\in \CC$, with $\re(z)\in J'$, 
we have
\[
\forall u\in\cD:\quad \|\la T\ra f^\perp(T) u\| \leq C_1\|(T_\epsilon-z)u\|+|\epsilon| C_2\|u\|_\cM.
\]
\end{lemma}

\begin{proof}
Compute first for $u\in\cD$
\[
\|(T_\epsilon-z)u\|^2 + \epsilon\la u,\ri [T,T']^\circ u\ra = \|(T-\re(z))u\|^2
+ \|(\epsilon T'+\im(z))u\|^2. 
\] 
Discarding the last term and appealing to \MCond{1} yields, for any $\sigma>0$, the bound
\begin{align*}
\|(T-\re(z))u\|^2 & \leq  \|(T_\epsilon-z)u\|^2 + |\epsilon| C \|u\| \|u\|_\cM \\
& \leq \|(T_\epsilon-z)u\|^2 + \sigma \|u\|^2 + |\epsilon|^2 \frac{C^2}{4\sigma} \|u\|_\cM^2.
\end{align*}
Let $v\in\cD$.
Inserting $u = f^\perp(T) v\in\cD$, cf.~Remark~\ref{Remark-T-Tprime}~\ref{Item-Rem-ImprMourre},  gives the bound
\[
(C_\kappa^2-\sigma)\|\la T\ra f^\perp(T) v\|^2\leq  \|(T_\epsilon-z)f^\perp(T) v\|^2
+|\epsilon|^2 \frac{C'}{\sigma}\|v\|_\cM^2,
\]
where $C_\kappa = \inf_{\lambda\in\RR\backslash J'_\kappa, \mu \in J'}|\lambda-\mu|/\la\lambda\ra>0$.
Here we used that $M^{1/2}f(T)M^{-1/2}$ is bounded, again appealing to Remark~\ref{Remark-T-Tprime}~\ref{Item-Rem-ImprMourre}. From now on we fix $\sigma = C_\kappa^2/2$.
Consider the term
\[
(T_\epsilon-z)f^\perp(T) v = f^\perp(T)(T_\epsilon-z)v - \ri\epsilon [T',f(T)]v.
\]
The last term is bounded by $|\epsilon|\|v\|_\cM$, which follows from \MCond{1} 
by writing $f(T)$ as an integral over resolvents using an almost analytic extension of $f$. 
See Proposition~\ref{Prop-DomInv}.
Here it is important that $[T,T']^\circ\in\cB(\cM;\cH)$. 
This completes the proof by subadditivity of the square root.
\end{proof}

Before continuing, we observe that \MCond{2} implies the following elementary bound, 
which holds for all $\epsilon$ and $z$ with $\epsilon\im(z)\geq 0$:
\begin{equation}\label{LAP-apri1}
\forall u\in\cD:\quad |\epsilon|\|u\|_\cM^2 + |\im(z)|\|u\|^2 \leq |\im\la u,(T_\epsilon-z)u\ra|
+ |\epsilon| C_M \|u\| \|f^\perp(T)u\|.
\end{equation}

\begin{lemma}\label{LAP-LemmaInv} There exists $\epsilon_0>0$ and $C>0$ such that for $\epsilon\in \RR$ and $z\in \CC$,
with $0 < |\epsilon| \leq \epsilon_0$, $\re(z)\in J'$ and $\epsilon\im(z)>0$, we have $z\in\rho(T_\epsilon)$ 
and the corresponding resolvent $R_\epsilon(z) = (T_\epsilon-z)^{-1}$ satisfies
\[
\|R_\epsilon(z)\|\leq C|e \epsilon+\im(z)|^{-1} \qquad \textup{and}\qquad \|M^{\frac12}R_\epsilon(z)\|\leq C \epsilon^{-1}. 
\]
\end{lemma}

\begin{proof} 
We may consider only the case $\epsilon>0$, and hence we have also $\im(z)>0$. 
The other case is similar.
From \eqref{LAP-apri1} we get, for $u\in\cD$,
\[
\epsilon\|u\|_\cM^2 + \im(z)\|u\|^2 \leq 
|\la u, (T_\epsilon-z)u\ra |
+\epsilon C_M \|u\| \|f^\perp(T)u\|.
\]
We estimate one factor in the last term using Lemma~\ref{LAP-Lemma1}, which yields
\[
\|f^\perp(T)u\| \leq C_1\|(T_\epsilon-z)u\| + \epsilon C_2 \|u\|_\cM.
\]
Hence, we get
\[
\epsilon\|u\|_\cM^2 +\im(z)\|u\|^2 \leq (1+\epsilon C_M C_1)\|u\|\|(T_\epsilon-z)u\|
+ \epsilon^2 C_2 C_M \|u\|\|u\|_\cM.
\]
Estimate one factor of $\|u\|_\cM$ from below by $\sqrt{e}\|u\|$, cf.~\eqref{Op-M},  
on the left-hand side and divide through by $\|u\|$. This yields the bound
\[
\epsilon\sqrt{e}\|u\|_\cM +\im(z)\|u\| \leq (1+\epsilon C_M C_1)\|(T_\epsilon-z)u\|
+ \epsilon^2C_2 C_M \|u\|_\cM.
\]
Choosing now $\epsilon\leq \epsilon_0 := \sqrt{e}(2C_2C_M)^{-1}$, we conclude that
\[
\epsilon\sqrt{e}\|u\|_\cM + 2\im(z)\|u\| \leq (2+\sqrt{e}C_1 C_2^{-1})\|(T_\epsilon-z)u\|.
\]
Estimating $\epsilon\sqrt{e}\|u\|_\cM + 2\im(z)\|u\| \geq (e \epsilon  + \im(z))\|u\|$ gives the first bound.
The second bound follows from the estimate $\epsilon\sqrt{e}\|u\|_\cM + 2\im(z)\|u\|\geq \epsilon \sqrt{e} \|M^{1/2}u\|$.
\end{proof}

From now on, we keep the $\epsilon_0$ from Lemma~\ref{LAP-LemmaInv} fixed.

\begin{lemma}\label{LAP-Lemma4} There exists $C>0$ such that for all $\epsilon\in \RR$, $z\in \CC$ and $u\in\cH$,
with $0<|\epsilon|\leq \epsilon_0$, $\re(z)\in J'$ and $\epsilon\im(z)>0$, we have
\[
\|R_\epsilon(z)f^\perp(T)u\|_\cM \leq C|\epsilon|^{-\frac12} \|u\| \quad \textup{and} \quad 
\|f^\perp(T)R_\epsilon(z)u\| \leq C|\epsilon|^{-\frac12}\|u\|_{\cM^*}.
\]
\end{lemma}

\begin{proof} It suffices to establish the bounds with $\epsilon>0$ and hence $\im(z) > 0$.
Furthermore, the first bound implies the second by duality.

Note first that by Lemmata~\ref{LAP-Lemma1} and~\ref{LAP-LemmaInv}, we have for any $v\in\cH$
\begin{equation}\label{LAP4-AprioriUnif1}
\|\la T \ra f^\perp(T) R_\epsilon(z) v\|\leq C_1\|v\| + \epsilon C_2\|R_\epsilon(z) v\|_\cM\leq C' \|v\|.
\end{equation}
In particular, we get for some $C_3>0$, the uniform bounds
\begin{equation}\label{LAP4-AprioriUnif2}
\|f^\perp(T)R_\epsilon(z)\|\leq C_3 \qquad \textup{and}\qquad \|R_\epsilon(z)f^\perp(T)\|\leq C_3,
\end{equation}
with the second bound following from the first, by taking the adjoint.

Estimate for $u\in\cH$, with $\|u\|\leq 1$, using \eqref{Op-M}
\begin{equation}\label{LAP3-InitEq}
\|M R_\epsilon(z)f^\perp(T) u\| \leq \|T' R_\epsilon(z)f^\perp(T)u\|
+C_M \|f^\perp(T) R_\epsilon(z)f^\perp(T) u\|.
\end{equation}
As for the first term on the right-hand side of \eqref{LAP3-InitEq}, we get using Lemmata~\ref{LAP-Lemma1} and~\ref{LAP-LemmaInv}
together with \eqref{LAP4-AprioriUnif2}
\begin{align*}
 \|T'R_\epsilon(z)f^\perp(T)u\| & = \frac1{\epsilon} \|((T_\epsilon-z) - (T-z))R_\epsilon(z)f^\perp(T) u\|\\
& \leq \frac1{\epsilon}  + \frac1{\epsilon}\| (T-z) R_\epsilon(z) f^\perp(T) u\|\\
&\leq  \frac{C'}{\epsilon} + \frac{C''}{\epsilon}\|\la T\ra f^\perp(T) R_\epsilon(z) f^\perp(T) u\|
\end{align*}
Inserting back into \eqref{LAP3-InitEq}, we arrive at
\begin{align*}
\|M  R_\epsilon(z)f^\perp(T)u\| &\leq
 \frac{C'}{\epsilon} +\frac{C''+C_M}{\epsilon}\|\la T\ra f^\perp(T) R_\epsilon(z) f^\perp(T) u\|\leq \frac{C_4}{\epsilon}, 
\end{align*}
where we made use of \eqref{LAP4-AprioriUnif1}.
Hence, combining with \eqref{LAP4-AprioriUnif2},
\[
\|M^\frac12 R_\epsilon(z)f^\perp(T)u\| \leq  \|R_\epsilon(z) f^\perp(T) u\|^{\frac12}
\|M R_\epsilon(z) f^\perp(T) u\|^{\frac12} \leq  \epsilon^{-\frac12}\sqrt{C_3 C_4}.
\]
This concludes the proof
\end{proof}

\begin{lemma}\label{LAP-Lemma5}  There exists $C_1,C_2>0$ such that for all 
$\epsilon\in\RR$ and $z\in \CC$, 
with $0<|\epsilon|\leq \epsilon_0$, $\re(z)\in J'$ and $\epsilon\im(z)>0$, we have
\begin{Enumerate}
\item\label{LAP5-1} $\forall u\in\cH:\quad |\epsilon|^\frac12 \| R_\epsilon(z)u\|_\cM \leq 
2|\la u,R_\epsilon(z)u\ra|^\frac12 +  C_1 \|u\|_{\cM^*}$.
\item\label{LAP5-2} $R_\epsilon(z)$ extends by continuity from an operator on $\cH$
  to an element of $\cB(\cM^*;\cM)$ and the extension satisfies the bound
$\|R_\epsilon(z)\|_{\cB(\cM^*;\cM)}\leq  C_2|\epsilon|^{-1}$.
\end{Enumerate}
\end{lemma}

\begin{proof}
Since, for $u\in\cH$, we have $R_\epsilon(z)u\in\cD$, we get from \eqref{LAP-apri1} and Lemma~\ref{LAP-Lemma4} the estimate
\begin{align*}
\epsilon\|R_\epsilon(z)u\|_\cM^2 &\leq |\im\la u,R_\epsilon(z)u\ra| + \epsilon C_M\|R_\epsilon(z)u\| \|f^\perp(T) R_\epsilon(z)u\|\\
&\leq |\la u,R_\epsilon(z)u\ra| + C\epsilon^\frac12 \|R_\epsilon(z)u\| \|u\|_{\cM^*}\\
& \leq  |\la u,R_\epsilon(z)u\ra| +  \frac{\epsilon e}{2} \|R_\epsilon(z)u\|^2 + \frac{C^2}{2e}  \|u\|^2_{\cM^*}.
\end{align*}
Since $e\|R_\epsilon(z)u\|^2 \leq \|R_\epsilon(z)u\|^2_\cM$,
this implies \ref{LAP5-1}.

To see \ref{LAP5-2}, observe that 
\[
|\la u, R_\epsilon(z) u\ra| \leq \|R_\epsilon(z)u\|_\cM\|u\|_{\cM^*}
\leq \frac{\epsilon}2 \|R_\epsilon(z)u\|_\cM^2 + \frac1{2\epsilon}\|u\|_{\cM^*}^2.
\]
From this estimate and \ref{LAP5-1},  the statement \ref{LAP5-2} follows.
\end{proof}

\begin{proposition}\label{Prop-LAP-ToRes} Let $z\in\CC$, with $\re(z)\in J'$ and $\im(z)\neq 0$.
We have the strong limit on $\cH$
\[
\slim_{\epsilon\to 0, \epsilon\im(z)>0} R_\epsilon(z) = (T-z)^{-1}.
\]
\end{proposition}

\begin{proof} Let $u\in\cD(M)=\cD(T')$ and compute the difference
\begin{align*}
(R_\epsilon(z)- (T-z)^{-1})u & = \ri\epsilon R_\epsilon(z)T' (T-z)^{-1}u\\
&= \ri\epsilon R_\epsilon(z) (T' M^{-1}) (M(T-z)^{-1}M^{-1}) Mu.
\end{align*}
It follows from Proposition~\ref{Prop-DomInv}
and Lemma~\ref{LAP-LemmaInv}, that the right-hand side  goes to zero, when $\epsilon\to 0$ while keeping $\epsilon\im(z)>0$.

 This implies the result, since $\cD(M)$ is dense in $\cH$ and $R_\epsilon(z)$ is uniformly bounded in 
 small $\epsilon$ with $\epsilon\im(z)>0$. See Lemma~\ref{LAP-LemmaInv} again.
\end{proof}

\subsection{Limiting Absorption Principle}

 Before we state and prove the main result of this section, 
the Limiting Absorption Principle, we need to impose conditions on how
a conjugate operator $A$ fits together with $T$ and $T'$.
 
  The operator $A$ should be maximally symmetric, i.e. have at least one deficiency index equal to $0$.
To conform with the example of Pauli-Fierz Hamiltonians, we assume $n_+ = \dim(\ker(A^*-\ri)) =0$,
such that $\set{z\in\CC}{\im(z)<0}\subset \rho(A)$.  We write $W_t$ for the semigroup \myquote{$\e^{\ri t A}$} of isometries generated
by $A$, and $W_t^*$ for the contraction semigroup generated by $-A^*$. A sub Hilbert space $\cG$ of $\cH$
is said to be $b$-preserved by a semigroup $W_t$ if $W_t\cG\subseteq\cG$ and $\sup_{0\leq t\leq 1} \|W_t u\|_\cG <\infty$ for all $u\in\cG$.
 
We add to \MCond{1} and \MCond{2} the following assumptions 
\begin{enumerate}
  \item[\MCond{3}] $(T'+\eta)^{1/2}$ is of class $C^1_\Mo(A)$ and $W_t^*$ $b$-preserves $\cM$. 
  \item[\MCond{4}] There exists $\sigma>0$ such that for $\zeta\in \CC$, with $\im(\zeta)\geq \sigma$, we have
  \[
  \forall u\in\cD:\quad \bigl\la u,\ri[(A + \zeta)^{-1},T]u\bigr\ra = \bigl\la u, (A+\zeta)^{-1} T' (A+\zeta)^{-1}u\bigr\ra.
  \]
\end{enumerate}
  
\begin{remarks}\label{Remark-M3-4} Let us make some observations pertaining to \MCond{3}
  and \MCond{4}:
\begin{enumerate}
\item\label{Item-Remark-M3} The first part of assumption \MCond{3}, together with Lemma~\ref{Lemma1-Prop-Mourre-C1}, ensures the existence of a $\sigma>0$ 
such that $(A+\zeta)^{-1}\colon \cM\to\cM$, provided $\im(\zeta)\geq \sigma$. 
The second part of assumption \MCond{3} ensures that $(A^*+\bar{\zeta})^{-1}\colon \cM\to\cM$, provided $\im(\zeta)\geq \sigma$ with a possibly larger $\sigma$. By duality
$(A+\zeta)^{-1}\colon\cM^*\to\cM^*$ for $\im(\zeta)\geq\sigma$.
Hence both sides of 
the expression in \MCond{4} makes sense for $\im(\zeta)$ larger than
this $\sigma$.
\item It is a consequence of \MCond{3} that $T'$ is of class
  $C^1(A)$ and $T'':= [T',A]^\circ \in \cB(\cM;\cM^*)$. See
  Lemma~\ref{Lemma-Tprime-Is-C1} below for a proof. The converse
  is probably false although we do not have a counter example.
\item\label{Item-StrongerM3} One can verify \MCond{3}, by checking the stronger condition
  that $T'$ is of class $C^1_\Mo(A)$ and 
  $[T',A]^\circ \in \cB(\cD((T'+\eta)^{\rho});\cH)$, for some $\rho\in [0,1)$. See
  Proposition~\ref{Prop-DomInv}~\ref{Item-Salpha-C1Mo}.
\item In the following, $\sigma$ refers to the constant from \MCond{4}.
The particular choice of the condition \MCond{4} is inspired by \cite{FaupinMoellerSkibsted2011a}.
\end{enumerate}
\end{remarks}

We are now in a position to formulate the basic form of the Limiting
Absorption Principle.


\begin{theorem}[Limiting Absorption Principle]\label{Thm-LAP} Assume $T,T'$ and $A$ 
satisfies the assumptions \MCond{1}--\MCond{4}. Then there exists $C>0$, such that for any 
$u\in\cD(A)$, we have
\[
\sup_{\stackrel{z\in\CC}{\re(z)\in J',\im(z)\neq 0}} \bigl|\big\la u,(T-z)^{-1}
u\big\ra\bigr|  
\leq C\bigl(\|u\|_{\cM^*}^2 + \|Au\|_{\cM^*}^2\bigr).
\]
\end{theorem}


\begin{lemma}\label{Lemma-Tprime-Is-C1} Suppose \MCond{3}. Then $T'$ is of class $C^1(A)$ and $[T',A]^\circ$ extends
from $\cD(T')$ to a bounded form on $\cM$, which we identify with a bounded operator
$T'' = [T',A]^\circ\colon\cM\to\cM^*$.
\end{lemma}

\begin{proof} Recall that $\eta\in\RR$ was chosen such that $T'+\eta\geq \one$. For $n\geq 0$, we have 
\[
\bigl(T'+\eta + n^2\bigr)^{-1} = \bigl((T'+\eta)^{\frac12} - \ri n\bigr)^{-1}\bigl((T'+\eta)^{\frac12} + \ri n\bigr)^{-1},
\] 
which preserves $\cD(A)$ by assumption.

We can thus compute for $\lambda>0$, as forms on $\cD(A)$, 
\begin{align*}
\bigl[\bigl(T'+\eta+ n^2\bigr)^{-1},A\bigr] & =  \bigl[\bigl((T'+\eta)^\frac12 - \ri n\bigr)^{-1},A\bigr]
\bigl((T'+\eta)^{\frac12} + \ri n\bigr)^{-1}\\
 &\quad  +\bigl((T'+\eta)^{\frac12} - \ri n\bigr)^{-1} \bigl[\bigl((T'+\eta)^\frac12+\ri n\bigr)^{-1},A\bigr]  \\
& =-\bigl((T'+\eta)^{\frac12} - \ri n\bigr)^{-1}\bigl[(T'+\eta)^{\frac12},A]^\circ\bigl(T'+\eta+ n^2\bigr)^{-1}\\
& \quad -\bigl(T'+\eta+n^2\bigr)^{-1}\bigl[(T'+\eta)^{\frac12},A\bigr]^\circ\bigl((T'+\eta)^{\frac12} + \ri n\bigr)^{-1},
\end{align*}
which extends to a bounded operator by assumption. Hence, $T'$ is of class $C^1(A)$.

Abbreviate as in Lemma~\ref{Lemma-In}, $I_{n^2}(T') = n^2(T'+\eta+ n^2)^{-1}$
and 
\[
T'_{n^2} = (T'+\eta)I_n(T') = n^2\one_\cH - n^4\bigl(T'+\eta+n^2\bigr)^{-1}.
\] 
On $\cD(A)\cap\cD(T')$ we compute
\begin{align*}
 & \bigl[T'_{n^2},A\bigr]\\
 &\quad = n^2\bigl((T'+\eta)^{\frac12} - \ri n\bigr)^{-1}\bigl[(T'+\eta)^{\frac12},A\bigr]^\circ I_{n^2}(T')\\
&\qquad +  I_{n^2}(T')\bigl[(T'+\eta)^{\frac12},A\bigr]^\circ n^2\bigl((T'+\eta)^{\frac12} + \ri n\bigr)^{-1}\\
&\quad = I_{n^2}(T')\Bigl\{\bigl((T'+\eta)^\frac12 +\ri n\bigr)\bigl[(T'+\eta)^{\frac12},A\bigr]^\circ\\
&\qquad +\bigl[(T'+\eta)^{\frac12},A\bigr]^\circ\bigl((T'+\eta)^\frac12 -\ri n\bigr)
 \Bigr\}I_{n^2}(T')\\
 & \quad = I_{n^2}(T')\Bigl\{ (T'+\eta)^\frac12\bigl[(T'+\eta)^{\frac12},A\bigr]^\circ
+\bigl[(T'+\eta)^{\frac12},A\bigr]^\circ(T'+\eta)^\frac12
 \Bigr\}I_{n^2}(T').
\end{align*}
Taking the limit $n\to\infty$, cf. \eqref{Reg-ToId}, results in the identity
\[
\bigl[T',A\bigr] = (T'+\eta)^\frac12\bigl[(T'+\eta)^{\frac12},A\bigr]^\circ
+\bigl[(T'+\eta)^\frac12,A\bigr]^\circ(T'+\eta)^\frac12
\]
in the sense of forms on $\cD(A)\cap\cD(T')$. The result now follows since the right-hand side extends to a bounded form on $\cM$.
\end{proof}

As in the previous subsection, we denote by $\epsilon_0>0$ 
the constant coming from an application of Lemma~\ref{LAP-LemmaInv}.

\begin{lemma}\label{LAP-derivative} For all $\epsilon\in\RR$ and $z\in \CC$, with $0<|\epsilon|\leq \epsilon_0$,
$\re(z)\in J'$ and $\epsilon\im(z) >0$, we have 
\[
\forall u\in\cD(A): \quad \la u, R_\epsilon(z)T' R_\epsilon(z)u\ra =
-\la u, \ri[R_\epsilon(z),A]u\ra + \ri\epsilon\la u,R_\epsilon(z),T'' R_\epsilon(z)u\ra.
\]
\end{lemma}

\begin{proof}
For $n>\sigma$ we abbreviate
$I_n(A) = \ri n(A+\ri n)^{-1}$ and $A_n = A I_n(A)$ as in Lemma~\ref{Lemma-In}.
Observe the identities
\begin{equation}\label{Eq-For-A-l}
A_n = \ri n \one + n^2(A+\ri n)^{-1} \quad \textup{and} \quad I_n(A)^* = I_{-n}(A^*). 
\end{equation}

For $u,v\in\cD$ and $n>\sigma$ we compute using \eqref{Eq-For-A-l}, \MCond{4} and Remark~\ref{Remark-M3-4}~\ref{Item-Remark-M3}
\begin{align*}
 \la I_{-n}(A^*)u,T' I_{n}(A)v\ra
&= -\la u,n^2 (A+\ri n)^{-1} T' (A+\ri n)^{-1}v\ra \\
& = \la u, \ri[T,n^2(A+\ri n)^{-1}] v\ra \\
 & =  \la u, \ri[T_\epsilon,n^2(A+\ri n)^{-1}] v\ra  - \epsilon\la u, [T',n^2(A+\ri n)^{-1}]v\ra\\
 &= \la u, \ri[T_\epsilon,A_n] v\ra  + \epsilon\la u,n(A+\ri n)^{-1}[T',A] n(A+\ri n)^{-1}v\ra\\
 &=  \la u, \ri[T_\epsilon,A_n] v\ra +\ri \la u, I_n(A) T'' I_n(A) u\ra.
\end{align*}
Replacing $u$ by  $R_\epsilon(z)^*u$ and $v$ by $R_\epsilon(z)u$ 
we find the identity
\begin{align*}
& \la u, R_\epsilon(z) I_{n}(A) T' I_{n}(A) R_\epsilon(z)u\ra\\
& \quad  =  -\la u, \ri[R_\epsilon(z),A_n ] u\ra 
  +\ri \epsilon\la u, R_\epsilon(z) I_n(A)T'' I_n(A) R_\epsilon(z) u\ra.
\end{align*}
Taking the limit $n\to\infty$, using \eqref{Reg-ToA} and \eqref{Reg-ToS}, we arrive at the desired
identity. Here we used \MCond{3}, cf. Lemma~\ref{Lemma-Tprime-Is-C1}, to ensure applicability of \eqref{Reg-ToS}.
Note that  we have $\|(A+\ri n)^{-1}u\|_{\cM^*}\leq C/n$ for $n >\sigma$,
see the proof of Lemma~\ref{Lemma1-Prop-Mourre-C1}, which shows that
$R_\epsilon(z) I_n(A) (T'+\eta)^{1/2}$ is uniformly bounded in large $n$. Hence
\[
\slim_{n\to\infty} R_\epsilon(z) I_n(A) (T'+\eta)^{1/2} =  R_\epsilon(z) (T'+\eta)^{1/2}
\]
 and we are done.
\end{proof}

We are now ready to give the

\begin{proof}[Proof of Theorem~\ref{Thm-LAP}] Fix a $z\in\CC$, with $\re(z)\in J'$ and $\im(z)\neq 0$.
 For $u\in\cD(A)$ and $\epsilon\in\RR$, with $\epsilon\im(z)>0$,
 we define
 \[
  F_z(\epsilon) = \la u,R_\epsilon(z) u\ra.
 \]
 Using Lemma~\ref{LAP-derivative} we can compute the $\epsilon$-derivative of $F_z$:
 \[
  \frac{\D F_z}{\D\epsilon}(\epsilon) = \ri \la u, R_\epsilon(z)  T'R_\epsilon(z)u\ra 
  = \la u,[R_\epsilon(z),A] u\ra -\epsilon \la u, R_\epsilon(z) T'' R_\epsilon(z)u\ra.
 \] 
Using that $T''\in \cB(\cM;\cM^*)$, cf.~Lemma~\ref{Lemma-Tprime-Is-C1}, we arrive at the bound
\begin{align*}
 \Big|\frac{\D F_z}{\D\epsilon}(\epsilon)\Big| & \leq \|R_{-\epsilon}(\bz)u\|_\cM \|Au\|_{\cM^*} + \|R_\epsilon(z)u\|_\cM \|Au\|_{\cM^*}\\
 &\quad + \epsilon C \|R_{-\epsilon}(\bz)u\|_\cM\|R_\epsilon(z)u\|_\cM.
\end{align*}
Here we used that $A^*u = Au$.

Before continuing we observe that $\overline{F}_z(\epsilon) = \la R_\epsilon(z)u,u\ra = \la u, R_{-\epsilon}(\bz)u\ra = F_{\bz}(-\epsilon)$.
Hence, their norms are the same.
Appealing to Lemma~\ref{LAP-Lemma5}~\ref{LAP5-1}, we thus get the differential inequality
\[
 \Big|\frac{\D F_z}{\D\epsilon}(\epsilon)\Big| \leq 2\epsilon^{-\frac12}(
 2|F_z(\epsilon)|^\frac12 + C\|u\|_{\cM^*})\|Au\|_{\cM^*} + C |F_z(\epsilon)| + C\|u\|_{\cM^*}^2.
\]
 A few applications of the inequality $ab\leq (a^2 + b^2)/2$ yields
\[
 \Big|\frac{\D F_z}{\D\epsilon}(\epsilon)\Big| \leq \epsilon^{-\frac12}C_1
 |F_z(\epsilon)| + \epsilon^{-\frac12}C_2 (\|u\|_{\cM^*}^2+\|Au\|_{\cM^*}^2),
\]
for some positive constants $C_1$ and $C_2$.
From Gronwall's inequality, and a subsequent application of 
Lemma~\ref{LAP-Lemma5}~\ref{LAP5-2}, we arrive at the bound 
\begin{align*}
 |F_z(\epsilon)| &\leq C_3( |F_z(\epsilon_0)|+\|u\|_{\cM^*}^2+\|Au\|_{\cM^*}^2)\\
 &\leq C_4(\|u\|_{\cM^*}^2+\|Au\|_{\cM^*}^2),
\end{align*}
where the constant $C_4$ does not depend on $z$ and $u\in\cD(A)$. 

Proposition~\ref{Prop-LAP-ToRes} now implies the theorem.
\end{proof}

As a consequence of Theorem~\ref{Thm-LAP} and \cite[Thm.~XIII.19]{ReedSimonIII1979}, 
keeping in mind that the compact subinterval $J'\subset J$ was arbitrary, we finally arrive at:

\begin{corollary} Suppose the triple of operators $T,T'$ and $A$ satisfies
\MCond{1}--\MCond{4}. Then $\sigma_{\mathrm{sc}}(T)\cap J = \emptyset$.
\end{corollary}

\begin{proof}[Proof of Theorems~\ref{Thm-SCH} and~\ref{Thm-SCL}:]
Let $\udelta'\in\Delta_0$ be such that Theorems~\ref{Thm-GGM}, \ref{Thm-ME-L-LargeE} and~\ref{Thm-ME-LT} apply. Fix $\udelta\in \Delta(\udelta')$.
Note that $A_\udelta$ is maximally symmetric with deficiency index $n_+=0$. The operator $\tA_{\udelta}$ is self-adjoint, hence -- in particular -- maximally symmetric with $n_+=0$.

Under the conditions considered, we have already established that the point spectrum
of $H$ and $\tL_\beta$ are finite. Hence, it suffices to show that
$\sigma_{\sico}(H)\backslash\sigma_\pp(H) =\emptyset$ and $\sigma_{\sico}(\tL_\beta)\backslash\sigma_\pp(\tL_\beta) = \emptyset$.
Let $J$ be an open bounded interval with $\bJ\cap \sigma_\pp(H) = \emptyset$, for the Hamiltonian,
and $\bJ\cap \sigma_\pp(\tL_\beta) = \emptyset$, for the Liouvillean.

Theorems~\ref{Thm-SCH} and~\ref{Thm-SCL} follow from Theorem~\ref{Thm-LAP} 
and \cite[Thm.~XIII.19]{ReedSimonIII1979}, 
once we have observed that \HGCond{2} implies that the triple
$H,H'_\udelta$ and $A_\udelta$ satisfies \MCond{1}--\MCond{4}; 
and similarly that \LGCond{2} implies \MCond{1}--\MCond{4} 
for the triple $\tL_\beta,\tL_{\beta',\udelta}$ and $\tA_\udelta$.
Recall that $\tL_\beta$ and $L_\beta$ have the same singular continuous spectra.

We begin by verifying \MCond{1}. 
From Lemma~\ref{Lemma-H-is-C1N}, we know that $H$ is of class $C^1_\Mo(N)$.
Hence $U_t = \e^{\ri t H}$ $b$-preserves $\cD(N) = \cD(H')$. Indeed,
an easy commutation argument yields $\| (N+1)(H+z)^{-1}(N+1)^{-1}\|\leq C/|\im(z)|$. 
From this and an approximation argument using $U_t = \slim_{n\to\infty} (1- (\ri t H)/n)^{-n}$, we conclude
the claim. Since $\cD(H)=\cD(H_0)$ and $\cD(H'_\udelta)=\cD(N)$ we see that $\cC$ is dense in
$\cD(H'_\udelta)\cap\cD(H)$ with respect to the intersection topology. On $\cC$ 
the commutator is
\begin{align*}
 \bigl[H_\udelta',H\bigr] & = -\bigl[\phi(\ri a_\udelta\coup),H\bigr] + \bigl[\D\Gamma(m_\udelta(|k|),\phi(\coup)\bigr]\\
&  = -\bigl[\phi(\ri a_\udelta\coup),K\otimes\one_\cF\bigr] + \ri \phi(|k| a_\udelta\coup) 
 + \re\bigl\la a_\udelta \coup,\coup\bigr\ra +\ri \phi(\ri m_\udelta(|k|)\coup).
\end{align*}
The computation extends by continuity to $\cD(H_\udelta')\cap\cD(H)$ and
the right-hand side is $\sqrt{N}$-bounded. This is where we need the condition \SC.
We conclude from Proposition~\ref{Prop-MourreEquiv} that $H_\udelta'$ is of class $C^1(H)$,
and, together with the computation above, also \MCond{1}.

Corollary~\ref{Cor-BasicLReg} established that $\tN$ is of class $C^1(\tL_\beta)$.
Hence the group $U_t$ generated by $\tL_\beta$ $b$-preserves $\cD(\tN)=\cD(\tL_{\beta,\udelta}')$.
As a form identity on $\tcC^\romL$ we have
\begin{align*}
& \bigl[\tL_{\beta,\udelta}',\tL_\beta\bigr] = -\bigl[\phi(\ri \ta_\udelta \tcoup_\beta),\tL_\beta\bigr] +
\bigl[\D\Gamma(m_\udelta(\omega)),\phi(\tcoup_\beta)\bigr]\\
& \qquad = -\bigl[\phi(\ri \ta_\udelta \tcoup_\beta),L_\romp\otimes\one_{\tcF}\bigr] +\ri \phi(\omega \ta_\udelta\tcoup_\beta)
 + \re\bigl\la \ta_\udelta\tcoup_\beta,\tcoup_\beta\bigr\ra +\ri \phi(\ri m_\udelta(\omega)\tcoup\beta).
\end{align*}
The right-hand side is $\sqrt{\tN}$-bounded, where we as above invoke \SC.
Since $\tcC^\romL$ is dense in $\cD(\tL_{\beta,\udelta}')\cap\cD(\tL_\beta)$, cf. Corollary~\ref{Cor-BasicLReg},
we are in a position to conclude from Proposition~\ref{Prop-MourreEquiv} that $\tL_{\beta,\udelta}'$ is of class $C^1_\Mo(\tL_\beta)$. Together with the computation above, we get \MCond{1} for the Liouvillean.

The Mourre estimate \MCond{2} was establish at zero temperature by
Theorem~\ref{Thm-GGM} and at positive temperature by Theorems~\ref{Thm-ME-L-LargeE} (high energy) and~\ref{Thm-ME-LT} (low temperature). Note that being away from eigenvalues, by the choice of $J$, we can get rid of a compact error by passing to a smaller energy window.

That $\tL_{\beta,\udelta}'$ is of class $C^1(\tA_\udelta)$ was established in Lemma~\ref{LprimeC1}.
The same proof applies to show that $H_\udelta'$ is of class $C^1(A_\udelta)$. Since the commutator in both cases is a field operator, we conclude from Remark~\ref{Remark-M3-4}~\ref{Item-StrongerM3} that the first part of \MCond{3} is satisfied at both zero and positive temperature. The second part of \MCond{3} is automatic for
self-adjoint $A$, that is for the Liouvillean.
For both the Hamiltonian and the Liouvillean, $\cM$ is the domain of the square root of the number operator
and the adjoint of the conjugate operators, $A_\udelta^*$  and $\tA_\udelta^*$, commute
 with the relevant number operator. Hence, the $b$-preservation part of \MCond{3} is automatic.

We are left with \MCond{4}, which was established for the Liouvillean in Lemma~\ref{ComputeFirstComm}.
As for $H$, one can proceed as in the proof of Lemma~\ref{ComputeFirstComm}. Note that
$N$ is of class $C^1_\Mo(H)$, by an argument simpler than the one which established \MCond{1} above.
\end{proof}

We end with the following improvement of Theorem~\ref{Thm-LAP}.
Let $S\geq \one$ be an auxiliary operator satisfying
\begin{enumerate}
  \item[\MCond{5}] $ \cD(M^{1/2})\subset \cD(S)$. 
  \item[\MCond{6}] $S$ is of class $C^1_\Mo(A)$.
\end{enumerate}
Abbreviate $\cD_S(A) :=  S^{-1}\cD(A)\subset \cD(A)$.

\begin{corollary}\label{Cor-Impr-LAP} Suppose \MCond{1}--\MCond{6}.
Let $z\in\CC$, with $\re(z)\in J$ and  $\im(z)\neq 0$.
The form $S (H-z)^{-1}S$, extends by continuity from $\cD_S(A)$
to a bounded form on $\cD(A)$, which we denote by the same expression. Furthermore,
there exists $C>0$ such that
for all $u\in\cD(A)$, we have
\[
 \sup_{\stackrel{z\in\CC}{\re(z)\in
    J',\im(z)\neq 0}} \bigl|\big\la u, S (H-z)^{-1} S u\big\ra\bigr| \leq 
C\bigl(\|u\|^2 + \| Au\|^2\bigr).
\]
\end{corollary}

\begin{proof} From Theorem~\ref{Thm-LAP}, we get for $u\in\cD_S(A)$ the bound
\begin{align*}
\bigl|\bigl\la S u, (H-z)^{-1} S u\bigr\ra\bigr| & \leq 
C\bigl(\| S u\|_{\cM^*}^2 + \| A S u\|_{\cM^*}^2\bigr)\\
&\leq C\bigl(\|M^{-\frac12} S u\|^2 + \|M^{-\frac12} A S u\|^2\bigr).
\end{align*}
Now compute, using  \MCond{5} and \MCond{6},
\[
M^{-\frac12} A  S u = M^{-\frac12} S A u +  (M^{-\frac12}S)(S^{-1}[A,S]^\circ) u.
\]
We thus get the bound
\[
\bigl|\big\la S u, (H-z)^{-1} S u\big\ra\bigr| \leq 
C\bigl(\|u\|^2 + \| A u\|^2\bigr).
\]
Density of $\cD_S(A)$ in $\cD(A)$ follows from Lemma~\ref{Lemma-In}, and concludes the proof.
\end{proof}

In the context of Pauli-Fierz systems, the corollary can be applied with $S$ equal to $(N+\one_\cH)^{1/2}$, 
for the Hamiltonian, and with $S$ equal to $(\tN+\one_{\tcH^\romL})^{1/2}$, for the standard Liouvillean. This version of the LAP, appears to afford some control over the infrared problem, in that decay in the conjugate operator $A$ permits to absorb a power of the number operator into the limiting resolvent. The correspondence between limiting absorption and Kato smoothness, cf.~\cite{ReedSimonIII1979},  may now be used to obtain infrared nontrivial integral propagation estimates \cite[Cor.~2.7]{GeorgescuGerardMoeller2004b}.

\subsection{Open Problems V}

We have in this section established absence of singular continuous 
spectrum for Pauli-Fierz Hamiltonians at zero temperature, imposing only the fairly natural \HGCond{2} condition.  At positive temperatures however, we have only established this fact at large energies or sufficiently low temperature. This is of course due to us only having access to positive commutator estimate in these two regimes. We do expect the result to remain valid also for large temperatures
under the condition \LGCond{2}.

\begin{problem}
 Prove that $\sigma_{\mathrm{sc}}(L_\beta) = \emptyset$, for all
 $\beta$ and $\coup$ satisfying \LGCond{2}. 
 This of course reduces to resolving Problem~\ref{Prob-AllTempME}.
\end{problem}

At a first glance the improved LAP Corollary~\ref{Cor-Impr-LAP},
appears to be extraordinarily useful, in particular seen from the point of view of the infrared problem. Recall that $T'$ and the number operator are comparable objects.
However, the author have not yet met an application where the ability to absorb two half-powers of the number operator into resolvents of the Hamiltonian/Liouvillean
was of any importance. In case control of the number operator was needed, there were always other ways of getting it.

\begin{problem} Identify an application where the full power of Corollary~\ref{Cor-Impr-LAP} is essentially needed.
\end{problem}

 This problem, is not so much a problem as it is a search for an application of Corollary~\ref{Cor-Impr-LAP}, where Theorem~\ref{Thm-LAP} does not suffice.
In the initial phase of the work resulting in the papers \cite{FaupinMoellerSkibsted2011a,FaupinMoellerSkibsted2011b}, the authors thought that Corollary~\ref{Cor-Impr-LAP} would be useful in controlling
the Fermi-Golden rule operator. However, under the assumptions needed to construct the Fermi-Golden rule operator, which is precisely \HGCond{2} and presumably \LGCond{2}, the projection onto the unperturbed eigenspaces can be shown to absorb
a full power of the number operator, rendering Corollary~\ref{Cor-Impr-LAP} 
unnecessary. See the discussion at the end of 
Subsects.~\ref{Subsec-NumberBoundH} and~\ref{Subsec-NumberBoundL}.
As a last comment in this direction, the integral propagation estimates one may derive
from Corollary~\ref{Cor-Impr-LAP} using Kato smoothness type arguments has so far not found any applications in the scattering theory for Pauli-Fierz models.

Finally we would like to mention a problem, not directly related to
Pauli-Fierz Hamiltonians, but more of a mathematical topic. 
The literature is abound with Limiting Absorption Principles, but the one we proved here is new, 
in the sense that it does not follow from an existing theorem.  
It is part of a family of LAP's proved under an umbrella called ``singular Mourre theory'' in \cite{FaupinMoellerSkibsted2011a}, 
characterized by the commutator $T'$ not being controlled in any way by the Hamiltonian $T$. 
This type of LAP goes back to Skibsted \cite{Skibsted1998}, 
with two different extensions in \cite{GeorgescuGerardMoeller2004b,MoellerSkibsted2004}. 
The original LAP of Skibsted is a special case of \cite{GeorgescuGerardMoeller2004b}.
The LAP established here together with those of \cite{GeorgescuGerardMoeller2004b,MoellerSkibsted2004}
form a bouquet of three LAP's none of which implies another.
The distinction pertains to how one deals with the double commutators $[T,T']^\circ$ and $T''$. 
Here, due to the fact that resolvents of the Liouvillean are of no help in bounding errors, 
we are forced to only make use of $T'$ when controlling double commutators. 
When studying e.g. the confined Nelson model \cite{FaupinMoellerSkibsted2011b,GeorgescuGerardMoeller2004b}
or AC-Stark systems with Coulomb pair-potentials \cite{MoellerSkibsted2004}, 
one is on the other hand forced to also make use of the Hamiltonian, 
at least partly, when controlling double commutators. This discussion serves to prepare the ground for the last problem:

\begin{problem} Establish a MOALAP, ``Mother Of All Limiting Absorption Principles", 
which includes the three known singular LAP's as well as the standard LAP
from regular Mourre theory, see e.g. \cite{Gerard2008,Sahbani1997}. 
\end{problem}

We remark that a new LAP was recently established for pairs of operators $H$ and $A$ 
in a Krein space setting \cite{GeorgescuGerardHaefner2013}, as opposed to the usual Hilbert space setting considered here.

\newpage

\begin{appendix}

\section{Second Quantization and Geometric Localization}\label{App-GeomLoc}

 In this appendix we briefly recall second quantization, using the notation of Segal. Most proofs can be located in \cite[Sect.~VIII.10]{ReedSimonI1980} and \cite[Sect.~X.7]{ReedSimonII1975}.
 We furthermore introduce a partition of unity in Fock space due to Derezi\'nski and G\'erard \cite{DerezinskiGerard1999}.

\subsection{Second Quantization}\label{SecondQuant}

To any Hilbert space $\gothh$, we associate a symmetric Fock space $\Gamma(\gothh) = \oplus_{n=0}^\infty \gothh^{\otimes_\roms n}$, where $\gothh^{\otimes_\roms 0} = \CC$
and $\otimes_\roms n$ denotes $n$-fold symmetric tensor product, when $n\geq 1$. With the obvious inner product, $\Gamma(\gothh)$ is naturally a Hilbert space itself. The special vector 
 $\vacuum = (1,0,0,\dotsc)\in\Gamma(\gothh)$ is called the vacuum vector.

We write $\Gamma_\fin(\gothh)$ for the dense subspace of $\Gamma(\gothh)$, consisting of states
$\psi$ of the form $\psi = (\psi_0,\psi_1,\dotsc,\psi_n,0,0,\dotsc)$ for some $n\in\NN$,
 where $\psi_j\in\gothh^{\otimes_\roms n}$, $j=0,\dotsc,n$.  For a  
 subspace $V\subset\gothh$, we write $\Gamma_\fin(V)$ for the subspace of $\Gamma_\fin(\gothh)$,
 where the $\psi_j$'s are elements  of the algebraic symmetric $j$-fold
  tensor power of $V$. If $V$ is dense in $\gothh$, then $\Gamma_\fin(V)$ is dense in $\Gamma(\gothh)$.
 
 Let $\gothh_1,\gothh_2$ be two Hilbert spaces (same scalars) and
 $b\colon \gothh_1\to\gothh_2$ a contraction, i.e., with $\|b\|_{\cB(\gothh_1,\gothh_2)}\leq 1$.
 Then we may lift $b$ to a contraction $\Gamma(b)\colon\Gamma(\gothh_1)\to\Gamma(\gothh_2)$ by setting
 \[
 \Gamma(b) = \bigoplus_{n=0}^\infty \overbrace{b\otimes b\otimes\cdots \otimes b}^{n \ \textup{factors}}.
 \]
 Here $b^{\otimes_\roms 0} = 1$, the identity operator on $\CC$. For a closed or closable operator $h$ with domain $\cD(h)\subset \gothh$, we associate a closable operator
 by setting
 \begin{equation}\label{DGamma}
 \D\Gamma(h) = \bigoplus_{n=1}^\infty \sum_{j=1}^n \one_\gothh\otimes \cdots\otimes \one_\gothh\otimes h\otimes\one_\gothh\otimes \cdots\otimes\one_\gothh,
 \end{equation}
 a priori with domain $\Gamma_\fin(\cD(h))$. We use the same notation $\D\Gamma(h)$ for the closure. In the above formula, $h$ is sitting in the $j$'th slot (out of $n$ total). If $h$ is essentially self-adjoint, then $\D\Gamma(h)$ is essentially self-adjoint on $\Gamma_\fin(\cD(h))$. 
 
 Finally, we may use the same formula to lift a sesquilinear form $q$ on $\gothh\times\gothh$ with 
 domain $\cQ_\roml(q)\times\cQ_\romr(q)$,
 to a sesquilinear form $\D\Gamma(q)$ on $\Gamma(\gothh)\times\Gamma(\gothh)$ with domain $\Gamma_\fin(\cQ_\roml(q))\times\Gamma_\fin(\cQ_\romr(q))$. If $q$ is a 
 semi-bounded quadratic form, so is $\D\Gamma(q)$.

 Suppose $g,h$ are densely defined operators on $\gothh$ with domains $\cD(g)$ and $\cD(h)$, respectively.
 Write $\cD(g^*)$ and $\cD(h^*)$ for the domains of their adjoints. Then, we may read
 the commutator $q=[g,h]$ as a sesquilinear form with $\cQ_\roml(q) = \cD(g^*)\cap\cD(h^*)$
 and $\cQ_\romr(q) = \cD(g)\cap\cD(h)$. With this interpretation of the notation, we have
 \begin{equation}\label{Gamma-Gamma}
  [\D\Gamma(g),\D\Gamma(h)] =  \D\Gamma([g,h])
 \end{equation}
 as an identity between sesquilinear forms. In applications, $g$ and $h$ will typically be such that the commutators may be identified with operators and the above computation becomes an operator identity.
 
 If $b\colon\gothh_1\to\gothh_2$ is a contraction, $h$ is an operator on $\gothh_1$ with domain $\cD(h)$ and $g$ is a densely defined operator on $\gothh_2$ with $\cD(g^*)$ the domain of its adjoint, then we may compute
 \begin{equation}\label{DGamma-Gamma}
 \D\Gamma(g) \Gamma(b) - \Gamma(b)\D\Gamma(h) = \D\Gamma(b,gb-bh)
 \end{equation} 
 as an identity between sesquilinear forms on $\Gamma_\fin(\cD(g^*))\times\Gamma_\fin(\cD(h))$.
 Here
 \begin{equation}\label{DGamma2}
  \D\Gamma(b,q) = \bigoplus_{n=1}^\infty \sum_{j=1}^n b\otimes \cdots\otimes b\otimes  q\otimes b\otimes \cdots\otimes b,
 \end{equation}
 where $q$ sits in the $j$'th slot. Here $q$ may either be an operator from $\gothh_1$ to $\gothh_2$ or a sesquilinear form on $\gothh_2\times\gothh_1$. Note that $\D\Gamma(\one_\gothh,q) = \D\Gamma(q)$.
 
  Of particular interest is the \emph{Number Operator}
  \[
  N = \D\Gamma(\one_\gothh),
  \]
  which may be used to control commutators of the form \eqref{Gamma-Gamma}
  and \eqref{DGamma-Gamma}, provided $[g,h]$ and $gb-bh$, respectively, are bounded forms.
  Indeed, for $0\leq \rho\leq 1$, $\psi\in\cD(N_2)$ and $\varphi\in\cD(N_1)$, 
  \begin{equation}\label{BasicNumberBound}
  |\la \psi,\D\Gamma(b,c)\varphi\ra| \leq \|c\|\|(N_2+\one_{\Gamma(\gothh_2)})^\rho\psi\|
  \|(N_1+\one_{\Gamma(\gothh_1)})^{1-\rho}\varphi\|.
  \end{equation}
  Here $N_i$ is the number operator on $\Gamma(\gothh_i)$, $b\colon\gothh_1\to\gothh_2$ is a contraction and $c$ is a bounded operator from $\gothh_1$ to $\gothh_2$ (or a bounded form on $\gothh_2\times\gothh_1$.)
   
 \subsection{Segal Field Operators}\label{SegalFields}
 
 Let $\gothh$ be a Hilbert space and $f\in\gothh$. We write $a(f)$ for the operator on $\Gamma(\gothh)$ annihilating a state $f$, and $a^*(f)$ for its adjoint, creating a state $f$.
 The annihilation and creation operators $a(f)$ and  $a^*(f)$ are 
 defined a priori on $\Gamma_\fin(\gothh)$ by the prescriptions
 \[
 \begin{aligned}
 a(f) \psi_n & = \sqrt{n} ( \la f \vert\otimes \one_{\gothh^{\otimes_\roms n-1}})\psi_n\in\gothh^{\otimes_\roms n-1}\\
 a^*(f) \psi_n & = \sqrt{n+1} S_n f\otimes \psi_n \in\gothh^{\otimes_\roms n+1}
 \end{aligned}
 \]
 and extension by linearity, where $\psi_n\in \gothh^{\otimes_\roms n}$ and $S_n\colon \gothh^{\otimes n} \to \gothh^{\otimes_\roms n}$ is the orthogonal projection onto the symmetric tensors. If $n=0$, the first line should be read as $a(f)\psi_0=0$, i.e., $a(f)$ annihilates the vacuum sector $\CC\vacuum$.
 
 The annihilation and creation operators are closable and we use the same notation for the closures. We remark that $\cD(\sqrt{N})\subset \cD(a(f))\cap \cD(a^*(f))$. 
 
 For $f\in\gothh$, we may now define Segal field operators by setting
 \[
 \phi(f) = \frac1{\sqrt{2}}\bigl(a(f) + a^*(f)\bigr),
 \]
 a priori as an operator on $\cD(a(f))\cap \cD(a^*(f))$.
 That $\phi(f)$ is essentially self-adjoint on $\Gamma_\fin(\gothh)$ follows from Nelson's Analytic Vector Theorem. We use the same notation for its closure.  Note that $\cD(\sqrt{N})\subset \cD(\phi(f))$.
 
  We have for $f,f'\in\gothh$ and $h$ a densely defined operator on $\gothh$ with $f\in\cD(h)$, the commutation relations
 \begin{equation}\label{phiphi-phidGamma-comm}
 \ri[\phi(f),\phi(f')] = \im\la f,f'\ra \quad \textup{and}\quad  \ri [\D\Gamma(h),\phi(f)] = -\phi(\ri h f).
 \end{equation}
  These identities are a priori form identities, although they make
   sense as operator identities on $\cD(N)$ and $\Gamma_\fin(\cD(h))$, respectively, as well.
 For $b\colon \gothh_1\to\gothh_2$ a contraction, $f_1\in\gothh_1$ and $f_2\in\gothh_2$, we furthermore have the intertwining relations
\begin{equation}\label{Gamma-a-Inter}
 \Gamma(b) a^*(f_1) = a^*(b f_1)\Gamma(b)\quad \textup{and} \quad a(f_2)\Gamma(b) = \Gamma(b)a(b^*f_2),
 \end{equation}
read as either form or operator identities on $\cD(\sqrt{N})$.

 In this paper, $\gothh$ is always a function space. In this context
 it is convenient to make use of annihilation and creation \myquote{operators} 
 $a(k)$ and $a^*(k)$, formally corresponding to $a^*(\delta(\cdot-k))$ and $a(\delta(\cdot-k))$. 
 For the sake of concreteness, we take $\gothh=L^2(\RR^3)$ here.
  
 The annihilation operator $a(k)$ is densely defined  
 with domain $\Gamma_\fin(V)$, where $V = L^2(\RR^3)\cap C(\RR^3)$.
 It is given by $(a(k)\psi_n)(k_1,\cdots,k_n) =\sqrt{n}\psi_n(k,k_1,\cdots,k_{n-1})$
 and $a(k)\vacuum=0$.
 It is however not a closable operator, but we may still define $a^*(k)$ as a form
 on $\Gamma_\fin(V)$. (The domain of $a(k)^*$ equals $\{0\}$.) Normal ordered expressions like
 \[
 a^*(k_1)\cdots a^*(k_n)a(k'_1)\cdots a(k'_m)
 \]
 are therefore meaningful as forms on $\Gamma_\fin(V)$.
 
 If $h$ is an operator of multiplication by a (locally square integrable) Borel function $k\to  h(k)$, 
 then we may write $\int_{\RR^3} h(k)a^*(k)a(k)\, \D k$, a priori defined as a form, and observe that it coincides with the form induced by the operator $\D\Gamma(h)$.
 Note that we may have to shrink $V$ to ensure $h V\subset L^2(\RR^3)$.
 Similarly, we may write for $f\in L^2(\RR^3)$
 \[
 \phi(f) = \frac1{\sqrt{2}}\int_{\RR^3} \bigl(f(k)a^*(k) + \overline{f(k)}a(k) \bigr)\, \D k.
 \]
 
\subsection{Abstract Geometric Partition of Unity}\label{Subsect-GeomLoc}

In this subsection, we introduce the geometric localization due to Derezi\'nski and G\'erard \cite{DerezinskiGerard1999}.
Let $\gothh,\gothh_0,\gothh_\infty$ be Hilbert spaces (same scalars) and let $b_i\colon \gothh \to \gothh_i$, $i=0,\infty$,
be two contractions.
We form a new contraction $b\colon \gothh\to \gothh_0\oplus\gothh_\infty$ by setting 
$bf = (b_0 f,b_\infty f)$. Using Segal's second quantization functor, we lift $b$ to a contraction
\[
\Gamma(b)\colon \Gamma(\gothh)\to \Gamma(\gothh_0\oplus\gothh_\infty).
\]
We have a unitary identification operator $I\colon \Gamma(\gothh_0\oplus\gothh_\infty)\to \Gamma(\gothh_0)\otimes\Gamma(\gothh_\infty)$ defined uniquely by the requirements:
\begin{equation}\label{CanonId}
I a^*((f,g)) = \bigl(a^*(f)\otimes\one_{\Gamma(\gothh_\infty)} + \one_{\Gamma(\gothh_0)}\otimes a^*(g)\bigr)I \quad \textup{and} \quad I\vacuum = \dvacuum,
\end{equation}
where $\dvacuum = \vacuum\otimes\vacuum$. Let $g_0$ and $g_\infty$ be densely defined operators on $\gothh_0$ and $\gothh_\infty$, respectively. We then have the intertwining relation
\begin{equation}\label{CanonIdProp}
 \bigl(\D\Gamma(g_0)\otimes\one_{\Gamma(\gothh_0)} + \one_{\Gamma(\gothh_\infty)}\otimes \D\Gamma(g_\infty)\bigr) I=I \D\Gamma\left( \begin{pmatrix} g_0 & 0 \\ 0 & g_\infty \end{pmatrix}\right),
\end{equation}
as an operator identity on $\Gamma_\fin(\cD(g_0))\oplus \Gamma_\fin(\cD(g_\infty))$.

 We may now define the contraction
\[
\cGamma(b) = I \Gamma(b) \colon \Gamma(\gothh)\to \Gamma(\gothh_0)\otimes\Gamma(\gothh_\infty),
\]
and observe that
\[
\cGamma(b)^* \cGamma(b) =  \Gamma(b)^* I^* I \Gamma(b) = \Gamma(b^*)\Gamma(b) = \Gamma(b^*b). 
\]
Hence, if $b$ is an isometry, so is $\cGamma(b)$. For $b$ to be an isometry, we must require
that 
\begin{equation}\label{b-isometry}
b_0^* b_0 +b_\infty^* b_\infty = \one_\gothh.
\end{equation}
  Similarly,
$\cGamma(b)\cGamma(b^*) = I \Gamma(b b^*) I^*$, such that $\cGamma(b)$ is unitary if $b$ is unitary. For $b$ to be unitary, we must apart from \eqref{b-isometry} require that
\begin{equation}\label{b-unitary}
b_0 b_0^* = \one_{\gothh_0}, \quad  b_\infty b_\infty^* = \one_{\gothh_\infty}\quad \textup{and} \quad b_0 b_\infty^* =0.
\end{equation}

To make use of the isometry property, we need to compute intertwiners of the form
$G\cGamma(b)- \cGamma(b) H$, for suitable operators $G$ and $H$.
We need two cases, both of which can be found in \cite[Lemma 2.16]{DerezinskiGerard1999}.

The first case is when $H = \D\Gamma(h)$ and 
$G = \D\Gamma(g_0) \otimes\one_{\Gamma(\gothh_\infty)} + \one_{\Gamma(\gothh_0)}\otimes\D\Gamma(g_\infty)$.
Here $h$ has domain $\cD(h)$ and $g_0,g_\infty$ are densely defined.
Write $\cD(g_0^*)$ and $\cD(g_\infty^*)$ for the domains of the two adjoints. Then, as a form identity on $\Gamma_\fin(\cD(g_0^*))\otimes\Gamma_\fin(\cD(g_\infty^*))\times\Gamma_\fin(\cD(h))$, we may compute using \eqref{DGamma-Gamma} and~\eqref{CanonIdProp}
\begin{equation}\label{dGammacGammaIntertwine}
\bigl(\D\Gamma(g_0) \otimes\one_{\Gamma(\gothh_\infty)} + \one_{\Gamma(\gothh_0)}\otimes\D\Gamma(g_\infty)\bigr)\cGamma(b) - \cGamma(b) \D\Gamma(h)
= \D\cGamma(b,q),
\end{equation}
where $\D\cGamma(b,q) = I \D\Gamma(b,q)$ and $q = (g_0 b_0 - b_0 h, g_\infty b_\infty - b_\infty h)$, read as a sesquilinear form on $(\cD(g_0^*)\oplus \cD(g_\infty^*))\times \cD(h)$.

The second intertwining relation we need to consider is the case where
$H = \phi(f)$ and $G = \phi(f_0)\otimes \one_{\Gamma(\gothh_\infty)} + \one_{\Gamma(\gothh_0)}\otimes \phi(f_\infty)$. Here $f\in\gothh$ and $f_i\in\gothh_i$, $i=0,\infty$. It follows from \eqref{Gamma-a-Inter} and~\eqref{CanonId} that
\begin{align}\label{phicGammaIntertwine}
\nonumber &\bigl(\phi(f_0)\otimes \one_{\Gamma(\gothh_\infty)} + \one_{\Gamma(\gothh_0)}\otimes \phi(f_\infty)\bigr)\cGamma(b)-
\cGamma(b) \phi(f)\\
& \qquad  = \frac1{\sqrt{2}}\bigl(a^*(f_0-b_0f)\otimes \one_{\Gamma(\gothh_\infty)} + \one_{\Gamma(\gothh_0)}\otimes a^*(f_\infty - b_\infty f)\bigr) \cGamma(b)\\
\nonumber & \qquad \quad 
+ \frac1{\sqrt{2}} \cGamma(b) a(b_0 f_0 + b_\infty f_\infty -f),
\end{align}
which may be read as an operator (or form) identity on $\cD(\sqrt{N})$.

\section{Commutator Calculus}\label{App-CommCalc}

In this appendix we recall the notion of $C^1(A)$ regularity from \cite{GeorgescuGerardMoeller2004a},
cf. also \cite{AmreinMonvelGeorgescu1996}, and develop the theory to the extend that it is needed in the notes. Basic well-known facts are supplied without proof,
which may most conveniently be found in \cite{AmreinMonvelGeorgescu1996}, 
whereas detailed arguments are given for claims that are not commonly used.

\subsection{Bounded Operators of Class $C^1(A)$}

The basic definition is the following. 

\begin{definition}[The $C^1(A)$ class of
 bounded operators]\label{Def-C1A} Let $A$ be a densely defined closed operator on $\cH$, with domain $\cD(A)$, and $B\in\cB(\cH)$ a bounded operator. We say that
  $B\in C^1(A)$ if the commutator form $[B,A]$ defined on $\cD(A)\cap\cD(A^*)$
  extends by continuity to a bounded form on $\cH$. We write $[B,A]^\circ\in\cB(\cH)$ for the bounded operator 
  representing the form.
\end{definition}

We will use the $C^1(A)$ calculus  for maximally symmetric $A$ only, where $\cD(A)\cap\cD(A^*) = \cD(A)$.
Hence, in the following $A$ will always be assumed to be (at least) maximally symmetric, which simplifies some results.
We refer the reader to \cite{GeorgescuGerardMoeller2004a} for the general case.

To conform with the example of Pauli-Fierz Hamiltonians, we will always 
assume $n_+ = \dim(\ker(A^*-\ri)) =0$,
such that $\set{z\in\CC}{\im(z) < 0}\subset \rho(A)$. 
With this choice $A$ generates a $C_0$-semigroup of isometries $W_t$:
\[
\forall\psi\in\cD(A):\quad \frac{\D}{\D t} W_t \psi = \ri A W_t\psi.
\]
Recall that a $C_0$-semigroup is a weakly  -- hence strongly -- continuous
semigroup of bounded operators.

The first lemma establish equivalent criteria for being of class $C^1(A)$.

\begin{lemma}\label{lemma:C1equiv}  Let $B\in\cB(\cH)$. The following are equivalent.
  \begin{Enumerate}
  \item\label{item:C1chardef} $B\in C^1(A)$.
  \item \label{item:C1charextend} $B$ maps $\cD(A)$ into itself and $AB-BA\colon\cD(A)\to\cH$
    extends by continuity to a bounded operator on $\cH$.
  \item \label{item:C1charlim} There exists $C>0$ such that 
    $\|B W_t - W_t B\|\leq Ct$, for $0\leq t\leq 1$.
  \end{Enumerate}
\end{lemma}

The following lemma establishes the rules of the calculus

\begin{lemma}\label{Lemma-Prop-C1A} Let $B,C\in\cB(\cH)$ such that $B,C\in C^1(A)$. 
The following holds 
\begin{Enumerate}
\item $BC\in C^1(A)$ and $[A,BC]^\circ = [A,B]^\circ C + B[A,C]^\circ$.
\item If $B$ is invertible, then $B^{-1}\in C^1(A)$ and $[B^{-1},A]^\circ = -B^{-1}[A,B]^\circ B^{-1}$.
\item If $B$ is self-adjoint, then $\ri[B,A]^\circ$ is self-adjoint.
\item\label{Item-Comm-From-Group} $\slim_{t\to 0_+} t^{-1}(BW_t-W_tB) = \ri[B,A]^\circ$. 
\item The linear operator $\ad_A\colon C^1(A)\to \cB(\cH)$ is closed, when both $C^1(A)$ and $\cB(\cH)$ are given the
weak operator topology. 
\end{Enumerate}
\end{lemma}

\subsection{Self-adjoint Operators of Class $C^1(A)$}

\begin{definition}[The $C^1(A)$ class of
 self-adjoint operators]\label{Def-SA-C1A} Let $A$ be a maximally symmetric operator on $\cH$, with domain $\cD(A)$, and $S$ a self-adjoint operator on $\cH$. We say that
  $S$ is of class $C^1(A)$ if there exists $z\in\rho(S)$, the resolvent set of $S$, such that
  $(S-z)^{-1}\in C^1(A)$.
\end{definition}

 We will be somewhat pedantic and say that $S$ is of class $C^1(A)$, instead of using the notation
$S\in C^1(A)$. We prefer to think of $C^1(A)$ as a subset of $\cB(\cH)$.
We remark that if $S$ happens to be bounded, the two definitions coincide.

\begin{lemma}\label{Lemma-Prop-SA-C1}  Let $S$ be self-adjoint and of class $C^1(A)$. The following holds
\begin{Enumerate}
\item For all $z\in\rho(S)$, we have $(S-z)^{-1}\in C^1(A)$. 
\item\label{Item-Cheap-DomInv}  For all $z\in\rho(S)$, we have $(S-z)^{-1}\colon\cD(A)\to\cD(A)$.
\item\label{Item-DenseInDS} $\cD(S)\cap\cD(A)$ is dense in $\cD(S)$.
\item The commutator form
  $[S,A]$, a priori defined on $\cD(A)\cap\cD(S)$, extends by continuity to a bounded form on
  $\cD(S)$.
\end{Enumerate}
\end{lemma}
 
  We write $[S,A]^\circ$ both for the form on $\cD(S)$ and for the bounded operator in
 $\cB(\cD(S);\cD(S)^*)$ representing the form. With this notation we have the important formula
 \begin{equation}\label{ResolventComm}
 \forall z\in\rho(S):\quad [(S-z)^{-1},A]^\circ = - (S-z)^{-1}[S,A]^\circ(S-z)^{-1}.
 \end{equation}

 The following theorem from \cite{GeorgescuGerardMoeller2004a} is crucial for the study of point spectrum
 and eigenstates using commutator methods. 
 If $A$ is self-adjoint the theorem goes back to \cite[Prop.~7.2.10]{AmreinMonvelGeorgescu1996}.
  
 \begin{theorem}[Virial Theorem]\label{Thm-Virial} Let $A$ be maximally symmetric and $S$ self-adjoint and of class $C^1(A)$.
 For any eigenstate $\psi$ of $S$ we have $\la\psi,[S,A]^\circ\psi\ra = 0$.
 \end{theorem}

\begin{proof} Let $\lambda$ be the eigenvalue associated with $\psi$, i.e. $S\psi=\lambda\psi$.
We compute using \eqref{ResolventComm}
\[
\la\psi,[(S+\ri)^{-1},A]^\circ\psi\ra = -\la(S-\ri)^{-1}\psi,[S,A]^\circ(S+\ri)^{-1}\psi\ra
= -(\lambda+\ri)^{-2}\la\psi,[S,A]^\circ\psi\ra.
\]
The left hand side can be computed using Lemma~\ref{Lemma-Prop-C1A}~\ref{Item-Comm-From-Group}
\[
\la\psi,[(S+\ri)^{-1},A]^\circ\psi\ra = \lim_{t\to 0_+} (\ri t)^{-1}\la\psi,((S+\ri)^{-1}W_t - W_t (S+\ri)^{-1})\psi\ra = 0.
\]
This concludes the proof.
\end{proof}

\subsection{The Mourre Class}\label{App-MourreClass}

 Of particular interest to us is the following class of operators.
 
 \begin{definition}[The $C^1_\Mo(A)$ class of
 self-adjoint operators]\label{Def-Mo-C1A} Let $A$ be a maximally symmetric operator on $\cH$ with domain $\cD(A)$, and $S$ a self-adjoint operator on $\cH$. We say that
  $S$ is of class $C^1_\Mo(A)$ if $S$ is of class $C^1(A)$ and $[S,A]^\circ\in\cB(\cD(S);\cH)$.
\end{definition}

Mourre  \cite{Mourre1981} used a different but equivalent definition of the $C^1_\Mo(A)$ class, 
cf. Proposition~\ref{Prop-MourreEquiv} below.
The following key lemma goes back to Mourre in the self-adjoint case.

\begin{lemma}\label{Lemma1-Prop-Mourre-C1} 
Let $S$ be a self-adjoint operator of class $C^1_\Mo(A)$. 
 There exists $\sigma>0$ such that for $z\in\CC$ with
$\im(z)\geq \sigma$, we have 
\[
(A+z)^{-1}\colon \cD(S)\to\cD(S)\quad\textup{and} \quad \bigl\|S(A+z)^{-1}(S+\ri)^{-1}\bigr\|\leq \frac{C}{\im(z)}.
\]
\end{lemma}

\begin{proof}
Compute in the sense of form on $\cH$ for $z\in\CC$ with $\im(z)>0$ such that $A+z$ is invertible:
\begin{align*}
(S+\ri)^{-1}(A+z)^{-1} &= (A+z)^{-1}(S+\ri)^{-1}  + [(S+\ri)^{-1},(A+z)^{-1}]\\
&=  (A+z)^{-1}(S+\ri)^{-1}+(A+z)^{-1} [A,(S+\ri)^{-1}](A+z)^{-1}\\
& = (A+z)^{-1}(S+\ri)^{-1}\bigl(\one_\cH  + [S,A]^\circ(S+\ri)^{-1}(A+z)^{-1}\bigr).
\end{align*}
Since  $[S,A]^\circ(S+\ri)^{-1}$ is bounded, the operator
$B(z) = \one_\cH  + [S,A]^\circ(S+\ri)^{-1}(A+z)^{-1}$ is invertible, provided $\im(z)>0$ is sufficiently large,
with $\|B(z)^{-1}\|$ bounded uniformly in large $\im(z)$.
For such $z$ we thus get
\[
(A+z)^{-1}(S+\ri)^{-1} = (S+\ri)^{-1}(A+z)^{-1} B(z)^{-1},
\]
which completes the proof of the lemma.
\end{proof}

The following lemma is used to approximate unbounded operators by bounded ones, in order  to facilitate computations and extract rigorous arguments from formal ones.

\begin{lemma}\label{Lemma-In} Let $A$ be a maximally symmetric operator on $\cH$. 
Define for $n\in\NN$ bounded operators:
\[
I_n(A) = \ri n (A + \ri n)^{-1} \quad \textup{and} \quad A_n = AI_n(A).
\]
The following holds
\begin{align}
\label{Reg-ToId}& \forall u\in \cH: & & \lim_{n\to+\infty} I_n(A)u = \lim_{n\to+\infty} I_{-n}(A^*)u = u,\\
\label{Reg-ToA} & \forall u\in \cD(A): & & \lim_{n\to+\infty} A_n u = A u.
\end{align}
Suppose $S$ is self-adjoint and of class $C^1_\Mo(A)$. Then
\begin{align}
\label{Reg-ToA2}  & \forall u\in \cD(A):  & & \lim_{n\to+\infty} AI_n(S)u = Au,\\
\label{Reg-ToS} & \forall u\in \cD(S):  & & \lim_{n\to+\infty} SI_n(A)u = Su.
\end{align}
If $S\geq \one$, \eqref{Reg-ToA2} also holds with $I_n(S) = n(S+n)^{-1}$.
\end{lemma}

\begin{proof} 
 The statement in \eqref{Reg-ToId} that  $\slim_{n\to +\infty} I_n(A) = \one$ is obvious since $I_n(A)$ is uniformly bounded
and the identity holds for $u\in\cD(A)$ by the computation $I_n(A)u = u - (A+\ri n)^{-1} Au$. The same argument applies to $I_{-n}(A^*)$.
 
As for \eqref{Reg-ToA}, it follows from the similar computation
\[
A_n u = \ri n A (A+\ri n)^{-1}u = -A (A+\ri n)^{-1}Au + Au.
\]
That the first term converges to zero again follows from $A (A+\ri n)^{-1}$ being uniformly bounded
in $n$, and converging to zero strongly on $\cD(A)$.

As for \eqref{Reg-ToA2} and \eqref{Reg-ToS}, 
we first remark that $I_n(S)\colon \cD(A)\to\cD(A)$, for all $n\geq 1$, and
there exists $n_0\in\NN$ such that $I_n(A)\colon\cD(S)\to\cD(S)$, for all $n\geq n_0$. 
These properties follow from Lemma~\ref{Lemma-Prop-SA-C1}~\ref{Item-Cheap-DomInv} 
and Lemma~\ref{Lemma1-Prop-Mourre-C1}.

To establish \eqref{Reg-ToA2} we compute for $u\in\cD(A)$ and $n\geq 1$:
\[
AI_n(S)u = \ri n A(S+\ri n)^{-1} u = I_n(S)Au + I_n(S)[S,A]^\circ(S+\ri n)^{-1}u.
\]
The first term converges to $Au$ by \eqref{Reg-ToId}, and the second converges to zero.
Here we again used that $\slim_{n\to\infty}S(S+\ri n)^{-1} = 0$.

Finally, for \eqref{Reg-ToS} we compute for $u\in\cD(S)$ and $n\geq n_0$:
\[
SI_n(A)u = \ri n S(A+\ri n)^{-1} u = I_n(A)Su - I_n(A)[S,A]^\circ(A+\ri n)^{-1}u.
\]
The result now follows from \eqref{Reg-ToId} and the bound in Lemma~\ref{Lemma1-Prop-Mourre-C1}.
\end{proof}

\begin{lemma}\label{Lemma2-Prop-Mourre-C1} 
Let $S$ be a self-adjoint operator of class $C^1_\Mo(A)$. The following holds
\begin{Enumerate}
\item $\cD(S)\cap\cD(A)$ is dense in both $\cD(S)$ and $\cD(A)$.
\item $\ri [S,A]^\circ$ is a symmetric operator on $\cD(S)$.
\end{Enumerate}
\end{lemma}

\begin{proof} To prove the first statement, we let $\psi\in\cD(A)$ and put
$\psi_n = I_n(S) \psi$. From 
Lemma~\ref{Lemma-Prop-SA-C1}~\ref{Item-Cheap-DomInv} and \eqref{Reg-ToA2} it now follows that $\psi_n\in\cD$ and  $\psi_n\to\psi$
in $\cD(A)$. Density in $\cD(S)$ holds true in larger generality, 
cf.~Lemma~\ref{Lemma-Prop-SA-C1}~\ref{Item-DenseInDS}. 
Alternatively one may use \eqref{Reg-ToS}.

The second claim follows from the first, since the form $\ri [S,A]$ is symmetric on 
$\cD(S)\cap\cD(A)$ and $\ri [S,A]^\circ$ is $S$-bounded by assumption.
\end{proof}

To make the connection with the assumptions used by Mourre in \cite{Mourre1981}, we have the following proposition.

\begin{proposition}\label{Prop-MourreEquiv} Let $S$ be self-adjoint. The following are equivalent
\begin{Enumerate}
\item\label{Item-OurC1} $S$ is of class $C^1_\Mo(A)$.
\item\label{Item-MourreC1} The following holds:
 \begin{Enumerate}
  \item\label{Item-bstability} $W_t\colon\cD(S)\to \cD(S)$ and $\sup_{0\leq t\leq 1} \|S W_t \psi\| < \infty$ for all $\psi\in \cD(S)$.\footnote{The property  \ref{Item-bstability} is called $b$-stability in \cite{GeorgescuGerardMoeller2004a}.}
  \item\label{Item-MourreCommBound} $\forall\psi,\varphi\in\cD(S)\cap\cD(A)$: $|\la\psi,[S,A]\varphi\ra|\leq C\|\psi\|\|S\varphi\|$.
 \end{Enumerate}
\end{Enumerate}
\end{proposition}

\begin{proof} To get from \ref{Item-OurC1} to \ref{Item-MourreC1}, 
we first show that \ref{Item-OurC1} implies \ref{Item-bstability}. To see this, we follow an argument from the proof of  
\cite[Prop.~2.34]{GeorgescuGerardMoeller2004a}. 
Using the notation from Lemma~\ref{Lemma-In}, we compute for $n\in\NN$ and $\psi\in\cD(S)$ using Duhamel's formula
\[
S_n W_t \psi = W_t S_n \psi + \int_0^t W_{t-s} \ri [S_n,A]W_s \psi\, \D s.
\]
Using the identity $[S_n,A]\psi = [\ri n \one - (\ri n)^2(S+\ri n)^{-1},A]\psi
= I_n(S)[S,A]^\circ I_n(S)\psi$, which follows from $S$ being of class $C^1(A)$, we can estimate
\[
\bigl\|S_n W_t \psi\bigr\| \leq \bigl\|S \psi\bigr\|
 + \bigl\|[S,A]^\circ (S+\ri)^{-1}\bigr\| \int_0^t \bigl\|(S+\ri) I_n(S) W_s\psi \bigr\|\, \D s.
\]
The factor in front of the integral is finite due to the assumption that $S$ is of class $C^1_\Mo(A)$. Writing $(S+\ri)I_n(S)W_s\psi = S_n W_s\psi + \ri I_n(S)W_s\psi$, the property \ref{Item-bstability} follows from Gronwall's Lemma, followed by an application of the Spectral Theorem and Fatou's Lemma.

The bound \ref{Item-MourreCommBound} on the commutator form $[S,A]$
is a direct consequence of the assumption on $[S,A]^\circ$ coming from the Mourre class $C^1_\Mo(A)$.

As for \ref{Item-MourreC1} implies \ref{Item-OurC1}, we note the representation formula
\[
\forall z\in\CC,\im(z)>0:\quad (A+z)^{-1} =\int_0^\infty \e^{\ri t z} W_t\, \D t.
\] 
The property  \ref{Item-bstability}, together with the uniform boundedness principle, implies 
the existence of a constant $C\geq 1$ such that
$\|(S+\ri) W_t (S+\ri)^{-1}\|\leq C^t$ for all $t\geq 0$. This, together with the representation formula,
implies that for $\im(z) > \sigma = \ln(C)$ we have
\[
(A+z)^{-1}\colon \cD(S)\to\cD(S).
\]
Compute now for $\mu\neq 0$ and $z\in\CC$ with $\im(z)>\sigma$:
\begin{align*}
& (A+z)^{-1}(S+\ri\mu)^{-1}  = (S+\ri\mu)^{-1}(A+z)^{-1} + [(A+z)^{-1},(S+\ri\mu)^{-1}]\\
& \qquad =  (S+\ri\mu)^{-1}(A+z)^{-1}+(S+\ri\mu)^{-1} [S,(A+z)^{-1}](S+\ri\mu)^{-1}\\
& \qquad = (S+\ri\mu)^{-1}(A+z)^{-1}\bigl(\one_\cH  - [S,A](A+z)^{-1}(S+\ri\mu)^{-1}\bigr).
\end{align*}
It now follows, as in the proof of Lemma~\ref{Lemma1-Prop-Mourre-C1},  that for $\mu$ sufficiently large,
we have $(S+\ri\mu)^{-1}\colon \cD(A)\to\cD(A)$ and 
\[
[(S+\ri\mu)^{-1},A] = -(S+\ri\mu)^{-1}[S,A](S+\ri\mu)^{-1},
\]
a priori as a form on $\cD(A)$. Hence, the left-hand side extends by continuity to a bounded form on $\cH$, proving that
$S$ is of class $C^1(A)$. That $[S,A]^\circ\in\cB(\cD(S);\cH)$
follows
directly from \ref{Item-MourreCommBound}. This shows that $S$ is of class $C^1_\Mo(A)$ 
\end{proof}

 We need the following special case of \cite[Thm.~2.25]{GeorgescuGerardMoeller2004a},
which is an extension of a result going back to \cite{Skibsted1998}.

\begin{proposition}\label{Prop-Skibsted} Let $A$  
and $S$ be self-adjoint operators on $\cH$ with $S$ of class $C^1_\Mo(A)$. Then the operators
$T_\pm = S\pm \ri A$ defined on $\cD(A)\cap \cD(S)$ are closed and $T_\pm^* = T_\mp$.
\end{proposition}

\begin{proof} 
Let $T_\pm = S \pm \ri A$ with domain $\cD = \cD(S)\cap\cD(A)$.
We aim to show that $\cD(T_+^*)= \cD$, which implies the result since
$T_-\subset T_+^*$ and the pair $S$ and  $-A$ also satisfies the assumptions of the theorem.

\noindent\emph{Step I:} To get started, we first show that for any $z\in\CC$ with $\im z \neq 0$,
we have $(S-z)^{-1}\cD(T_+^*)\subseteq \cD(A)$. Let $\psi\in\cD(A)$, $\varphi\in\cD(T_+^*)$ and compute
\begin{align*}
\la A\psi, (S-z)^{-1}\varphi\ra & = \la  (S-\bz)^{-1} A\psi,\varphi\ra\\
& =  \la  A (S-\bz)^{-1} \psi,\varphi\ra +  \la (S-\bz)^{-1} [S, A]^\circ (S-\bz)^{-1}\psi,\varphi\ra,
\end{align*}
where we used that $(S-\bz)^{-1}\cD(A)\subseteq \cD(A)$ and
 \eqref{ResolventComm} to compute $[(S-\bz)^{-1},A]$. 
Since  $(S-\bz)^{-1} \psi\in\cD$, we may write $A = \ri(S-T_+)$ to arrive at the
estimate
\begin{align*}
|\la A\psi, (S-z)^{-1}\varphi\ra | & = |\la (S-\bz)^{-1} \psi,T_+^*\varphi\ra|
 + |\la S(S-\bz)^{-1} \psi,\varphi\ra|\\
& \qquad  + | \la (S-\bz)^{-1} [S, A]^\circ (S-\bz)^{-1}\psi,\varphi\ra|.
\end{align*}
That $(S-z)^{-1}\varphi\in\cD(A)$ now follows from Cauchy-Schwarz, since $A$ is self-adjoint.

\noindent\emph{Step II:} We proceed to argue that $\cD(T_+^*)\subseteq \cD(S)$.
Let $\psi\in\cD(T_+^*)$ and compute
\[
\|S_n \psi\|^2 = \la S I_{-n}(S)S_n \psi,\psi\ra.
\]
Since $I_{-n}(S) S_n = -\ri n S(S-\ri n)^{-1} I_n(S) =
-\ri n I_n(S) + \ri n I_{-n}(S) I_n(S)$, we conclude from Step I that
$ I_{-n}(S)S_n\psi\in\cD$. We can therefore write $S = T_+ - \ri A$
and obtain
\[
\|S_n \psi\|^2 = \re \la I_{-n}(S)S_n\psi,T_+^*\psi\ra
- \re \la \ri A  I_{-n}(S)S_n \psi,\psi\ra.
\]
We may insert real parts, since the left-hand side is real.
As for the last term, we compute for $\varphi\in\cD$
\begin{align*}
& \re \la \ri A  I_{-n}(S)S_n \varphi,\varphi\ra  = \frac12 \la \varphi, \ri [A,-\ri n I_n(S) + \ri n I_{-n}(S)I_n(S)]\varphi\ra \\
& \qquad  = -\frac12 \la \varphi,  I_{n}(S)\ri [S,A]^\circ I_{n}(S)\varphi\ra
- \frac12 \la\varphi,  I_{-n}(S) \ri [S,A]^\circ I_{-n}(S)I_n(S)\varphi\ra\\
& \qquad \quad 
+ \frac12 \la\varphi,  I_{-n}(S) I_n(S)\ri [S,A]^\circ I_{n}(S)\varphi\ra.
\end{align*}
Applying Cauchy-Schwarz and the inequality $2ab\leq \sigma a^2+\sigma^{-1}b^2$,
valid for all real $a,b$ and $\sigma>0$, yields the estimate
\[
 |\re \la \ri A  I_{-n}(S)S_n \varphi,\varphi\ra |
 \leq \frac13 \|S_n\varphi\|^2 + C \|\varphi\|^2, 
\]
for some $C>0$. By Step I and \eqref{Reg-ToA2}, this estimate is also valid with $\varphi\in\cD$ replaced by $\psi\in\cD(T_+^*)$. We arrive at the bound
\[
\|S_n \psi\|^2 \leq \frac23 \|S_n\psi\|^2 + \tC \|\psi\|^2,
\]
for some new constant $\tC>0$, which does not depend on $n$. This estimate implies that $\psi\in\cD(S)$ and completes Step II.

\noindent\emph{Step III:} We complete the proof by showing that $\cD(T_+^*)\subseteq \cD(A)$ as well. This follows easily from Step II and the computation
$\la A \varphi, \psi\ra = \la  -\ri T_+\varphi ,\psi\ra + \la \ri S \varphi,\psi\ra$,
valid for all $\varphi\in\cD$ and $\psi\in\cD(T_+^*)$. Recall from  Lemma~\ref{Lemma2-Prop-Mourre-C1} that
$\cD$ is dense in $\cD(A)$. 
\end{proof}

\subsection{Roots of Positive Operators of Class $C^1_\Mo(A)$.}

\begin{lemma}\label{Lemma3-Prop-Mourre-C1} 
Let $S\geq \one_\cH$ be a self-adjoint operator of class $C^1_\Mo(A)$.
 The following holds for
 $0<\alpha<1$:
\begin{Enumerate}
\item\label{Item-Root-C1} The operator $S^\alpha$ is of class $C^1(A)$.
\item \label{Item-InterDomDense}
$\cD(S)\cap \cD(A)$ is dense in $\cD(S^\alpha)\cap\cD(A)$ with respect to the intersection topology.
\end{Enumerate}
\end{lemma}

\begin{proof} The claim \ref{Item-Root-C1} follows from the  
norm convergent integral representation formula
\begin{equation}\label{RootOfS}
S^{-\alpha} = c_\alpha \int_0^\infty t^{-\alpha}(S+t)^{-1}\,\D t,
\end{equation}
where $c_\alpha=\sin(\alpha\pi)/\pi$.

As for \ref{Item-InterDomDense}, let $\psi\in \cD(S^\alpha)\cap\cD(A)$ and put
$\psi_n = I_n(S)\psi\in\cD(S)\cap\cD(A)$ with $I_n(S) = n(S+n)^{-1}$ as in Lemma~\ref{Lemma-In}.
Clearly $S^\alpha \psi_n = I_n(S)S^\alpha\psi \to S^\alpha\psi$ by \eqref{Reg-ToId}.
That $A\psi_n\to A\psi$ follows from \eqref{Reg-ToA2}. This concludes the proof.
\end{proof}

\begin{proposition}\label{Prop-DomInv} Let $S\geq \one_\cH$ be a self-adjoint operator on $\cH$ of class $C^1_\Mo(A)$, and suppose there
exists $0\leq\rho <1$ such that $[S,A]^\circ\in \cB(\cD(S^\rho);\cH)$. The following holds
\begin{Enumerate}
\item\label{Item-Salpha-C1Mo} For $0<\alpha<1$, we have $S^\alpha\in C^1_\Mo(A)$ and for 
$\trho\in [0,1]$ with $\trho > 1-(1-\rho)/\alpha$, we have $[S^\alpha,A]^\circ\in \cB(\cD(S^{\trho\alpha}),\cH)$.
If $\alpha<1-\rho$, the operator $[S^\alpha,A]^\circ$ is in particular bounded.
\item\label{Item-ResInv-MoPlus} Let $n\geq 1/(1-\rho)$ be an integer. There exists $C>0$ such that: 
for any $z\in\CC$, with $\im(z)>0$, we have $(A+z)^{-1}\colon\cD(S)\to\cD(S)$ and 
$\|S(A+z)^{-1}S^{-1}\|\leq C(1+\im(z)^{-n-2})$.   
\item\label{Item-FPS-Invariance} Suppose $A$ is furthermore self-adjoint. Then, 
for any function $f\in C^\infty(\RR)$ satisfying that $\sup_{x\in\RR} |x|^n|\D^n f/\D x^n(x)|<\infty$ for all $n\in \NN_0$, we have
$f(A)\colon\cD(S)\to \cD(S)$ continuously. 
\end{Enumerate}
\end{proposition}

\begin{proof} We begin with \ref{Item-Salpha-C1Mo}.
 First of all, recall that by 
Lemma~\ref{Lemma3-Prop-Mourre-C1}~\ref{Item-Root-C1},
$S^\alpha$ is of class $C^1(A)$ for all $0\leq \alpha\leq 1$.
Observe the representation formula
\[
S^{\alpha} = c_\alpha \int_0^\infty t^\alpha(t^{-1}-(S+t)^{-1})\,\D t,
\]
where the integral converges strongly on $\cD(S)$. Here $c_\alpha=\sin(\alpha\pi)/\pi$ is as in \eqref{RootOfS}.
We can now compute as a form on $\cD(S)\cap\cD(A)$:
\begin{align*}
[A,S^\alpha] & = c_\alpha \int_0^\infty t^\alpha (S+t)^{-1}[S,A]^\circ(S+t)^{-1}\, \D t\\
& = c_\alpha \int_0^\infty t^\alpha (S+t)^{-1}B S^\rho(S+t)^{-1}\, \D t,
\end{align*}
where $B = [S,A]^\circ S^{-\rho}$ is bounded by assumption. Hence,
for $\varphi,\psi\in\cD(S)\cap\cD(A)$  and $\trho\in (0,1)$ with $\trho\alpha\leq \rho$, we have:
\begin{align*}
|\la\varphi,[A,S^\alpha]^\circ\psi\ra| & \leq c_\alpha\int_0^\infty 
t^\alpha |\la (S+t)^{-1}\varphi, B (S+t)^{-1} S^\rho\psi\ra|\,\D t\\
& \leq  c_\alpha\|B\|\int_0^\infty 
t^{\alpha-1} \|(S+t)^{-1} S^{\rho-\trho\alpha}\|\,\D t\,\|\varphi\|\|S^{\trho\alpha}\psi\|\\
&\leq  c_\alpha\|B\|\int_0^\infty 
t^{\alpha-2 + \rho-\trho\alpha}\,\D t\,\|\varphi\|\|S^{\trho\alpha}\psi\|.
\end{align*}
To get something finite, $\trho$ has to be chosen such that
$\alpha(1-\trho) +\rho -2 < - 1$.
Hence, we must take $\trho > 1- (1-\rho)/\alpha$. For such $\trho$ the estimate above 
now extends by continuity first to $\varphi,\psi\in\cD(S)$, 
cf.~Lemma~\ref{Lemma-Prop-SA-C1}~\ref{Item-DenseInDS}, and subsequently to
$\varphi\in\cH$ and $\psi\in\cD(S^{\trho\alpha})$.
This completes the proof of \ref{Item-Salpha-C1Mo}.

We proceed to establish \ref{Item-ResInv-MoPlus}.
Let $n\in \NN$ be the smallest integer with $n \geq 1/(1-\rho)$. Put $\kappa = (1+1/(1-\rho))^{-1}$. Note that $0\leq 1-\kappa n\leq \kappa$.
Put $\rho_j = 1 - \kappa j$, for $j=0,1,\dotsc,n$. Since $\kappa < 1-\rho$, we find that
$\rho_{j+1} > \rho_j(1-(1-\rho)/\rho_j) = \rho_j  - 1 + \rho$, for $j=0,1,\dotsc,n-1$. Hence, appealing to \ref{Item-Salpha-C1Mo}, we find that
\[
\forall 0\leq j\leq n-1:\ [S^{\rho_j},A]^\circ S^{-\rho_{j+1}} \ \textup{and} \ [S^{\rho_n},A]^\circ \
\textup{are bounded.} 
\]

Let $\sigma =\max_{j=0,\dotsc,n}\sigma_j$, where $\sigma_j$ comes from 
Lemma~\ref{Lemma1-Prop-Mourre-C1} applied with
the self-adjoint operator $S^{\rho_{j}}$.
Put $B_j = [S^{\rho_j},A]^\circ S^{-\rho_{j+1}}$, for $j = 0,1,\dotsc, n-1$ and set 
$B_n = [S^{\rho_n},A]^\circ$.

We can now compute for $z\in\CC$ with $\im(z)>\sigma$ as an operator identity on 
$\cD(S)$
\begin{align*}
S(A+z)^{-1} & - (A+z)^{-1} S  = -(A+z)^{-1}B_0 S^{\rho_1}(A+z)^{-1} \\
 & =  \sum_{j=0}^{n-1} (-1)^{j+1} \bigl((A+z)^{-1}B_0\bigr)\cdots \bigl((A+z)^{-1}B_j\bigr)(A+z)^{-1} S^{\rho_{j+1}}
\\
&\quad  + (-1)^{n+1} \bigl((A+z)^{-1}B_0\bigr)\cdots \bigl((A+z)^{-1}B_n\bigr)(A+z)^{-1} .
\end{align*}
The right-hand side extends analytically to $z\in\CC$ with $\im(z)>0$ and hence,
for all such $z$ we have proved that $(A+z)^{-1}\colon \cD(S)\to\cD(S)$ and 
\begin{equation}\label{Eq-ResBound-n}
\|S(A+z)^{-1}S^{-1}\|\leq C(1+\im(z)^{-n-2}),
\end{equation}
where $C>0$ does not depend on $z$. This proves \ref{Item-ResInv-MoPlus}.

 If $A$ is self-adjoint, the conclusion extends to $z$ with
$\im(z)<0$ provided an absolute value is inserted on the right-hand side of \eqref{Eq-ResBound-n}.
The last claim \ref{Item-FPS-Invariance} now follows by an almost analytic extension argument,
just as in the last step of the proof of \cite[Lemma~3.3]{FaupinMoellerSkibsted2011a}. We skip the details here
and refer the reader to \cite{Moeller2000} for almost analytic extension.
\end{proof}

\end{appendix}

\newpage

\providecommand{\bysame}{\leavevmode\hbox to3em{\hrulefill}\thinspace}
\providecommand{\MR}{\relax\ifhmode\unskip\space\fi MR }
\providecommand{\MRhref}[2]{%
  \href{http://www.ams.org/mathscinet-getitem?mr=#1}{#2}
}
\providecommand{\href}[2]{#2}

\end{document}